\documentclass{iopart}

\usepackage{tikz}
\usepackage{graphicx,float}
\usepackage{subcaption}
\usepackage{placeins}
\usepackage{hyperref}
\usepackage{amsthm,amssymb}
\expandafter\let\csname equation*\endcsname\undefined
\expandafter\let\csname endequation*\endcsname\undefined
\usepackage{amsmath,amsfonts,mathtools}
\usepackage{bm}
\usepackage{bbm}
\usepackage{physics}
\usepackage{enumitem}
\usepackage{mathrsfs}
\usepackage{dsfont}
\usepackage[most]{tcolorbox}
\usepackage[sort&compress,numbers]{natbib}

\usepackage{hyperref}
\definecolor{darkred}  {rgb}{0.5,0,0}
\definecolor{darkblue} {rgb}{0,0,0.5}
\definecolor{darkgreen}{rgb}{0,0.5,0}
\definecolor{cool_green}{rgb}{0.0, 0.5, 0.0}

\hypersetup{
  colorlinks = true,
  urlcolor  = blue,         
  linkcolor = darkblue,     
  citecolor = darkgreen,    
  filecolor = darkred       
}
\theoremstyle{definition}

\newtheorem{corollary}{Corollary}
\newtheorem{definition}{Definition}

\newtheorem{lemma}{Lemma}
\newtheorem{proposition}{Proposition}
\newtheorem{theorem}{Theorem}

\newtheorem{example}{Example}

\newtheorem*{theorem*}{Theorem}

\newcommand{\ve}{\varepsilon}
\newcommand{\LOSQC}{\mathrm{LOSQC}}
\renewcommand{\BB}{\mathrm{BB84}}
\newcommand{\linspan}{\mathrm{span}}

\newtheorem*{remark}{Remark}

\newcommand{\mbf}{\mathbf}
\newcommand{\mbb}{\mathbb}
\newcommand{\mbbm}{\mathbbm}
\newcommand{\mc}{\mathcal}

\newcommand{\wt}{\widetilde}

\newcommand{\ol}{\overline}
\usepackage{multirow}

\newcommand{\Trans}{\mathrm{T}}
\newcommand{\Pos}{\mathrm{Pos}}

\newcommand{\id}{\textrm{id}}

\newcommand{\cE}{\mathcal{E}}

\newcommand{\cI}{\mathcal{I}}

\newcommand{\cK}{\mathcal{K}}

\newcommand{\cM}{\mathcal{M}}

\newcommand{\cS}{\mathcal{S}}

\newcommand{\cX}{\mathcal{X}}

\newcommand{\cZ}{\mathcal{Z}}

\newcommand{\TD}{D_{\text{tr}}}

\begin{document}
\title{Orthogonality Broadcasting and Quantum Position Verification}

\author{Ian George$^1$, Rene Allerstorfer$^2$, Philip Verduyn Lunel$^3$, and Eric Chitambar$^4$}
\address{$^1$Department of Electrical and Computer Engineering, Centre for Quantum Technologies, National University of Singapore, Singapore}
\address{$^2$QuSoft, CWI Amsterdam, The Netherlands}
\address{$^3$Sorbonne Universit\'{e}, CNRS, LIP6, France}
\address{$^4$Department of Electrical and Computer Engineering,
University of Illinois at Urbana-Champaign}
\ead{qit.george@gmail.com} 

\begin{abstract}
The no-cloning theorem leads to information-theoretic security in various quantum cryptographic protocols. However, this security typically derives from a possibly weaker property that classical information encoded in certain quantum states cannot be broadcast.  To formally capture this property, we introduce the study of ``orthogonality broadcasting."  When attempting to broadcast the orthogonality of two different qubit bases, we establish that the power of classical and quantum communication is equivalent. However, quantum communication is shown to be strictly more powerful for broadcasting orthogonality in higher dimensions. We then relate orthogonality broadcasting to quantum position verification and provide a new method for establishing error bounds in the no pre-shared entanglement model that can address protocols previous methods could not. Our key technical contribution is an uncertainty relation that uses the geometric relation of the states that undergo broadcasting rather than the non-commutative aspect of the final measurements.
\end{abstract}

\maketitle 

\section{Introduction}

The famous no-cloning principle in quantum mechanics prohibits the copying of non-orthogonal quantum states \cite{Wootters-1982a}.  More precisely, if $\ket{\psi}$ and $\ket{\varphi}$ are two non-orthogonal states of some quantum system, then there does not exist any physically realizable mapping $U$ such that $U\ket{\psi}=\ket{\psi}^{\otimes 2}$ and $U\ket{\varphi}=\ket{\varphi}^{\otimes 2}$.  This fundamental fact provides a foundation for designing cryptographic protocols \cite{Broadbent-2020a}, with specific examples being quantum key distribution \cite{Bennett-2014a}, secret sharing \cite{Cleve-1999a}, and position verification \cite{Kent-2011a}.  

Quantum broadcasting extends cloning to a communication setting and relaxes the goal to 
duplication on the level of reduced density matrices.  That is, a set of states $\{\rho_i\}_i$ on system $S$ is broadcast to Alice ($A$) and Bob ($B$) if there exists a completely-positive trace-preserving (CPTP) map $\mc{E}^{S\to AB}$ such that $\mc{E}(\rho_i)=\sigma_i^{AB}$ with $\tr_A(\sigma_i^{AB})=\rho_i$ and $\tr_B(\sigma^{AB})=\rho_i$ for all $i$.  
Similar to the no-cloning theorem, it has been shown that quantum mechanics \cite{Barnum-1996a, Lindblad-1999a, Kalev-2008a} and even more general theories of nature \cite{Barnum-2007a} prohibit universal broadcasting.  Only if the input states $\{\rho_i\}_i$ are pairwise commuting can such a broadcasting map $\mc{E}^{S\to AB}$ be constructed.

\begin{figure}[b]
\centering
\includegraphics[width=0.5\textwidth]{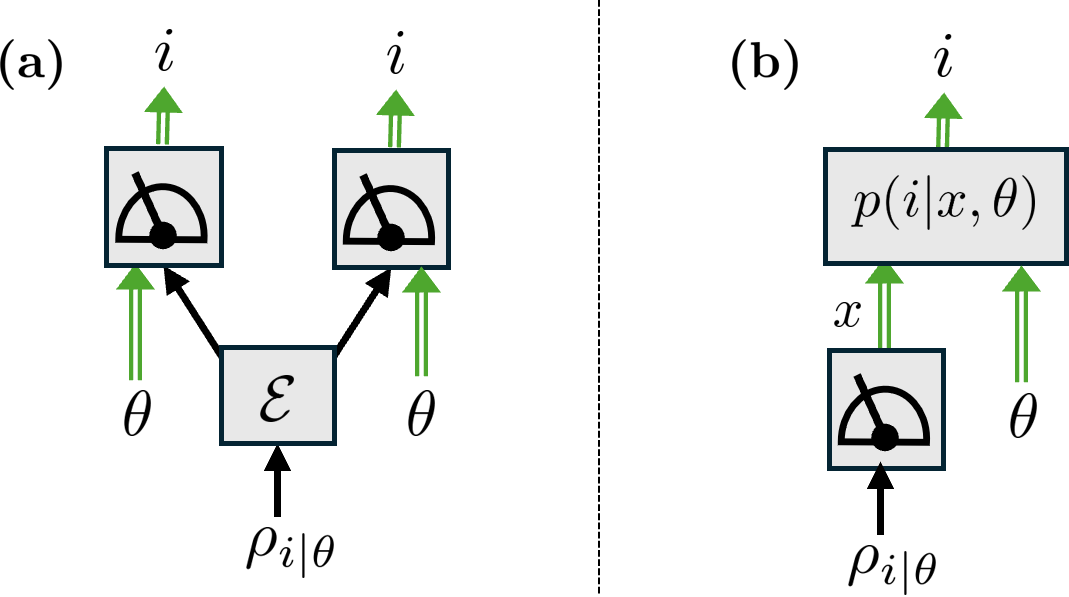}
\caption{\footnotesize (a) Orthogonality broadcasting a set of states $\mc{S}_{\Theta}=\{\rho_{i|\theta}\}_{i,\theta}$ with $\rho_{i|\theta}\perp\rho_{j|\theta}$ for all $\theta\in\Theta$ and $i\not=j$.  A map $\mc{E}$ distributes the orthogonality of $\rho_{i|\theta}$ so that both receivers can recover the index $i$ given value $\theta$.  (b) In a second task, a quantum measurement must be performed on $\rho_{i|\theta}$ \textit{before} receiving $\theta$, which is used only later to recover the index $i$ in classical post-processing.  Orthogonality broadcasting reduces to (b) when one of the broadcast systems is classical. In all figures, green, double-line arrows represent classical information, and black, single-line arrows represent quantum information.
}
\label{Fig:orthogonality_broadcasting}
\end{figure}

In this work we propose a weaker form of broadcasting called \textit{orthogonality broadcasting}.  Let $\mc{S}_{\Theta}=\{\rho_{i|\theta}\}_{i\in\mc{I}_\theta,\theta\in\Theta}$ be $\vert \Theta \vert$ sets of pairwise orthogonal states $\{\rho_{i}\}_{i \in \cI_{\theta}}$, i.e. for all $\theta \in \Theta$ and $i \neq j \in \cI_{\theta}$, $\rho_{i|\theta}\perp \rho_{j|\theta}$.  We say that a CPTP map $\mc{E}^{S\to AB}$ broadcasts the orthogonality of $\mc{S}_{\Theta}$ if 
\[\tr_A(\sigma^{AB}_{i|\theta})\perp \tr_A(\sigma^{AB}_{j|\theta})\;\;\text{and}\;\;\tr_B(\sigma^{AB}_{i|\theta})\perp \tr_B(\sigma^{AB}_{j|\theta})
\]
for all  $\theta, i\not= j$, where $\sigma_{i|\theta}=\mc{E}(\rho_{i|\theta})$ (see Fig.~\ref{Fig:orthogonality_broadcasting} (a)).  Orthogonality broadcasting can be understood as a basic quantum communication task with post-information.  Suppose that a quantum system is prepared in some unknown state chosen from an ensemble like $\mc{S}$.  For example, this could be one of the BB84 states $\{\ket{0},\ket{1},\ket{+},\ket{-}\}$.  Once the the value $\theta$ is revealed as ``post-information,'' the system state can be determined by measuring in the correct orthonormal basis; e.g in the BB84 case, this basis is either $\{\ket{0},\ket{1}\}$ or $\{\ket{+},\ket{-}\}$.  Orthogonality broadcasting enables both Alice and Bob to learn the state preparation through local measurement and post-information.  This goal is less demanding than state broadcasting since all we require is that the correct orthogonality be preserved by the reduced density matrices of the distributed state $\sigma_{i|\theta}^{AB}$.  Nevertheless, through the map $\mc{E}^{S\to AB}$, orthogonality broadcasting still requires copying some type of information that is encoded in a family of potentially non-orthogonal states.  It can thus be viewed as another communication primitive, in the same spirit as cloning and state broadcasting.

A special type of orthogonality broadcasting arises when one of the output systems is classical.  In this case, the broadcasting map takes the form $\mc{E}^{S\to AX}$, with $X$ being a classical register, and we refer to this as \textit{classical orthogonality broadcasting}.  It is important to note that when describing $\theta$ as post-information in the context of broadcasting, it means that $\theta$ arrives after the application of the broadcasting map.  However, there is another notion of post-information in which $\theta$ arrives after the quantum measurement  \cite{Ballester-2008a, Carmeli-2022a}.  However, as we observe in Proposition \ref{Prop:classical-orthogonal-broadcasting} below, under a natural measure of error, these two notions of post-information are operationally equivalent when restricting to classical orthogonality broadcasting, meaning that the optimal success probability of identifying $i$ in one task is the same as the other.

\subsection*{Uncloneable Cryptography and Quantum Position Verification}
\begin{figure}[H]
\centering
\begin{tikzpicture}
	\draw[->,black,thick] (-0.5,-3) -- (-0.5,0.5) node[above]{$t$};
	\draw[->,black,thick] (-0.5,-3) -- (3.5,-3) node[below]{$x$};

    \draw[darkgreen, double]  (0,0) -- (0.5,-0.5) node[midway, above, sloped]{$k$};
    \draw[black]  (0.5,-0.5) -- (2,-2);
    \draw[black]  (0.5,-0.5) -- (0.5,-2.5);
    \draw[black]  (0.5,-2.5) -- (0,-3) node[midway, above, sloped]{\small $\ket{a_k}$};
    \filldraw [black] (0,-3) circle (2pt) node[below]{$\mathsf{V_A}$};

    \draw[darkgreen, double]  (3,0) -- (2,-1) node[midway, above, sloped]{$k$};
    \draw[black]  (2,-1) -- (0.5,-2.5);
    \draw[black]  (2,-1) -- (2,-2);
    \draw[black]  (2.15,-2.15) -- (3,-3) node[midway, above, sloped]{\small $\ket{b_k}$};
    \filldraw [black] (3,-3) circle (2pt) node[below]{$\mathsf{V_B}$};      

    \filldraw[darkred] (0.35,-0.65) rectangle ++(0.3,0.3);
    \filldraw[darkred] (1.85,-1.15) rectangle ++(0.3,0.3);
    \filldraw[darkred] (0.35,-2.65) rectangle ++(0.3,0.3);
    \filldraw[darkred] (1.85,-2.15) rectangle ++(0.3,0.3);

    \filldraw[darkblue] (1.5,-1.5) circle (3pt);
    \draw[black] (1.5,-1.75) node {$\mbf{x}$};

    \filldraw [black] (0,0) circle (2pt);
    \filldraw [black] (3,0) circle (2pt);
    
    \filldraw [darkred] (0.5,-3) circle (2pt) node[below]{$\mathsf{A}$};
    \filldraw [darkred] (2,-3) circle (2pt) node[below]{$\mathsf{B}$};

    \draw[black] (-0.75,-3) node[] (2pt) {$t_{0}$};
    \draw[black] (-0.75,0) node[] (2pt) {$t_{1}$};
\end{tikzpicture}
\caption{\footnotesize Spacetime diagram of a QPV protocol using GOP set $\mc{S}=\{\ket{a_k}\ket{b_k}\}_k$. With probability $p(k)$, the verifiers $\mathsf{V_A}$, $\mathsf{V_B}$ simultaneously send $\ket{a_k}$ and $\ket{b_k}$ respectively at time $t_0$. An honest prover at spacetime point $\mbf{x}$ could measure jointly on the states and identify $k$.  In contrast, dishonest provers, $\mathsf{A}$ and $\mathsf{B}$, would need to clone the orthogonality of $\ket{a_{k}}$ and $\ket{b_{k}}$ at the two lower boxes in order to return the correct value $k$ at time $t_1$. The adversaries' strategy can be limited by the same tools as studying orthogonality broadcasting (Section \ref{sec:QPV}). The communication setting of the dishonest provers is an example of ``local operations and simultaneous quantum communication" (LOSQC).}
\label{fig:QPV}
\end{figure}
Beyond the foundational importance of understanding the no-cloning theorem in quantum information theory, our primary motivation for studying both classical and non-classical orthogonality broadcasting is quantum cryptography. Most quantum cryptographic protocols derive their security from the no-cloning theorem on some level because the protocol `hides' some classical information in non-commuting quantum states that an adversary cannot perfectly distinguish nor clone. This basic idea has given rise to the study of `uncloneable cryptography' (See \cite{Sattath-2023a} for an introduction). 

A major method for establishing security in uncloneable cryptography has been monogamy-of-entanglement (MoE) games introduced in \cite{Tomamichel-2013a} and used subsequently, e.g.~\cite{Broadbent-2020a,Broadbent-2023a,Ananth-2023a, Goyal-2024a,Poremba-2024a,Escola-2024lossy}. However, MoE games do not directly capture the limitations of cloning quantum states, but rather test for a property of the entangled state that arises when trying to clone non-commuting states. Indeed, concurrent work \cite{Poremba-2024a} defines a specific case of orthogonality broadcasting as a `$1 \to 2$ cloning game' and shows it to be equivalent to a \textit{restricted} variant of the MoE game. This in total means that orthogonality broadcasting can only be captured by a (restricted) MoE game under special conditions. As MoE games and orthogonality broadcasting are generally distinct, they require independent studies. For clarity, \ref{app:ob-cloning-games-and-MoE} includes a detailed discussion on the relationship between orthogonality broadcasting and MoE games.

A more concrete motivating example comes from establishing the security of quantum position verification (QPV) in the no pre-shared entanglement (no-PE) model. The limitations on orthogonality broadcasting underpins the security of many QPV protocols under the no-PE model. In these QPV protocols, at time $t_{0}$, spatially separated verifiers draw a state from a set of globally orthogonal product (GOP) quantum states, $\{\ket{a_{k}}\ket{b_{k}}\}_{k}$, and send it to claimed location of the honest prover (See Fig.~\ref{fig:QPV}). The verifiers `accept' if at time $t_{1}$ they both receive the correct index $k$ of the state sent. The honest prover can always achieve this by applying the projective measurement that discriminates the set $\{\ket{a_{k}}\ket{b_{k}}\}_{k}$ at the honest position. However, \textit{dishonest} provers will not be located at the appropriate spacetime point, but spatially separated around this point. Thus they have to attempt to clone the \textit{orthogonality} of $\ket{a_{k}}$ and $\ket{b_{k}}$.\footnote{We stress orthogonality as they only need to be able to extract the index $k$ rather than hold copies of the states $\ket{a_{k}}$,$\ket{b_{k}}$.} The simplest example of such a QPV protocol is `BB84 QPV' where the set of $\{\ket{a_{k}}\ket{b_{k}}\}$ is given by
\begin{align}
	\{\ket{0}^{A}\ket{0}^{B} \, , \, \ket{0}^{A}\ket{1}^{B} \, , \, \ket{1}^{A}\ket{+}^{B} \, , \, \ket{1}^{A}\ket{-}^{B} \} \ . 
\end{align}
In this case, the adversaries' optimal attack is to approximate orthogonality broadcasting as much as possible. This is because the party who receives system $B$ acts as the broadcasting channel $\cE$ and the party who receives the fully classical state, system $A$, copies it and forwards it, acting as the value $\theta \in \Theta$ arriving after broadcasting (recall Fig.~\ref{Fig:orthogonality_broadcasting}(a)). Note this exact identification holds so long as one of the registers is always classical and the adversaries hold no entanglement, i.e. one considers the no-PE model.

In principle, this above correspondence alone would motivate the study of orthogonality broadcasting. However, in the case of BB84 QPV, it turns out the reduction to an MoE game is tight \cite{Tomamichel-2013a}, which is likely why the limitations of cloning and MoE games have been treated as inextricable while distinct. Nonetheless, it is easy to break this exact connection to MoE games: if the set $\{\ket{a_{k}}^{A}\ket{b_{k}}^{B}\}_{k}$ is not of the form $\{\ket{\theta}^{A}\, U_{\theta}\ket{x}^{B}\}_{x,\theta}$ where $\{U_{\theta}\}_{\theta}$ is some set of unitaries acting on the $B$ space, the reduction to an MoE game breaks down (see \ref{app:ob-cloning-games-and-MoE} for more details). The set of states not satisfying this aforementioned structure can happen if both parties receive a quantum input or simply by not having the $\ket{b_{k}}$ be generated by $x$ and a unitary. Both of these situations are natural to consider in designing more powerful, practical QPV protocols as well as studying the structure of no-cloning.

\subsection*{Summary of Results and Outline}
At a high-level, this work contributes to three domains: the foundations of quantum mechanics, the security of quantum position verification, and, at a technical level, the development of quantitative uncertainty relations. We summarize the contribution to these respective fields below. Any omitted proofs from the main text are collected in the appendices. \\

\noindent \textbf{Contribution 1: Orthogonality Broadcasting and Foundations of Quantum Mechanics} In Section \ref{sec:orthogonality-broadcasting}, we introduce orthogonality broadcasting and analyze both the perfect and approximate setting. For perfect orthogonality broadcasting, we show a separation between the power of using classical and quantum communication (Theorem \ref{Thm:classical-broadcasting-separation}) and further demonstrate that perfect orthogonality broadcasting sometimes requires the use of entangled states (Theorem \ref{Thm:no-sep-comm}). We however show that when one system is a qubit system, there is generally not an advantage to quantum communication (Proposition \ref{Prop:classical-broadcasting-equal-2}). These results also imply a similar property in the communication setting of a QPV protocol without entangled inputs (Theorem \ref{thm:C2-Cd-LOSCC}). This implies a separation between adversaries having access to classical and quantum communication when attacking a QPV protocol without entangled inputs in the no-PE model (Corollary \ref{cor:LOSQC-LOSCC-separation}). This resolves an open question from \cite{Allerstorfer-2022a}. As previously highlighted, orthogonality broadcasting is a relaxation of state broadcasting which is itself a relaxation of cloning. Thus, these results refine our understanding of the structure of quantum mechanics through an information-theoretic lens. \\
	
\noindent \textbf{Contribution 2: The Security of QPV in the no-PE model} In Section \ref{sec:QPV}, we consider QPV protocols where the task of the honest party is to distinguish globally orthogonal product (GOP) states $\{\ket{a_k}\ket{b_k}\}_k$ (see Fig.~\ref{fig:QPV}). We identify the adversaries' communication structure in attacking such QPV protocols in the no pre-shared entanglement (no-PE) model as optimizing strategies that make up `local operations and simultaneous quantum communication' (LOSQC). We show how LOSQC relates to orthogonality broadcasting, either through an exact correspondence (as described previously in the case of BB84 QPV) or via relaxations for more involved GOP states. This allows us to upper bound the success probability of state discrimination on GOP ensembles which, by the exact correspondence between LOSQC and QPV attacks in the no-PE model, implies security bounds on these QPV protocols. 

Our method culminates in three theorems. The first (Theorem \ref{thm:discrimination-cond}) provides bounds for any set of product states that contains four product states that generalize the structure of the BB84 states. This may be applied to many well-known GOP ensembles and recovers the tight bound for BB84 QPV (Corollary \ref{cor:LOSQC-BB84}). The second (Theorem \ref{thm:Overlapping-BB84}) establishes bounds when the set of globally orthogonal product states to discriminate are in qutrit space such that the ``error per state'' outperforms that of BB84 QPV. Theorem \ref{thm:Overlapping-BB84} implies a significant practical advantage in using only slightly higher-dimensional QPV protocols. The third (Theorem \ref{thm:qq-qc-separation}) demonstrates that the adversarial success probability can strictly decrease when requiring \textit{both} adversaries to broadcast locally non-commuting states. While intuitive, to the best of our knowledge, this is the first method that rigorously establishes the added strength of QPV protocols in which both verifiers send quantum states. In particular, in \ref{app:ob-cloning-games-and-MoE}, we show the ensembles considered in Theorems \ref{thm:Overlapping-BB84} and \ref{thm:qq-qc-separation} cannot be analyzed using MoE games with pre-existing methods.  \\

\noindent \textbf{Contribution 3: A New Type of Uncertainty Relation} Contributions 1 and 2 stem from a more direct analysis of no-cloning than previous methods. In order to achieve this more direct analysis, one of the main technical contributions of this work is a new uncertainty relation that only depends on the relation between the given states to be cloned (Lemma \ref{lem:main-text-UR} and Theorem \ref{thm:generalized-uncertainty-relation}), and thus can be used to bound the error in orthogonality broadcasting of more general ensembles than previous methods. In other words, our uncertainty relation differs from entropic uncertainty relations studied in quantum information theory \cite{Deutsch-1983a, Renes-2009a, Berta-2010a, Coles-2017a} as well as the operator inequality used to analyze Monogamy-of-Entanglement games \cite{Tomamichel-2013a}, which all rely on incompatible measurements. To the best of our knowledge, this is the first such uncertainty relation of its type and we hope it motivates more general and direct methods for bounding strategies based on cloning and thus more direct applications to cryptographic protocols whose security depends on the no-cloning theorem.

\section{Orthogonality broadcasting}\label{sec:orthogonality-broadcasting}

Let us analyze the task of orthogonality broadcasting in more detail.  To prepare us for its application in QPV, suppose that Alice receives an ensemble of states $\mc{S}_\Theta = \{\rho_{i|\theta}\}_{i,\theta}$ that are pairwise orthogonal for each fixed $\theta\in\Theta$. Her goal is to broadcast the orthogonality of this ensemble to herself and Bob. We label the input system as $A$ and the output systems as $A'B'$.  Let us proceed by distinguishing between the cases of perfect and approximate broadcasting.

\subsection{Perfect broadcasting}

It is often convenient to assume that Alice's broadcasting map takes the form of an isometry $U:A\to A'B'$.  This can always be done without loss of generality since a CPTP map can be built by combining an isometry with a discarding of subystems, and the latter can always be done by Alice or Bob after the broadcasting.  For isometry $U$, Alice then receives the output of the quantum channel $\mc{N}(\cdot)=\tr_{B'}U(\cdot)U^\dagger$ while Bob receives the output of its complementary channel $\mc{N}^c(\cdot)=\tr_{A} U(\cdot)U^\dagger$.  From this it follows that $\mc{S}_\Theta$ admits orthogonality broadcasting if and only if there exists a channel $\mc{N}$ such that
\begin{equation}
\; \Tr[\mc{N}\left(\rho_{i|\theta}\right)\mc{N}\left(\rho_{j|\theta}\right)]=\Tr[\mc{N}^c\left(\rho_{i|\theta}\right)\mc{N}^c\left(\rho_{j|\theta}\right)]=0\label{Eq:orth-broadcast-channel}
\end{equation}
for all $\theta_\Theta$ and $i\not=j$.

In the special case of classical orthogonality broadcasting, when Alice can only communicate classical information to Bob, the isometry $U$ gets replaced by a channel $\mc{E}:A\to A'X$ with a quantum-classical (qc) output.  We can express this as $\mc{E}(\cdot)=\sum_x \mc{A}_x(\cdot)\otimes\op{x}{x}$ with the $\{\mc{A}_x\}_x$ being a family of CP maps whose sum is CPTP (i.e. an instrument).  The condition for orthogonality on Bob's side is $\Tr[\mc{A}_x(\rho_{i|\theta})]\Tr[\mc{A}_x(\rho_{j|\theta})]=0$ for all $x$. Letting $\Pi_x=\mc{A}_x^\dagger(\mbb{I})$ denote elements of a POVM, where $\mc{A}_x^\dagger$ is the adjoint map, we can thus express the condition for orthogonality on Bob's side as the existence of a POVM $\{\Pi_x\}_x$ satisfying
\begin{equation}
\label{Eq:classical-broadcast-condition}
\Tr[\Pi_x\rho_{i|\theta}]\Tr[\Pi_x\rho_{j|\theta}]=0
\end{equation}
for all $x$, $\theta$, and $i\not=j$.  In other words, for every fixed $\theta$, there is at most one value $i$ such that $\Tr[\Pi_x\rho_{i|\theta}]\not=0$.   This means that the POVM $\{\Pi_x\}_x$ is able to perfectly identify the value $i$ given $x$ and $\theta$. In the next section, we will generalize this finding to not only perfect identification but also minimum error state discrimination.  For now, we consider the comparative powers of classical and non-classical orthogonality broadcasting in the case of zero error. 

We begin with a simple example that shows that perfect orthogonality broadcasting is possible using only classical communication even when not all states in $\cS_{\Theta}$ are simultaneously diagonalizable. This is in contrast to cloning and state broadcasting described in the introduction, which both require the ensemble to consist of pairwise commuting states.  Moreover, it shows that even for qutrits such examples exist, which highlights the nuance of orthogonality broadcasting in any dimension greater than two.
\begin{proposition}\label{prop:qubit-qutrit-example}
    Let
    \begin{align}\label{eq:minimal-ensemble}
        \cS_{\Theta} = \begin{Bmatrix} \rho_{0|0} \coloneq \ket{0} & \rho_{1|0} \coloneq \ket{1} \\
        \rho_{0|1} \coloneq \ket{\psi_{+}} & \rho_{1|1} \coloneq \ket{\psi_{-}}
        \end{Bmatrix} \ ,
\end{align}
where $\ket{\psi_{\pm}} := \frac{1}{2}\ket{0} + \frac{1}{2}\ket{1} \pm \frac{1}{\sqrt{2}}\ket{2} = \frac{1}{\sqrt{2}}(\ket{+} \pm \ket{2})$. This set admits perfect classical orthogonality broadcasting.
\end{proposition}
\begin{proof}
We show that perfect discrimination is possible by a broadcasting map with both parties receiving classical information, i.e. $\mc{E}(\cdot)=\sum_{x=0}^3\tr[\Pi_x(\cdot)]\op{x}{x}\otimes\op{x}{x}$, a general fact that holds for all classical orthogonal broadcasting (see Proposition \ref{Prop:classical-orthogonal-broadcasting}).  Consider the set of matrices
\begin{align*}
    \Pi_{0} \coloneq \begin{bmatrix}
        \frac{1}{2} & 0 & \frac{1}{2\sqrt{2}} \\
        0 & 0 & 0 \\
        \frac{1}{2\sqrt{2}} & 0 & \frac{1}{4}
    \end{bmatrix}
    \quad 
     & \Pi_{1} \coloneq \begin{bmatrix}
        \frac{1}{2} & 0 & -\frac{1}{2\sqrt{2}} \\
        0 & 0 & 0 \\
        -\frac{1}{2\sqrt{2}} & 0 & \frac{1}{4}
    \end{bmatrix} \\
    \Pi_{2} \coloneq \begin{bmatrix}
        0 & 0 & 0 \\
        0 & \frac{1}{2} & \frac{1}{2\sqrt{2}} \\ 
        0 & \frac{1}{2\sqrt{2}} & \frac{1}{4} 
    \end{bmatrix}
    \quad & 
    \Pi_{3} := \begin{bmatrix}
        0 & 0 & 0 \\
        0 & \frac{1}{2} & -\frac{1}{2\sqrt{2}} \\ 
        0 & -\frac{1}{2\sqrt{2}} & \frac{1}{4} 
    \end{bmatrix} \ . 
\end{align*}
A direct calculation will verify that each of these matrices have eigenvalues $\{3/4,0,0\}$ and by inspection they sum to identity. Thus, $\{\Pi_{i}\}_{i \in [3]}$ is a POVM. By direct calculation using Born's rule, one may verify that if Alice measures the input state with the given POVM, her possible (non-zero probability) outcomes are captured by the following table:
\begin{table}[H]
    \centering
    \begin{tabular}{c|c}
        State $\rho_{i|\theta}$ & Possible Outcomes \\ \hline
        $\rho_{0|0}$ & $\{0,1\}$ \\
        $\rho_{1|0}$ & $\{2,3\}$ \\
        $\rho_{0|1}$ & $\{0,2\}$ \\
        $\rho_{1|1}$ & $\{1,3\}$
    \end{tabular} . 
\end{table}

It follows that her possible outcomes partition upon receiving $\theta \in \{0,1\}$. Thus, if Alice measures the state with the given POVM and broadcasts the answer, then when Alice and Bob receive $\theta$, they can both determine the state correctly.
\end{proof}

It may be surprising that the communication of classical information suffices in the above example. Proposition \ref{Prop:classical-broadcasting-equal-2} below implies that in fact classical and quantum communication are equally powerful for perfect orthogonality broadcasting for all ensembles with $|\Theta|=2$. However, the following theorem shows a separation already exists as soon as $|\Theta|>2$.
\begin{theorem}
\label{Thm:classical-broadcasting-separation}
      Let
      \begin{align}
        \mc{S}_{\Theta} = 
        \begin{Bmatrix} 
            \rho_{0\vert0} \coloneq \ket{+_{01}} & \rho_{1\vert0} \coloneq \ket{-_{01}} \\
            \rho_{0\vert1} \coloneq \ket{+_{02}} & 
            \rho_{1\vert1} \coloneq \ket{-_{02}} \\
            \rho_{0\vert2} \coloneq \ket{\wt{+}_{12}} & \rho_{1\vert2} \coloneq \ket{\wt{-}_{12}}
        \end{Bmatrix} 
      \end{align} 
      be $|\Theta|=3$ pairs of orthogonal states, where $\ket{\pm_{mn}}=\tfrac{1}{\sqrt{2}}(\ket{m}\pm\ket{n})$ and $\ket{\wt{\pm}_{mn}}=\tfrac{1}{\sqrt{2}}(\ket{m}\pm i\ket{n})$.  Then it is possible to broadcast the orthogonality of $\mc{S}_{\Theta}$, yet classical orthogonality broadcasting of $\mc{S}_{\Theta}$ is impossible.
\end{theorem}
\begin{proof}
    We first construct an orthogonality broadcasting map for $\mc{S}_\Theta$. For $A=\mbb{C}^3$, define the isometry $U:A\to A'B'$ by 
\begin{align*}
    U\ket{0} = \frac{1}{\sqrt{2}}\left(\ket{00} + \ket{11} \right) \quad    U\ket{1} = \frac{1}{\sqrt{2}}\left(\ket{00} - \ket{11} \right) \quad
    U\ket{2} = \frac{1}{\sqrt{2}}\left(\ket{01} + \ket{10} \right) \ .
\end{align*}
Then 
\begin{align}
U\ket{+_{01}}&=\ket{0}^{A'}\ket{0}^{B'}, &U\ket{-{01}} &= \ket{1}^{A'}\ket{1}^{B},\notag\\
U\ket{+_{02}}&=\ket{+_{01}}^{A'}\ket{+_{01}}^{B'}, &U\ket{-_{02}}&=\ket{-_{01}}^{A'}\ket{-_{01}}^{B'}\notag\\
U\ket{\wt{+}_{12}}&=\ket{\wt{+}_{01}}^{A'}\ket{\wt{+}_{01}}^{B'}, &U\ket{\wt{-}_{12}}&=\ket{\wt{-}_{01}}^{A'}\ket{\wt{-}_{01}}^{B'}.\notag
\end{align}
Thus, the necessary orthogonality is preserved for both Alice and Bob for all elements of $\mc{S}_{\Theta}$.

We now show that there is no classical orthogonal broadcasting map for $\mc{S}_\Theta$.  According to Eq. \eqref{Eq:classical-broadcast-condition}, this would require a POVM $\{\Pi_x\}_x$ such that
\begin{align}
    0&=\bra{+_{01}}\Pi_x\ket{+_{01}}\bra{-_{01}}\Pi_x\ket{-_{01}}\notag\\
    0&=\bra{+_{02}}\Pi_x\ket{+_{02}}\bra{-_{02}}\Pi_x\ket{-_{02}}\notag\\    0&=\bra{\wt{+}_{12}}\Pi_x\ket{\wt{+}_{12}}\bra{\wt{-}_{12}}\Pi_x\ket{\wt{-}_{12}}\notag
\end{align}
for all $x$.  These three equalities can be satisfied only if three states from $\mc{S}_\Theta$ lie in the kernel  of $\Pi_x$.  However, any subset of three states from $\mc{S}_\Theta$ are linearly independent, which means the only solution to these equations is $\Pi_x=0$.
\end{proof}

\subsection{Separable Communication in Orthogonality Broadcasting}
With the separation of classical and quantum communication in orthogonality broadcasting established, an interesting question is whether orthogonality broadcasting actually needs the use of entanglement.  In other words, does there exist ensembles $\mc{S}=\{\rho_{i|\theta}\}_{i,\theta}$ whose orthogonality can be broadcast by map $\mc{E}:A\to A'B'$ only if the broadcast states $\sigma_{i|\theta}=\mc{E}(\rho_{i|\theta})$ are entangled?  
This question is fundamental for two reasons. First, the requirement for entangled states in the broadcasting map would be a sharp departure from cloning and state broadcasting, which, when feasible, can always be accomplished by generating product and separable output states, respectively.  Second, a more general question in quantum communication is to understand which distributed tasks are possible using the transmission of separable rather than entangled states.  For example, it is possible to distribute entanglement through the exchange of separable states \cite{Cubitt-2003a}, indicating the curious utility of quantum correlations beyond entanglement \cite{Chuan-2012a, Streltsov-2012a}.  The following theorem shows that separable carriers are insufficent in general for broadcasting orthogonality. 
\begin{theorem}
\label{Thm:no-sep-comm}
Let 
\begin{align}
    \mc{S}_{\Theta}=
    \begin{Bmatrix} 
        \rho_{0 \vert 0} \coloneq \ket{1} &
        \rho_{1 \vert 0} \coloneq \ket{2} \\
        \rho_{0 \vert 1} \coloneq \ket{+_{23}} & 
        \rho_{1 \vert 1} \coloneq \ket{-_{23}} \\
        \rho_{0 \vert 2} \coloneq \ket{+_{24}} & 
        \rho_{1 \vert 2} \coloneq \ket{-_{24}} \\
        \rho_{0 \vert 3} \coloneq \ket{\widetilde{+}_{34}} & \rho_{1 \vert 3} \coloneq \ket{\widetilde{-}_{34}}
    \end{Bmatrix}
\end{align}
be $|\Theta|=4$ pairs of orthogonal states and we remind the reader $\ket{\pm_{mn}}=\tfrac{1}{\sqrt{2}}(\ket{m}\pm\ket{n})$ and $\ket{\wt{\pm}_{mn}}=\tfrac{1}{\sqrt{2}}(\ket{m}\pm i\ket{n})$. Then $P_{\textbf{bc}}(\mc{S}_{\Theta})=1$. 
Furthermore, any broadcasting map must transform at least one of the states in $\mc{S}_{\Theta}$ into an entangled state.
\end{theorem}
\noindent Theorems \ref{Thm:classical-broadcasting-separation} and \ref{Thm:no-sep-comm} collectively establish different types of orthogonality broadcasting.  The ensemble $\mc{S}_{\Theta}$ in Theorem \ref{Thm:classical-broadcasting-separation} requires broadcasting a map that transforms the elements of $\mc{S}_{\Theta}$ into non-classical but possibly separable states.  Even stronger, the ensemble $\mc{S}_{\Theta}$ in Theorem \ref{Thm:no-sep-comm} requires a map that transforms at least one element of $\mc{S}_{\Theta}$ into an entangled state.

\medskip

\subsection{Approximate orthogonality broadcasting}
Even if perfect orthogonality broadcasting is not possible for $\mc{S}_\Theta$, one can consider approximate orthogonality broadcasting.  Here, it is natural to quantify the quality of approximation as the largest worst case error probability among Alice and Bob for correctly guessing the value $i\in\mc{I}_\theta$ given $\theta\in\Theta$.  For a given ensemble $\mc{S}_\Theta=\{\rho_{i|\theta}\}_{i,\theta}$ and joint distribution $p$ over $\Theta \times \cI$ where $\cI \coloneq \cup_{\theta} \cI_{\theta}$, this gives the measure 
\begin{align}
    \label{Eq:broadcasting-defn}
P_{\text{bc}}(\mc{S}_\Theta):=\max_{\mc{E}}\min_{P\in\{A,B\}}\sum_{\theta\in\Theta} p(\theta) P_{\text{guess}}(\{\sigma^P_{i|\theta}\}_{i}),
\end{align}
where the maximization is taken over all broadcasting maps $\mc{E}:\rho_{i|\theta}\mapsto \sigma_{i|\theta}^{AB}$, and 
\[P_{\text{guess}}(\{\sigma_{i|\theta}\}_{i}):=\max_{\{\Pi_i\}_{i}}\sum_{i\in\mc{I}_\theta}p(i\vert \theta)\Tr(\Pi_i\sigma_{i|\theta})\] 
is the maximum guessing probability for a uniform ensemble of states $\{\sigma_{i|\theta}\}_{i}$ using a positive operator-valued measure (POVM) $\{\Pi_i\}_i$.
Note that this approximation is well-defined for all ensembles $\mc{S}_{\Theta}=\{\rho_{i|\theta}\}_{i,\theta}$ even if pairwise orthogonality does not hold for each $\theta$.  Then orthogonality broadcasting of $\mc{S}_{\Theta}$ is possible iff $P_{\text{bc}}(\mc{S}_\Theta)=1$.  We let $P_{\text{c-bc}}(\mc{S}_\Theta)$ quantify approximate \textit{classical} orthogonality broadcasting, which is defined in the same way except with the maximization in Eq. \eqref{Eq:broadcasting-defn} restricted to broadcasting maps having a quantum-classical (qc) output. This notion of approximate orthogonality broadcasting proposed here is analogous to the well-studied tasks of approximate cloning and state broadcasting \cite{Buzek-1996a, Scarani-2005a}.

Note that the choice of local POVMs in the definition of $P_{\text{guess}}(\{\sigma_{i|\theta}^{AB}\})$ can depend on the value $\theta$.  The value $\theta$ is called post-information since it comes after the broadcasting map $\mc{E}$.  An alternative notion of post-information state discrimination has been considered in which the post-information comes after the choice of POVM \cite{Ballester-2008a, Carmeli-2022a} (see Fig.~\ref{Fig:orthogonality_broadcasting}).  The optimal guessing probability for ensemble $\mc{S}_\Theta$ in this scenario is then given by
\begin{align}\label{eq:post-information-measure}
    P_{\text{p-i}}(\mc{S}_\Theta):=\max \sum_{\theta,i}\sum_{x} q(i|x,\theta)p(\theta,i)\Tr(\rho_{i|\theta}\Pi_{x}),
\end{align}
where the maximization is taken over all POVMs $\{\Pi_x\}_x$ and all classical post-processing channels $q(i|x,\theta)$.  Despite this conceptual difference, it is not difficult to show that they are operationally equivalent.

\begin{proposition}\label{Prop:classical-orthogonal-broadcasting}
     For all $\mc{S}_{\Theta}$ and any choice of prior distribution on $\Theta \times \cI$, $P_{\text{c-bc}}(\mc{S}_\Theta)=P_{\text{p-i}}(\mc{S}_\Theta)$. In particular, this is because $P_{\text{c-bc}}(\cS_\Theta)$ is always achieved using a fully classical broadcasting channel $\cE^{A \to XY}$.
\end{proposition}
\begin{proof}
   We begin with the second claim. By \eqref{Eq:broadcasting-defn}, $P_{\text{c-bc}}(\cS_{\Theta})$ is always achieved by a fully classical broadcasting map $\cE^{A \to XY}$.  Indeed if the worse performing player receives classical information, then the minimum success probability for both parties will not change if both players follow the strategy of the classical player; and, if the worse performing party receives quantum information, then they can improve the minimum success probability by both following the classical player's strategy. 
    
    To establish the first claim, note the above argument shows the optimal fully classical broadcast map applies some optimal POVM $\{\Pi_{x}\}$ to $\rho_{i \vert \theta}$, and broadcasts the outcome to both parties. It follows for each $\theta$ both parties hold $\{\sigma_{i\vert \theta}^{X}\}_{i} = \{\sum_{x}\Tr[\Pi_{x}\rho_{i \vert \theta}]\dyad{x}\}_{i}$ distributed according to $p(i\vert \theta)$ and then apply an optimal POVM $\{\tau^{\theta}_{i}\}_{i}$ to guess the value of $i \in \cI_\Theta$. As $\{\sigma^{X}_{i\vert \theta}\}_{i}$ are classical distributions, for each $\theta \in \Theta$ the optimal POVM $\{\tau^{\theta}_{i}\}_{i}$ may be reduced to a conditional distribution $\{w_{\theta}(i\vert x)\}_{i,x}$. Thus, defining $q(i \vert x,\theta) \coloneq w_{\theta}(i \vert x)$ along with $\{\Pi_{x}\}_{x}$ defines a feasible strategy for \eqref{eq:post-information-measure}, so $P_{\text{c-bc}}(\cS_{\Theta}) \leq P_{\text{p-i}}(\cS_{\Theta})$. On the other hand, the reverse inequality follows from the fact that any POVM and post-processing strategy used to optimize $P_{\text{p-i}}(\mc{S}_\Theta)$ can be converted into a fully classical broadcasting map in which the POVM is first performed on $\rho_{i|\theta}$, the classical outcome is then broadcast to both Alice and Bob, and finally they both perform the post-processing locally.
\end{proof}

We can extend the relationship between classical broadcasting and post-information to general quantum broadcasting in the special case that $|\Theta|=2$.  The following proposition is a counterpart to Theorem \ref{Thm:classical-broadcasting-separation}.
\begin{proposition}
\label{Prop:classical-broadcasting-equal-2}
    For any prior distribution over $\Theta \times \cI$, $P_{\text{c-bc}}(\mc{S}_\Theta)=P_{\text{bc}}(\mc{S}_\Theta)$ for all $\mc{S}_{\Theta}$ with $|\Theta|\leq 2$ .
\end{proposition}

\begin{proof}
    It suffices to show $P_{\text{c-bc}}(\mc{S}_\Theta)\geq P_{\text{bc}}(\mc{S}_\Theta)$ for $\Theta=\{0,1\}$.  For any broadcasting map $\mc{E}:\rho_{i|\theta}\mapsto \sigma_{i|\theta}^{AB}$, there must be a $\theta_0\in\{0,1\}$ such that
    \begin{align}
    &\min_{P\in\{A,B\}}[p(0)P_{\text{guess}}(\{\sigma_{i|0}^P\}_i)+p(1)P_{\text{guess}}(\{\sigma_{i|1}^P\}_i)]\notag\\
    & \hspace{3cm} \leq \; p(\theta_0)P_{\text{guess}}(\{\sigma_{i|\theta_0}^A\}_i)  +p(\theta_0^c) P_{\text{guess}}(\{\sigma_{i|\theta_0^c}^B\}_i),
    \end{align}
    where $\theta_0^c=\theta_0\oplus 1$.  Let $\{\Pi_{i|\theta_0}^A\}_{i}$ and $\{\tau_{i|\theta_0^c}^B\}_i$ be optimal POVMs attaining the maximum guessing probability in $P_{\text{guess}}(\{\sigma_{i|\theta_0}^A\}_i)$ and $P_{\text{guess}}(\{\sigma_{i|\theta_0^c}^B\}_i)$, respectively.  Using these measurements we construct the fully classical broadcasting map
    \[\rho_{i|\theta}\mapsto\! \sum_{j,k} \Tr[\Pi_{j|\theta_0}^A\!\otimes\!\tau^B_{k|\theta_0^c}\mc{E}(\rho_{i|\theta})]\op{j,k}{j,k}^A\otimes\op{j,k}^B\!\!\!.\]
Alice and Bob both hold the classical variables $(j,k)$ with probability $p(j,k)=\Tr[\Pi_{j|\theta_0}^A\!\otimes\!\tau^B_{k|\theta_0^c}\mc{E}(\rho_{i|\theta})]$.  When $\theta$ is announced as post-information, they both submit $j$ as their guess if $\theta=\theta_0$ or they both submit $k$ as their guess if $\theta=\theta_1$.  Hence, they both have the same overall probability of correctly guessing, which is given by 
\begin{align}
    &p(\theta_{0})\sum_{j}\Tr[\Pi_{j|\theta_0}^A\!\otimes\! \mbbm{1}^B \mc{E}(\rho_{j|\theta})]\! + p(\theta^{c}_{0}) \sum_{k} \Tr[\mbbm{1}^A \otimes \tau^B_{k|\theta_0^c}\mc{E}(\rho_{k|\theta})]\notag\\
    =&p(\theta_{0}) P_{\text{guess}}(\{\sigma_{i|\theta_0}^A\}_i)+ p(\theta_{0}^{c}) P_{\text{guess}}(\{\sigma_{i|\theta_0^c}^B\}_i)\notag\\
    \geq& \min_{P\in\{A,B\}}[p(0)P_{\text{guess}}(\{\sigma_{i|0}^P\}_i) +p(1)P_{\text{guess}}(\{\sigma_{i|1}^P\}_i)] \ . 
\end{align}
Since we originally chose $\mc{E}$ as an arbitrary broadcasting map, it follows that $P_{\text{c-bc}}(\mc{S}_\Theta)\geq P_{\text{bc}}(\mc{S}_\Theta)$.
\end{proof}

Note that $P_{\text{p-i}}(\mc{S}_\Theta)$ in \eqref{eq:post-information-measure} is always optimized by a deterministic distribution as for any choice of POVM $\{\Pi_{x}\}$, for each $(\theta,x) \in \Theta \times \cX$ there is an $i'$ that maximizes $p(\theta,i)\Tr[\rho_{i\vert \theta} \Pi_{x}]$ over $i$ and $P_{\text{p-i}}$ is a maximization. As an extremal POVM has at most $d^{2}$ outcomes, the optimal value is determined using one of $N = (\max_{\theta \in \Theta} |\cI_{\Theta}|)^{d^{2}|\Theta|}$ deterministic distributions to search over. As maximizing the POVM is an SDP, there are $N$ SDPs that need to be evaluated to determine $P_{\text{p-i}}(\cS_{\Theta})$. Combining Propositions \ref{Prop:classical-orthogonal-broadcasting} and \ref{Prop:classical-broadcasting-equal-2}, for $\vert \Theta \vert \leq 2$, this gives a method for determining $P_{\text{bc}}(\cS_\Theta)$. In contrast, computing $P_{\text{bc}}(\mc{S}_\Theta)$ appears to be a challenging bilinear optimization problem in general.  Given this complexity, we next present an uncertainty relation that can be used to compute upper bounds on $P_{\text{bc}}(\mc{S}_\Theta)$ for certain ensembles.  The intuition behind the uncertainty relation is as follows.  Consider a pure state ensemble $\{\ket{a_\mu}\}_\mu$ with $\ket{\alpha_\mu}^{AB'}=U\ket{a_\mu}$.  If $\ket{\alpha_0}$ and $\ket{\alpha_1}$ are two entangled states that Alice can distinguish with high probability by measuring subsystem $A$, then the reduced density matrices $\alpha_0^A$ and $\alpha_1^A$ must be nearly orthogonal.  Consequently, tracing out Alice in any superposition of $\ket{\alpha_0}$ and $\ket{\alpha_1}$ will effectively destroy the relative phase between these states, thereby making it difficult for Bob to distinguish superpositions of $\ket{\alpha_0}$ and $\ket{\alpha_1}$.  We can quantify this tradeoff in terms of the fidelity $F(\rho,\sigma)=\Vert \sqrt{\rho}\sqrt{\sigma}\Vert_1$ and trace distance $\TD(\rho,\sigma)=\tfrac{1}{2}\Vert\rho-\sigma\Vert_1$ between hermitian operators. 

\begin{lemma}[Uncertainty Relation]\label{lem:main-text-UR}
    For any vectors $\ket{\alpha_{0}}^{AB'}$ and $\ket{\alpha_{1}}^{AB'}$, if $\ket{\alpha_{\theta}}^{AB'} = \cos(\theta/2)\ket{\alpha_{0}}^{AB'} + e^{i\phi}\sin(\theta/2)\ket{\alpha_{1}}^{AB'}$ and $\ket{\alpha_{\omega}}^{AB'} = \cos(\omega/2)\ket{\alpha_{0}}^{AB'} + e^{i\phi'}\sin(\omega/2)\ket{\alpha_{1}}^{AB'}$, then
    \begin{equation}
    \; \TD(\alpha_{\theta}^{B'},\alpha_{\omega}^{B'}) 
        \leq  |z_1|F(\alpha_0^A,\alpha_1^A)+|z_2| \TD(\alpha_{0}^{B'}, \alpha_{1}^{B'})  \label{eq:TD-F-uncertainty-relation},
    \end{equation}
    where $z_1 = \frac{1}{2}(\sin(\theta)e^{-i\phi}-\sin(\omega)e^{-i\phi'})$ and $z_2 = \frac{1}{2}(\cos(\theta)-\cos(\omega))$.
\end{lemma}
We remark that we in fact prove a more general uncertainty relation for two vectors that both decompose into linear combinations of the same set of vectors (see the Supplemental Material), but the above suffices for the main results of this work. Moreover, we highlight that this lemma may be seen as our major technical contribution. As noted earlier, unlike other uncertainty relations that rely upon the incompatibility of the measurement directly (see \cite{Coles-2017a} for a review), ours follows from the geometric relation between the initial quantum states.

To appreciate the utility of Proposition \ref{lem:main-text-UR}, we immediately obtain a no-go result for broadcasting the orthogonality of any two distinct pairs of basis vectors $\{\ket{0},\ket{1}\}$ and $\{\ket{\hat{n}},\ket{-\hat{n}}\}$, where $\ket{\hat{n}}=\cos(\theta/2)\ket{0}+\sin(\theta/2)e^{i\phi}\ket{1}$ and $\ket{-\hat{n}}=\sin(\theta/2)\ket{0}-\cos(\theta/2)e^{i\phi}\ket{1}$ with $\theta\in(0,\pi)$. 
\begin{corollary}
\label{Cor:orthogonality-broadcasting-no-go}
    Orthogonality broadcasting is not possible for the set of states $\mc{S}_\Theta=\{\ket{0},\ket{1},\ket{\hat{n}},\ket{-\hat{n}}\}$ if $\ket{\hat{n}} \not \in \{\ket{0},\ket{1}\}$.
\end{corollary}
\begin{proof} Without loss of generality, we assume that $\{\ket{0},\ket{1}\}$ remains locally orthogonal for Alice.  The isometry maps $\{\ket{0},\ket{1},\ket{\hat{n}},\ket{-\hat{n}}\}$ into $\{\ket{\alpha_0}^{AB'},\ket{\alpha_1}^{AB'},\ket{\alpha_\theta}^{AB'},\ket{\alpha_{\theta-\pi}}^{AB'}\}$ of Proposition \ref{lem:main-text-UR}, where we take $\omega=\theta-\pi$ and $\phi=-\phi$.  Equation \eqref{eq:TD-F-uncertainty-relation} then reads  
\begin{equation}
\TD(\alpha_{\theta}^{B'},\alpha_{\theta-\pi}^{B'}) 
        \leq  |\sin\theta| F(\alpha_0^A,\alpha_1^A)+|\cos\theta| \TD(\alpha_{0}^{B'}, \alpha_{1}^{B'}).\notag
\end{equation}
Orthogonality of $\alpha_0^A$ and $\alpha_1^A$ implies that $F(\alpha_0^A,\alpha_1^A)=0$, and so $\TD(\alpha_{\theta}^{B'},\alpha_{\theta-\pi}^{B'}) 
        \leq  |\cos\theta| \TD(\alpha_{0}^{B'}, \alpha_{1}^{B'})\leq |\cos\theta|<1$, assuming $\theta \not \in \{\frac{k \pi}{2}\}_{k \in \mbb{N}}$, which is equivalent to $\ket{\hat{n}} \not \in \{\ket{0},\ket{1}\}$.  
This proves the corollary since $\alpha_{\theta}^{B'}$ and $\alpha_{\theta-\pi}^{B'}$ are orthogonal if and only if $\TD(\alpha_{\theta}^{B'},\alpha_{\theta-\pi}^{B'})=1$.
\end{proof}

Before moving forward, we remark that we may relax Lemma \ref{lem:main-text-UR} in terms of guessing probability, which follows from the Fuchs-van de Graaf inequality and the equiprobable case of the Holevo-Helstrom theorem, i.e.~$D_{\tr}(\rho,\sigma) = 2p_{g}(\rho,\sigma)-1$ where $p_{g}(\rho,\sigma)$ denotes the optimal guessing probability when the two states are equiprobable \cite{WatrousBook}.
\begin{corollary}\label{cor:guessing-probability-UR}
    Under the same conditions as Lemma \ref{lem:main-text-UR},
    \begin{equation}
    \begin{aligned}
        p_{g}(\alpha_{\theta}^{B'},\alpha_{\omega}^{B'}) \leq& \frac{1}{2}\Big[\vert z_{1} \vert \sqrt{1-(2p_{g}(\alpha_{0}^{A},\alpha_{1}^{A})-1)^{2}} + \vert z_{2} \vert(2p_{g}(\alpha_{0}^{B'},\alpha_{1}^{B'})-1) + 1\Big] \ .
    \end{aligned}
    \end{equation}
\end{corollary}
In the case $\mathcal{S}_{\Theta}$ decomposes into states
\begin{align}
    \rho_{0\vert 0} &= \ket{a_{0}}^{S} \quad \rho_{1\vert 0} = \ket{a_{1}}^{S} \\
    \rho_{0 \vert 1} &= \cos(\theta/2)\ket{a_{0}}^{S} + e^{i\phi}\sin(\theta/2)\ket{a_{1}}^{S} \\
    \rho_{1 \vert 1} &= \cos(\omega/2)\ket{a_{0}}^{S} + e^{i\phi}\sin(\omega/2)\ket{a_{1}}^{S} \ , 
\end{align}
Corollary \ref{cor:guessing-probability-UR} may be used to provide bounds on $P_{\text{bc}}(\cS_{\Theta})$ as the constraint will hold for \textit{any} choice of broadcasting channel $\cE$. We however use this idea in further generality in a subsequent section, Section \ref{sec:SDP-Bounds}, so we omit an example here. We remark, in the special case given above, that Propositions \ref{Prop:classical-orthogonal-broadcasting} and \ref{Prop:classical-broadcasting-equal-2} already imply there exists a direct method to determine the optimal value, although it is more computationally intensive.

\section{LOSQC/LOSCC state discrimination and QPV in the no Pre-Shared Entanglement Model}\label{sec:QPV} 
We now shift to the application of our results to QPV. The class of protocols we consider asks an honest prover to distinguish the state drawn from a GOP ensemble $\{\ket{a_k}\ket{b_k}\}_k$ by performing a joint measurement, which is depicted in Fig.~\ref{subfig:honest-QPV}. For clarity, we begin by identifying the communication structure of attacking such a QPV protocol, which is depicted in Fig.~\ref{subfig:dishonest-QPV}.  From the adversarial perspective, there are two dishonest provers, which we can call Alice and Bob, who each receive a portion of the global state that intersects their worldline on the spacetime diagram, split their local state in two to share part with the other party, and then, once each holds their respective shares, they may each guess $k$ to forward to their respective verifier.\footnote{We remark that without loss of generality there are only two dishonest verifiers as there are only two messages to intercept and respond to in this one spatial-dimensional setting.} Formally, Alice intercepts $\ket{a_k}^{A}$ and performs the isometry $U:\ket{a_k}^{A} \mapsto \ket{\alpha_{k}}^{A_{0}B_{1}}$ while Bob intercepts $\ket{b_k}^{B}$ and applies $V:\ket{b_k}^{B} \mapsto \ket{\beta_{k}}^{B_{0}A_{1}}$. This is without loss of generality as any communication strategy allowed by quantum mechanics under the given timing constraints can be represented by this form. After the quantum communication, Alice holds systems $A_{0}A_{1}$ and Bob holds $B_{0}B_{1}$, and they both must perform local measurements to learn the value $k$. This is the most general strategy allowed by quantum mechanics in the no pre-shared entanglement (no-PE) model under the time constraints of special relativity. We call such the set of such communication strategies `local operations and simultaneous quantum communication' (LOSQC) and the optimal success probability over all such strategies can be expressed as 
\begin{align}\label{eq:optimal-LOSQC-strategy}
    \hspace{-2mm} \Pr_{\LOSQC}[\cS] \coloneq \max  \sum_{k} p(k) \Tr[\Pi_{k} \otimes \tau_{k} (U \otimes V)(\rho_{k})(U \otimes V)^{\dagger}] \ ,
\end{align}
where the $\{\rho_{k}\}_{k}$ denote the possible states, and as we just explained, without loss of generality the maximization is over isometries $U^{A \to A_0 B_1}$, $V^{B \to B_0 A_1}$ and POVMs $\{\Pi^{A_0 A_1}_{k}\}$, $\{\tau^{B_0 B_1}_{k}\}$. In other words, \eqref{eq:optimal-LOSQC-strategy} gives the exact optimal success probability of attacking a state discrimination QPV protocol on inputs $\{\rho_{k}\}_{k}$ in the no-PE model. If we restricted the adversaries to only using classical communication between them, we would consider `local operations and simultaneous classical communication' (LOSCC). In this case, rather than general isometries, Alice and Bob perform local instruments $\{\mc{A}_x\}_x$ and $\{\mc{B}_y\}_y$, respectively, with classical information $x$ going to Bob and $y$ going to Alice.

\begin{remark}
Before moving forward, we note that LOSQC is the relevant attack model even if there is pre-shared entanglement. For example, consider Alice and Bob will intercept some state from $\{\wt{\sigma}_{k}^{AB}\}_{k \in \cK}$ and pre-share some resource $\zeta^{A'B'}$. Then defining $\rho_{k} \coloneq \wt{\sigma}_{k} \otimes \zeta$ for all $k$ reduces studying attacking the QPV with resource state $\tau$ to the form of \eqref{eq:optimal-LOSQC-strategy} where the isometries now are $U^{AA' \to A_{0}B_{1}}, V^{BB' \to B_{0}A_{1}}$.
\end{remark}

\begin{figure}[H]
    \centering
    \begin{subfigure}[t]{0.5\textwidth}
        \centering
        \includegraphics[width=\columnwidth]{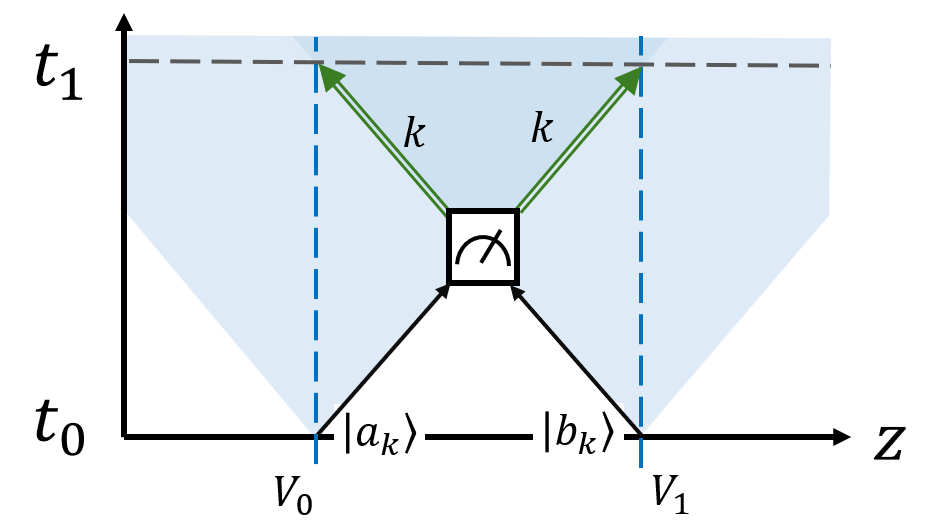}
        \caption{}\label{subfig:honest-QPV}
    \end{subfigure}%
    ~ 
    \begin{subfigure}[t]{0.5\textwidth}
        \centering
        \includegraphics[width=\columnwidth]{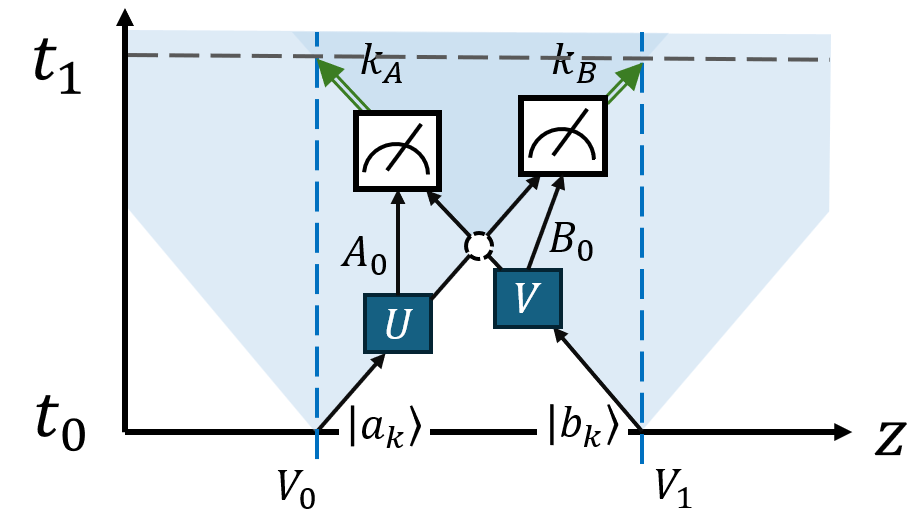}
        \caption{}\label{subfig:dishonest-QPV}
    \end{subfigure}
    \caption{Spacetime diagram of the QPV protocols considered in Section \ref{sec:QPV}. Green double-lines denote classical communication, black lines depict quantum communication. Blue dashed lines denote the worldlines of the verifiers $V_{0},V_{1}$. The dashed grey line denotes the required time to respond to the verifiers. (a) \textit{Honest prover case.} Verifiers $V_{0},V_{1}$ send $\ket{a_{k}},\ket{b_{k}}$ respectively at time $t_{0}$. The honest prover uses any measurement that distinguishes the set of GOP states $\{\ket{a_{k}}\ket{b_{k}}\}_{k}$, which determines $k$ with certainty. The prover then forwards the answer to each verifier to be received at $t_{1}$. (b) \textit{Dishonest prover case.} The dishonest provers intercept the states $\ket{a_{k}},\ket{b_{k}}$ respectively. As they both must answer, they split the inputs using isometries $U^{A \to A_{0}B_{1}}, V^{B \to A_{1}B_{0}}$ and forward the corresponding pieces to eachother. They then perform measurements to guess the value of $k$ and forward the verifiers their guesses. This is the structure of local operations and simultaneous communication (LOSQC). Note the dishonest setting requires neither prover's wordline be where the honest prover should be (the dotted white circle).}\label{fig:QPV-section}
\end{figure}

\subsection{Zero-Error LOSQC and LOSCC State Discrimination} 
Before directly addressing QPV, we study the problem of zero-error state discrimination in LOSQC and LOSCC. For GOPs with one system being a qubit, we find that LOSQC is no more powerful than LOSCC, and we obtain a reduction to classical orthogonality broadcasting. 
\begin{theorem}
\label{thm:C2-Cd-LOSCC}
    A GOP $\{\ket{a_k}^A\ket{b_k}^B\}_k\subset\mbb{C}^2\otimes\mbb{C}^d$ is perfectly distinguishable by LOSCC if and only if it is perfectly distinguishable by LOSQC.  Moreover, to be distinguishable by either, the ensemble must have the form
    \begin{align}\label{eq:C2-Cd-LOSCC-form}
\ket{0}^A\ket{s_{i|0}}^B,\quad \ket{1}^A\ket{s_{i|1}}^B,\quad
\ket{\phi_i}^A\ket{t_i}^B,
    \end{align}
    where the $\ket{\phi_i}$ are arbitrary states, the $\{\ket{s_{i|\theta}}\}_{i,\theta}$ is an ensemble whose orthogonality $\ip{s_{i|\theta}}{s_{j|\theta}}=\delta_{ij}$ can be classically broadcast, and each $\ket{t_i}$ is orthogonal to every other state on Bob's side.
\end{theorem}

While it is straightforward to see one may efficiently determine if an input set of product states has the form in \eqref{eq:C2-Cd-LOSCC-form} via Gram matrices, it is unclear if there exists an efficient method for determining if a set of the form in \eqref{eq:C2-Cd-LOSCC-form} is in fact perfectly discriminable under LOSCC. In particular, by Propositions \ref{Prop:classical-orthogonal-broadcasting} and \ref{Prop:classical-broadcasting-equal-2}, deciding whether a $2\otimes d$ GOP ensemble can be perfectly distinguished by LOSQC/LOSCC reduces to the post-information discrimination problem depicted in Fig.~\ref{Fig:orthogonality_broadcasting}(b). However, the method for determining $P_{\text{bc}}(\cS_{\Theta})$ provided below the Proof of Proposition \ref{Prop:classical-broadcasting-equal-2} is not efficient in the sense that it scales exponentially in the dimension of the states. Nonetheless, for the simplest case of two qubits, one yet again recovers Corollary \ref{Cor:orthogonality-broadcasting-no-go}, so in this case it is easy to determine. On the other hand, Proposition \ref{prop:qubit-qutrit-example} shows that in $\mbb{C}^{2} \otimes \mbb{C}^{3}$ the characterization of all GOP sets perfectly distinguishable under LOSQC or LOSCC is already non-trivial.

The equivalence between LOSCC and LOSQC stated in Theorem \ref{thm:C2-Cd-LOSCC} is an extension of Proposition \ref{Prop:classical-broadcasting-equal-2} to the QPV setting. We can obtain a separation between LOSCC and LOSQC from the ensemble of Theorem \ref{Thm:classical-broadcasting-separation} as a direct corollary to Theorem \ref{thm:C2-Cd-LOSCC}.
\begin{corollary}\label{cor:LOSQC-LOSCC-separation}
    The following states are distinguishable by LOSQC but not LOSCC:
    \begin{align}
\ket{\psi_1}\!&=\ket{0}\ket{+_{12}},&\ket{\psi_2}\!&=\ket{1}\ket{+_{13}},&\ket{\psi_3}\!&=\ket{2}\ket{\wt{+}_{23}}\notag\\
\ket{\psi^\perp_1}\!&=\ket{0}\ket{-_{12}},&\ket{\psi^\perp_2}\!&=\ket{1}\ket{-_{13}},&\ket{\psi^\perp_3}\!&=\ket{2}\ket{\wt{-}_{23}}.\notag
    \end{align}
\end{corollary}

The ensemble in the above corollary provides the first example of a product state input QPV protocol that, in the no-PE model, is secure against attackers restricted to LOSCC but completely broken by attackers restricted to LOSQC operations. This resolves an open question from \cite{Allerstorfer-2022a}, where a separation between LOSQC and LOSCC attackers was shown when the inputs were entangled.

\subsection{Analytic Error Bounds for QPV in the no-PE Model} While the problem of perfect state discrimination is important for understanding the fundamental limitations of LOSCC and the comparative power of LOSQC, in practical QPV and related tasks one typically needs to replace perfect discrimination with bounded error. Specifically, as already shown, to establish security guarantees against adversaries with quantum communication, we need error bounds on how well a GOP ensemble can be discriminated by LOSQC. Establishing such bounds is difficult because the success probability of the optimal attack given by \eqref{eq:optimal-LOSQC-strategy} is highly non-linear.  However, when the inputs are product $\cS = \{\ket{a_{k}}^{A}\ket{b_{k}}^{B}\}$ we may use the tools and insights we developed for orthogonality broadcasting to obtain upper bounds on \eqref{eq:optimal-LOSQC-strategy}. A simple way to see this is to note that because the optimal strategy is limited by the party that guesses incorrectly most frequently, we have
\begin{equation}\label{eq:LOSQC-to-broadcasting-measure}
\begin{aligned}
    \Pr_{\LOSQC}[\cS] \leq \max_{U,V} \min_{P \in \{A,B\}} \max_{ \{\gamma_{k}\}} \sum_{k}p(k)\Tr[\gamma_{k} \sigma_{k}^{P_0 P_1}] \ ,
\end{aligned}
\end{equation}
where $\sigma_{k}^{P_0 P_1} \coloneq \Tr_{\ol{P}_0 \ol{P}_{1}}[(U \otimes V)(\rho_{k})(U \otimes V)^{\dagger}]$, $\ol{P}$ denotes the other party than $P$, and the second maximization is over POVMs. Note that the RHS may be viewed as a direct generalization of $P_{\text{bc}}(\cS_{\theta})$ in \eqref{Eq:broadcasting-defn}. Indeed, in the case that one input is always classical, so that without loss of generality it is simply copied, the optimal LOSQC (resp.~LOSCC) strategy is determined by the optimal orthogonality (resp.~classical orthogonality) broadcasting strategy. This observation alone allows us to obtain results from our earlier results.
\begin{corollary}\label{cor:LOSQC-claims-from-ortho-broadcast}
    Let $\cS = \{\ket{a_{k}}^{A}\ket{b_{k}}^{B}\}$ where Bob's input is classical, i.e. $\ket{b_{k}} \in  \{\ket{z}\}_{z \in \cZ}$ for all $k$ where $\{\ket{z}^{Z}\}$ is an orthonormal basis of $B$. Then there exists $\cS_{\Theta}$ induced by $\cS$, such that
    \begin{align}
        \Pr_{\text{LOSCC}}[\cS] = P_{\text{c-bc}}[\cS_{\Theta}] = P_{\text{p-i}}[\cS_\Theta ] \ .
    \end{align}
    This in particular implies in this scenario the optimal LOSCC strategy never requires a quantum memory. Moreover, if $|\cZ| = 2$, then the optimal LOSQC strategy is an LOSCC strategy.
\end{corollary}
\begin{proof}
    By assumption, $\cS$ may be taken to be $\cS_{\Theta}$. Thus, $\Pr_{\text{LOSCC}}[\cS] = \Pr_{\text{LOSCC}}[\cS_{\Theta}] \leq P_{\text{c-bc}}[\cS_{\theta}] = P_{\text{p-i}}[\cS_{\Theta}]$, where the last equality is Proposition \ref{Prop:classical-orthogonal-broadcasting}. The concern would be that the inequality is strict. However, as explained below the proof of Proposition \ref{Prop:classical-broadcasting-equal-2}, $\Pr_{\text{p-i}}[\cS_{\Theta}]$ is always achieved by a deterministic conditional distribution. Thus, by Bob performing the optimal POVM for $P_{\text{p-i}}[\cS_{\Theta}]$ (See \eqref{eq:post-information-measure}) and having Alice and Bob use the same deterministic conditional distribution, they always guess the same $i$. Thus, $\Pr_{\text{LOSCC}}[\cS_{\Theta}] = P_{\text{c-bc}}[\cS_{\theta}]$. Moreover, this strategy uses a fully classical broadcast channel, so it requires no quantum memory. This establishes the first point. 
    
    For the second, $\Pr_{\LOSQC}[\cS_{\Theta}] \leq P_{\text{bc}}[\cS_{\Theta}] = P_{\text{c-bc}}[\cS_{\Theta}] = \Pr_{\text{LOSCC}}[\cS_{\Theta}] \leq \Pr_{\text{LOSQC}}[\cS_{\Theta}]$, where the first equality is Proposition \ref{Prop:classical-broadcasting-equal-2} and the second equality is what we just proved.
\end{proof}
We remark that the second point of Corollary \ref{cor:LOSQC-claims-from-ortho-broadcast} generalizes the observation that the optimal strategy for BB84 QPV \cite{Buhrman-2014a} is an LOSCC strategy as first proven in \cite{Tomamichel-2013a}.

While Corollary \ref{cor:LOSQC-claims-from-ortho-broadcast} is structurally important, as mentioned, the most important problem for QPV is upper bounds on the achievable state discrimination under LOSQC. We can obtain such bounds as an application of the uncertainty relation in Lemma \ref{lem:main-text-UR}. 
\begin{theorem}
\label{thm:discrimination-cond}
Consider any ensemble containing four states of the form
\begin{equation}
\begin{aligned}\label{min_ensemble}
    \ket{\psi_0}^{AB}&=\ket{a_0}^A\ket{b_0}^B\\ \ket{\psi_1}^{AB}&=\ket{a_1}^A\ket{b_1}^B \\    \ket{\psi_2}^{AB}&=\ket{a_2}^A(\cos\tfrac{\theta}{2}\ket{b_0}+\sin\tfrac{\theta}{2} e^{i\phi} \ket{b_1})^B \\    \ket{\psi_3}^{AB}&=\ket{a_3}^A(\cos\tfrac{\omega}{2}\ket{b_0}+\sin\tfrac{\omega}{2} e^{i\phi'} \ket{b_1})^B,
\end{aligned}
\end{equation}
with $\bra{a_{0}}\ket{a_{1}} > 0$. Suppose Alice and Bob can identify each state with at least probability $1-\ve$ using some LOSQC protocol; i.e. given state $\ket{\psi_k}^{AB}$ they both guess $k$ with probability at least $1-\ve$.  Then
\begin{align}
\label{Eq:thm-ineq}
    1 \leq \frac{2|z_1|\sqrt{\ve(1-\ve)}}{|\langle a_{0} | a_{1} \rangle|^{2}}&+ |z_2|\sqrt{1-|\langle b_{0} | b_{1} \rangle|^{2}} +\sqrt{1-|\langle a_{2} | a_{3} \rangle|^{2}}+ 2\ve,
\end{align}
where $z_1$ and $z_2$ are defined in Lemma \ref{lem:main-text-UR}.
\end{theorem}

From Theorem \ref{thm:discrimination-cond}, we can recover the optimal LOSQC error bounds for BB84 QPV determined in \cite{Buhrman-2014a,Tomamichel-2013a}. This shows that Theorem 4, and by extension all steps in the proof including our uncertainty relation, are tight.
\begin{corollary}\label{cor:LOSQC-BB84}
    Given the set of globally orthogonal product states 
    $$ \{\ket{0}^{A}\ket{0}^{B}, \ket{0}^{A}\ket{1}^{B}, \ket{1}^{A}\ket{\hat{n}}^{B}, \ket{1}^{A}\ket{-\hat{n}}^{B}\} \ , $$ where $\ket{\hat{n}} = \cos(\theta/2)\ket{0} + \sin(\theta/2)e^{i\phi} \ket{1}$. The minimal $\ve$ such that both parties can guess the state with probability at least $\ve$ under LOSQC is 
    \begin{align}\label{eq:min-epsilon}
         \ve^{\star}(\theta) \geq \frac{1-\cos(\theta)+(1+\sqrt{2})\sin(\theta)^{2}}{2(1+\sin(\theta)^{2})} \ . 
    \end{align}
    In particular, this recovers the tight bounds for BB84 QPV where $\theta = \pi/2$.
\end{corollary}
\begin{proof}
    Note without loss of generality we may let $\phi = 0 = \phi'$ by rotating the states into the $\{\ket{0},\ket{1},\ket{+},\ket{-}\}$ plane of the Bloch sphere. For the states we consider, $\ket{-\hat{n}} = \cos(\frac{\theta-\pi}{2})\ket{0} + \sin(\frac{\theta-\pi}{2})\ket{1}$. It follows $z_{1} = \sin(\theta)$ and $z_2 = \cos(\theta)$ in Theorem \ref{thm:discrimination-cond}. Using Theorem \ref{thm:discrimination-cond}, one finds the condition 
    $$1 \leq 2\sin(\theta)\sqrt{\ve(1-\ve)} +\cos(\theta) + 2\ve \ . $$  Solving this as an equality condition, one finds \eqref{eq:min-epsilon}. To recover the BB84 case, one uses $\theta = \pi/2$ to obtain $1 -\ve^{\star} = 1 - \frac{1}{4}(2+\sqrt{2}) = \frac{1}{2}+\frac{1}{2\sqrt{2}}$. However, this can be achieved by Bob measuring each state in the `Breidbart basis', so this bound is tight.
\end{proof}
We remark that by generalizing the Breidbart basis measurement to measuring the states equidistant on the Bloch sphere between $\ket{0},\ket{\hat{n}}$, i.e. projecting $\ket{\hat{m}_{\theta}} = \cos(\theta/4)\ket{0} + \sin(\theta/4)\ket{1}$ and its orthogonal vector, one can obtain a minimum error of $1-\cos(\theta/4)^{2}$. This is a slightly larger error than given in \eqref{eq:min-epsilon}, which would suggest \eqref{eq:min-epsilon} is loose except for $\theta \in \{0,\pi/2\}$.

We provide another example of how Theorem \ref{thm:discrimination-cond} can be used to bound the minimum error guessing probability.

\begin{example} An unextendible product basis (UPB). Consider the tripartite UPB known as \textbf{\textsf{Shifts}} \cite{Bennett-1999a}.  Combining two of the parties yields the bipartite ensemble
\begin{align}
\ket{\psi_0}&=\ket{00}\ket{0},&\ket{\psi_1}&=\ket{+-}\ket{1},\notag\\    \ket{\psi_2}&=\ket{-1}\ket{+},&\ket{\psi_3}&=\ket{1+}\ket{-}.
\end{align}
If one of these is chosen with uniform probability and distributed to Alice and Bob, their smallest possible guessing error $\ve$ using LOSQC satisfies $\ve>5.52\times 10^{-4}$.  Clearly this also provides a lower bound on the tripartite error probability for \textbf{\textsf{Shifts}} under LOSQC.
 \end{example}

\subsection{Semidefinite Program Bounds}\label{sec:SDP-Bounds}
We remark a major appeal of Theorem \ref{thm:discrimination-cond} is that it is quick to evaluate. One can however in various cases directly obtain bounds on the probability of distinguishing a set of input states $\cS \coloneq \{\ket{\psi_{k}}^{AB}\}_{k}$ over a given distribution using LOSQC using the uncertainty relation to convert the problem into a SDP over vectors. For example, using the same tools as to establish Theorem \ref{thm:discrimination-cond}, we can obtain the following SDP, which is provably generally non-trivial.
\begin{proposition}\label{prop:SDP-version-of-Thm-4}
    Consider the ensemble $\cS$ of the form given in Theorem \ref{thm:discrimination-cond} and $\psi_0, \psi_1$ (resp.~$\psi_\theta, \psi_\omega$) are each provided with probability $p/2$ (resp.~$(1-p)/2$). Then, $\Pr_{\LOSQC}[\cS]$ is upper bounded by the optimal value of 
    \begin{equation}\label{eq:thm-4-sdp}
        \begin{aligned}
            \max \; \; & \min\{f(p,\cS,\mbf{r}), f(p,\cS,\mbf{s})\}-1 \\
            \text{s.t.} \; \; & 
            r_1 \leq \frac{1}{2}\Big[\vert z_{1} \vert \sqrt{1-(2s_0-1)^{2}} + \vert z_{2} \vert(2r_0-1) + 1\Big] \\
            & s_1 \leq \frac{1}{2}\Big[\vert z_{1} \vert \sqrt{1-(2r_0-1)^{2}} + \vert z_{2} \vert(2s_0-1) + 1\Big] \\
            & 0 \leq \mbf{r},\mbf{s} \leq 1
        \end{aligned}
    \end{equation}
    where $f(p,\cS,\mbf{x}) \coloneq  p \cdot [p_{g}(a_{0},a_{1}) + x_0] + (1-p) [p_{g}(a_3,a_4) + x_1]$. This is a convex optimization problem and equals unity if and only if $\langle a_0 \vert a_1 \rangle = 0 = \langle a_3 \vert a_4 \rangle$ and $\theta = \pm \arccos(1+\cos(\omega))$.
\end{proposition}
\noindent We remark that Proposition \ref{prop:SDP-version-of-Thm-4} may be used as yet another method for proving Corollary \ref{cor:LOSQC-BB84} although we omit this.

There are two key ideas in establishing Proposition \ref{prop:SDP-version-of-Thm-4} as well as the other results in this section. The first is captured in \eqref{eq:LOSQC-to-broadcasting-measure}. The second is that the uncertainty relation (Corollary \ref{cor:guessing-probability-UR}) implies constraints on the guessing probabilities independent of Alice and Bob's choices of broadcasting channels, so we may replace maximizing over the isometries by maximizing the guessing probabilities allowed according to the uncertainty relation as may be seen in \eqref{eq:thm-4-sdp}. We refer the reader to the appendices for more information. 

The major strength of these semidefinite program bounds is not in Proposition \ref{prop:SDP-version-of-Thm-4}, but in the more direct analysis of discriminating product state ensembles $\cS \coloneq \{\ket{a_{k}}^{A}\ket{b_{k}}^{B}\}$ that are not of the form in Theorem \ref{thm:discrimination-cond} but are limited by the inability to broadcast orthogonality. This results in obtaining stronger, i.e. smaller-valued, upper bounds on the optimal distinguishing probability under LOSQC, $\Pr_{\LOSQC}[\cS]$.\footnote{Smaller-valued upper bounds are stronger because, in QPV, LOSQC represents the attack model of dishonest provers, so if the probability of success is lower, then the protocol is more secure against dishonest provers.} We will in particular establish a qutrit classical-quantum product state that outperforms the BB84 set in terms of amplifying the error (Theorem \ref{thm:Overlapping-BB84}) and, to the best of our knowledge, the first proof of error bounds on quantum-quantum product states (Theorem \ref{thm:qq-qc-separation}). We note that, as shown in \ref{app:ob-cloning-games-and-MoE}, neither of these cases can be analyzed using Monogamy-of-Entanglement games.

We begin by showing there exist sets of globally orthogonal classical-quantum product states that result in provably stronger security than using BB84 states. In particular, we consider the following set of states:
\begin{align}\label{eq:overlapping-BB84-states}
		\cS_{OBB} = \begin{Bmatrix}
			\ket{\psi_{1}} = \ket{0}\ket{1} \\ 
            \ket{\psi_{2}} = \ket{0}\ket{2}   &
			\ket{\psi_{3}} = \ket{1}\ket{1+2}  \\
            \ket{\psi_{4}} = \ket{1}\ket{1-2}   &
			\ket{\psi_{5}} = \ket{0}\ket{0}  \\
            \ket{\psi_{6}} = \ket{2}\ket{0+1}  & \ket{\psi_{7}} = \ket{2}\ket{0-1}   
		\end{Bmatrix} \ .
\end{align}
It is not hard to see that the subsets $\cS'_{1} \coloneq \{\ket{\psi_{1}},\ket{\psi_{2}},\ket{\psi_{3}},\ket{\psi_{4}}\}$ and $\cS'_{2} \coloneq \{\ket{\psi_{1}},\ket{\psi_{5}},\ket{\psi_{6}},\ket{\psi_{7}}\}$ are both equivalent to $\cS_{\BB}$ up to local unitaries. However, note that $\cS_{1}',\cS_{2}'$ contain overlapping states on Bob's side, which is why we label them OBB (overlapping BB84 states). It follows that using the optimal strategy for each subset, given in the proof of Corollary \ref{cor:LOSQC-BB84}, may do quite poorly on the other set of states. In other words, one expects the optimal approximate broadcasting of $\cS_{O\BB}$ to be lower than $\cS_{\BB}$. Using our methodology, we are indeed able to prove this.
\begin{theorem}\label{thm:Overlapping-BB84}
    Consider the set of globally orthogonal product qutrit states $\cS_{O\BB}$ in \eqref{eq:overlapping-BB84-states} where $\ket{\psi_{1}}$ occurs with probability $1/4$ and the rest occur with probability $1/8$. Then 
\begin{align}\label{eq:LOSQC-c-separation-strong}
    \Pr_{\text{LOSQC}}[\cS_{OBB}] \leq 0.603554  \ .
\end{align}
\end{theorem}
We first stress that the ability to establish such results follows from our ability to work with the geometry of the initial states using Lemma \ref{lem:main-text-UR}. Moreover, we note that this result is quite strong in the following sense. By Corollary \ref{cor:LOSQC-BB84}, $\Pr_{\LOSQC}[\cS_{OBB}] < \Pr_{\LOSQC}[\cS_{\BB}]$, but $\cS_{O\BB}$ requires more resources to implement (namely, a qutrit quantum system and the ability to prepare certain superposition states). To place the two ensembles on an equal footing, we could 
measure the error `per state' in comparison to global success probability $\Pr[\cS]$:
\begin{align}
    \Delta_{\LOSQC}(\cS) &= \left(\Pr[\cS]-\Pr_{\LOSQC}[\cS]\right) /|\cS| \ . \label{eq:LOSQC-error-per-state}
\end{align}
A direct calculation will find $\Delta_{\LOSQC}(\cS_{\BB}) < 0.03662$ whereas $\Delta_{\LOSQC}(\cS_{O\BB}) > 0.05663$. Thus, by increasing to a qutrit (resp.~trit) space on Bob's (resp.~Alice's) side, we may increase the error per state, which may have practical relevance in cryptographic schemes that rely on the impossibility of orthogonality broadcasting.

Intuitively, we should be able to further amplify the result of Theorem \ref{thm:Overlapping-BB84} by giving \textit{both} Alice and Bob non-commuting states from a globally orthogonal set. One would expect under certain conditions this would amplify the result as then Alice will also need to implement approximate orthogonal broadcasting. This is motivated not only from a foundational perspective, but also the possibility of building secret sharing schemes \cite{Cleve-1999a} out of product states. 

It appears that our uncertainty relation, Lemma \ref{lem:main-text-UR} as well as its generalization, is not strong enough to place constraints to prove that making Alice's states from $\cS_{O\BB}$ in \eqref{eq:overlapping-BB84-states} quantum as well will increase the error probability specifically (See \ref{app:LOSQC-error-bounds} for further discussion), but we are able to construct a different example such that there is a non-trivial increase in the probability of error in distinguishing the states by having both parties receive quantum inputs. In particular, we consider the globally orthogonal classical-quantum product states 
\begin{align}\label{eq:cq-states}
	\cS_{cq} = \begin{Bmatrix}
		\ket{\psi_{1}} = \ket{1}\ket{1} & \ket{\psi_{2}} = \ket{1}\ket{0+2} \\ 
        \ket{\psi_{3}} = \ket{1}\ket{0-2}   &
		\ket{\psi_{4}} = \ket{0}\ket{0+1}  \\ 
        \ket{\psi_{5}} = \ket{0}\ket{0-1}  & \ket{\psi_{6}} = \ket{2}\ket{1+2}  \\
        \ket{\psi_{7}} = \ket{2}\ket{1-2}  &
	\end{Bmatrix} 
\end{align}
and the set of globally orthogonal quantum-quantum states
\begin{align}\label{eq:qq-states}
	\cS_{qq} = \begin{Bmatrix}
		\ket{\psi_{1}} = \ket{1}\ket{1} & \ket{\phi_{2}} = \ket{1-0}\ket{2} \\ 
        \ket{\psi_{3}} = \ket{0}\ket{0+1}   &
		\ket{\psi_{4}} = \ket{0}\ket{0-1}  \\ 
        \ket{\psi_{5}} = \ket{2}\ket{1+2}  & \ket{\psi_{6}} = \ket{2}\ket{1-2}  \\
        \ket{\phi_{7}} = \ket{1+2}\ket{0}  &
	\end{Bmatrix}  \ .
\end{align}
One may see that $\cS_{qq}$ is a globally orthogonal product set obtained from $\cS_{cq}$ by altering which party has the coherence for two of the states, i.e. $\ket{\psi_{2}} \to \ket{\phi_{2}}$, $\ket{\psi_{7}} \to \ket{\phi_{7}}$, which will then suffer from both parties needing to approximately broadcast (see Fig.~\ref{fig:tiling} for visual comparison).

\begin{figure}
\includegraphics[width=8cm]{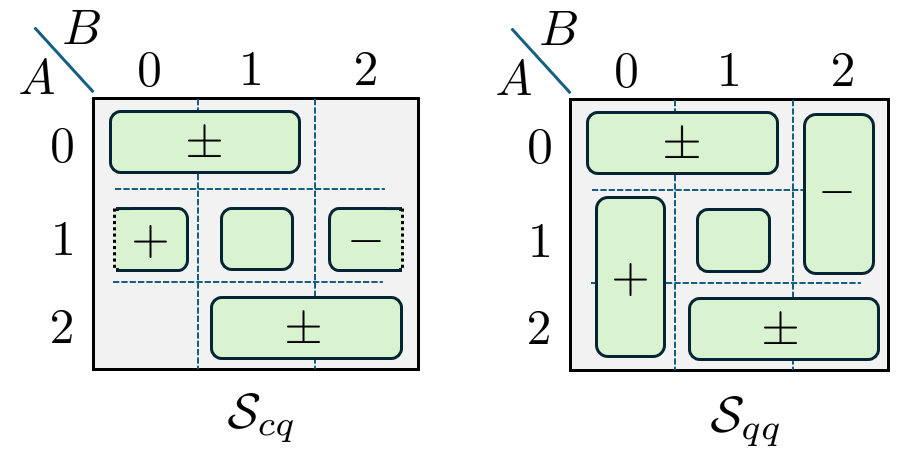}
\caption{Depiction of the sets $\cS_{cq}$ and $\cS_{qq}$ (Eqs.~\eqref{eq:cq-states} and \eqref{eq:qq-states} respectively) in terms of ``tilings" \cite{Bennett-1999b,Childs-2013a}. A state $\ket{\psi}^{AB}$ is depicted on a ``tile" $(i,j)$ if $\bra{i}^{A}\bra{j}^{B}\ket{\psi} \neq 0$, i.e. if the state has `support' on the joint computational basis state. The $\pm$ symbol denotes that two superposition states lay on the tile. The dashed tile in the left image means that Bob's state has support on both $\ket{0}$ and $\ket{2}$.}
\label{fig:tiling}
\end{figure}

The choice of classical states in $\cS_{cq}$ allows us to establish that, independent of the distribution on the states, 
$$ \Pr_{\LOSQC}[\cS_{cq}] \geq \Pr_{\text{LOSCC}}[\cS_{cq}] \geq \frac{1}{2}\left(1+ \frac{1}{\sqrt{2}}\right)  \ . $$
On the other hand, we are able to show using our semidefinite program method that the success probability of distinguishing $\cS_{qq}$ is strictly lower. While intuitive, to the best of our knowledge, this is the first proof of such an example.
\begin{theorem}\label{thm:qq-qc-separation}
    Consider $\cS_{cq}$ and $\cS_{qq}$ where in both cases $\ket{\psi_{1}}$ occurs with probability $1/4$ and the rest occur with probability $1/8$. Then
    \begin{align*}
        \Pr_{\LOSQC}[\cS_{qq}] \leq 0.7805 < \frac{1}{2}\left(1+\frac{1}{\sqrt{2}}\right) \leq \Pr_{\LOSQC}[\cS_{cq}] \ .
    \end{align*}
    That is, forcing both parties to approximately broadcast increases the probability of error.
\end{theorem}


\section{Conclusion and Outlook}
In this paper, we have introduced the study of orthogonality broadcasting. We have established a strong characterization of its properties in low dimensions, related it to the cryptographic task of quantum position verification, and established new methods for error bounds on quantum position verification protocols when there is no pre-shared entanglement. In doing so, we have introduced a new uncertainty relation that uses the geometric relation of the initial states to be broadcasted and shown both its utility and its limitations. 

We note that the ability to combine our uncertainty relation to derive semidefinite programming bounds on vectors of guessing probabilities in the LOSQC setting (e.g.~Proposition \ref{prop:SDP-version-of-Thm-4} and the derivations of Theorems \ref{thm:Overlapping-BB84} and \ref{thm:qq-qc-separation}) may be of interest in the framework of optimization programs capturing correlations. In particular, it is well-known that non-signalling correlations are quite useful for deriving linear program bounds in various settings \cite{Ito-2010a,Brandao-2013a,Fawzi-2023a}. However, a key technical issue in LOSQC is the fact that there is signaling (the simultaneous communication), which is what makes it distinct from the framework of non-local games \cite{Ito-2010a,Palazuelos-2016a}. Roughly speaking, our uncertainty relation is constraining the strength of the correlations achievable via signaling at the expense of the linear program becoming a semidefinite program.

Finally, we highlight two natural future directions based on this work. First, the study of (approximate) orthogonality broadcasting seems deeply related to approximate symmetric quantum cloning and the notion of symmetric extendability. A more in-depth study of the relation between these may be of both fundamental and practical value. Second, in Section \ref{sec:SDP-Bounds}, it was highlighted that the SDP bounds are in effect coming from the uncertainty relation bounding the correlations achievable through signaling and that the uncertainty relation can be loose under certain situations. It could be of value to develop more general methods for constraining the achievable correlations under signaling through broadcast channels.

\section*{Acknowledgments}
\noindent We thank Harry Buhrman for fruitful initial discussions in the early stages of this work. This research was conducted when R.A. and P.VL. were affiliated with CWI Amsterdam and QuSoft. R.A. was supported by the Dutch Research Council (NWO/OCW), as part of the Quantum Software Consortium programme (project number 024.003.037). The majority of this work was performed while P.VL.~was at CWI. P.VL.~was supported by the Dutch Research Council (NWO/OCW), as part of the NWO Gravitation Programme Networks (project number 024.002.003). The majority of this work was performed while I.G.~was at University of Illinois at Urbana-Champaign. E.C. and I.G.~were supported by the U.S. Department of Energy Office of Science National Quantum Information Science Research Centers. 

\newcommand{\newblock}{}
\bibliographystyle{apsrev4-2}
\bibliography{QPV.bib}{}

\begin{thebibliography}{42}%
\makeatletter
\providecommand \@ifxundefined [1]{%
 \@ifx{#1\undefined}
}%
\providecommand \@ifnum [1]{%
 \ifnum #1\expandafter \@firstoftwo
 \else \expandafter \@secondoftwo
 \fi
}%
\providecommand \@ifx [1]{%
 \ifx #1\expandafter \@firstoftwo
 \else \expandafter \@secondoftwo
 \fi
}%
\providecommand \natexlab [1]{#1}%
\providecommand \enquote  [1]{``#1''}%
\providecommand \bibnamefont  [1]{#1}%
\providecommand \bibfnamefont [1]{#1}%
\providecommand \citenamefont [1]{#1}%
\providecommand \href@noop [0]{\@secondoftwo}%
\providecommand \href [0]{\begingroup \@sanitize@url \@href}%
\providecommand \@href[1]{\@@startlink{#1}\@@href}%
\providecommand \@@href[1]{\endgroup#1\@@endlink}%
\providecommand \@sanitize@url [0]{\catcode `\\12\catcode `\$12\catcode
  `\&12\catcode `\#12\catcode `\^12\catcode `\_12\catcode `\%12\relax}%
\providecommand \@@startlink[1]{}%
\providecommand \@@endlink[0]{}%
\providecommand \url  [0]{\begingroup\@sanitize@url \@url }%
\providecommand \@url [1]{\endgroup\@href {#1}{\urlprefix }}%
\providecommand \urlprefix  [0]{URL }%
\providecommand \Eprint [0]{\href }%
\providecommand \doibase [0]{https://doi.org/}%
\providecommand \selectlanguage [0]{\@gobble}%
\providecommand \bibinfo  [0]{\@secondoftwo}%
\providecommand \bibfield  [0]{\@secondoftwo}%
\providecommand \translation [1]{[#1]}%
\providecommand \BibitemOpen [0]{}%
\providecommand \bibitemStop [0]{}%
\providecommand \bibitemNoStop [0]{.\EOS\space}%
\providecommand \EOS [0]{\spacefactor3000\relax}%
\providecommand \BibitemShut  [1]{\csname bibitem#1\endcsname}%
\let\auto@bib@innerbib\@empty
\bibitem [{\citenamefont {Wootters}\ and\ \citenamefont
  {Zurek}(1982)}]{Wootters-1982a}%
  \BibitemOpen
  \bibfield  {author} {\bibinfo {author} {\bibfnamefont {W.~K.}\ \bibnamefont
  {Wootters}}\ and\ \bibinfo {author} {\bibfnamefont {W.~H.}\ \bibnamefont
  {Zurek}},\ }\href {https://www.nature.com/articles/299802a0} {\bibfield
  {journal} {\bibinfo  {journal} {Nature}\ }\textbf {\bibinfo {volume} {299}},\
  \bibinfo {pages} {802} (\bibinfo {year} {1982})}\BibitemShut {NoStop}%
\bibitem [{\citenamefont {Broadbent}\ and\ \citenamefont
  {Lord}(2020)}]{Broadbent-2020a}%
  \BibitemOpen
  \bibfield  {author} {\bibinfo {author} {\bibfnamefont {A.}~\bibnamefont
  {Broadbent}}\ and\ \bibinfo {author} {\bibfnamefont {S.}~\bibnamefont
  {Lord}}\ }(\bibinfo  {publisher} {Schloss Dagstuhl – Leibniz-Zentrum
  f\"{u}r Informatik},\ \bibinfo {year} {2020})\BibitemShut {NoStop}%
\bibitem [{\citenamefont {Bennett}\ and\ \citenamefont
  {Brassard}(2014)}]{Bennett-2014a}%
  \BibitemOpen
  \bibfield  {author} {\bibinfo {author} {\bibfnamefont {C.~H.}\ \bibnamefont
  {Bennett}}\ and\ \bibinfo {author} {\bibfnamefont {G.}~\bibnamefont
  {Brassard}},\ }\href
  {https://doi.org/https://doi.org/10.1016/j.tcs.2014.05.025} {\bibfield
  {journal} {\bibinfo  {journal} {Theoretical Computer Science}\ }\textbf
  {\bibinfo {volume} {560}},\ \bibinfo {pages} {7} (\bibinfo {year} {2014})},\
  \bibinfo {note} {theoretical Aspects of Quantum Cryptography – celebrating
  30 years of BB84}\BibitemShut {NoStop}%
\bibitem [{\citenamefont {Cleve}\ \emph {et~al.}(1999)\citenamefont {Cleve},
  \citenamefont {Gottesman},\ and\ \citenamefont {Lo}}]{Cleve-1999a}%
  \BibitemOpen
  \bibfield  {author} {\bibinfo {author} {\bibfnamefont {R.}~\bibnamefont
  {Cleve}}, \bibinfo {author} {\bibfnamefont {D.}~\bibnamefont {Gottesman}},\
  and\ \bibinfo {author} {\bibfnamefont {H.-K.}\ \bibnamefont {Lo}},\ }\href
  {https://doi.org/10.1103/PhysRevLett.83.648} {\bibfield  {journal} {\bibinfo
  {journal} {Phys. Rev. Lett.}\ }\textbf {\bibinfo {volume} {83}},\ \bibinfo
  {pages} {648} (\bibinfo {year} {1999})}\BibitemShut {NoStop}%
\bibitem [{\citenamefont {Kent}\ \emph {et~al.}(2011)\citenamefont {Kent},
  \citenamefont {Munro},\ and\ \citenamefont {Spiller}}]{Kent-2011a}%
  \BibitemOpen
  \bibfield  {author} {\bibinfo {author} {\bibfnamefont {A.}~\bibnamefont
  {Kent}}, \bibinfo {author} {\bibfnamefont {W.~J.}\ \bibnamefont {Munro}},\
  and\ \bibinfo {author} {\bibfnamefont {T.~P.}\ \bibnamefont {Spiller}},\
  }\href@noop {} {\bibfield  {journal} {\bibinfo  {journal} {Physical Review
  A}\ }\textbf {\bibinfo {volume} {84}},\ \bibinfo {pages} {012326} (\bibinfo
  {year} {2011})}\BibitemShut {NoStop}%
\bibitem [{\citenamefont {Barnum}\ \emph {et~al.}(1996)\citenamefont {Barnum},
  \citenamefont {Caves}, \citenamefont {Fuchs}, \citenamefont {Jozsa},\ and\
  \citenamefont {Schumacher}}]{Barnum-1996a}%
  \BibitemOpen
  \bibfield  {author} {\bibinfo {author} {\bibfnamefont {H.}~\bibnamefont
  {Barnum}}, \bibinfo {author} {\bibfnamefont {C.~M.}\ \bibnamefont {Caves}},
  \bibinfo {author} {\bibfnamefont {C.~A.}\ \bibnamefont {Fuchs}}, \bibinfo
  {author} {\bibfnamefont {R.}~\bibnamefont {Jozsa}},\ and\ \bibinfo {author}
  {\bibfnamefont {B.}~\bibnamefont {Schumacher}},\ }\href
  {https://doi.org/10.1103/PhysRevLett.76.2818} {\bibfield  {journal} {\bibinfo
   {journal} {Phys. Rev. Lett.}\ }\textbf {\bibinfo {volume} {76}},\ \bibinfo
  {pages} {2818} (\bibinfo {year} {1996})}\BibitemShut {NoStop}%
\bibitem [{\citenamefont {Lindblad}(1999)}]{Lindblad-1999a}%
  \BibitemOpen
  \bibfield  {author} {\bibinfo {author} {\bibfnamefont {G.}~\bibnamefont
  {Lindblad}},\ }\href {https://doi.org/10.1023/a:1007581027660} {\bibfield
  {journal} {\bibinfo  {journal} {Letters in Mathematical Physics}\ }\textbf
  {\bibinfo {volume} {47}},\ \bibinfo {pages} {189–196} (\bibinfo {year}
  {1999})}\BibitemShut {NoStop}%
\bibitem [{\citenamefont {Kalev}\ and\ \citenamefont
  {Hen}(2008)}]{Kalev-2008a}%
  \BibitemOpen
  \bibfield  {author} {\bibinfo {author} {\bibfnamefont {A.}~\bibnamefont
  {Kalev}}\ and\ \bibinfo {author} {\bibfnamefont {I.}~\bibnamefont {Hen}},\
  }\href {https://doi.org/10.1103/PhysRevLett.100.210502} {\bibfield  {journal}
  {\bibinfo  {journal} {Phys. Rev. Lett.}\ }\textbf {\bibinfo {volume} {100}},\
  \bibinfo {pages} {210502} (\bibinfo {year} {2008})}\BibitemShut {NoStop}%
\bibitem [{\citenamefont {Barnum}\ \emph {et~al.}(2007)\citenamefont {Barnum},
  \citenamefont {Barrett}, \citenamefont {Leifer},\ and\ \citenamefont
  {Wilce}}]{Barnum-2007a}%
  \BibitemOpen
  \bibfield  {author} {\bibinfo {author} {\bibfnamefont {H.}~\bibnamefont
  {Barnum}}, \bibinfo {author} {\bibfnamefont {J.}~\bibnamefont {Barrett}},
  \bibinfo {author} {\bibfnamefont {M.}~\bibnamefont {Leifer}},\ and\ \bibinfo
  {author} {\bibfnamefont {A.}~\bibnamefont {Wilce}},\ }\href
  {https://doi.org/10.1103/PhysRevLett.99.240501} {\bibfield  {journal}
  {\bibinfo  {journal} {Phys. Rev. Lett.}\ }\textbf {\bibinfo {volume} {99}},\
  \bibinfo {pages} {240501} (\bibinfo {year} {2007})}\BibitemShut {NoStop}%
\bibitem [{\citenamefont {Ballester}\ \emph {et~al.}(2008)\citenamefont
  {Ballester}, \citenamefont {Wehner},\ and\ \citenamefont
  {Winter}}]{Ballester-2008a}%
  \BibitemOpen
  \bibfield  {author} {\bibinfo {author} {\bibfnamefont {M.~A.}\ \bibnamefont
  {Ballester}}, \bibinfo {author} {\bibfnamefont {S.}~\bibnamefont {Wehner}},\
  and\ \bibinfo {author} {\bibfnamefont {A.}~\bibnamefont {Winter}},\ }\href
  {https://doi.org/10.1109/TIT.2008.928276} {\bibfield  {journal} {\bibinfo
  {journal} {IEEE Transactions on Information Theory}\ }\textbf {\bibinfo
  {volume} {54}},\ \bibinfo {pages} {4183} (\bibinfo {year}
  {2008})}\BibitemShut {NoStop}%
\bibitem [{\citenamefont {Carmeli}\ \emph {et~al.}(2022)\citenamefont
  {Carmeli}, \citenamefont {Heinosaari},\ and\ \citenamefont
  {Toigo}}]{Carmeli-2022a}%
  \BibitemOpen
  \bibfield  {author} {\bibinfo {author} {\bibfnamefont {C.}~\bibnamefont
  {Carmeli}}, \bibinfo {author} {\bibfnamefont {T.}~\bibnamefont
  {Heinosaari}},\ and\ \bibinfo {author} {\bibfnamefont {A.}~\bibnamefont
  {Toigo}},\ }\href {https://doi.org/10.1088/1361-6633/ac6f0e} {\bibfield
  {journal} {\bibinfo  {journal} {Reports on Progress in Physics}\ }\textbf
  {\bibinfo {volume} {85}},\ \bibinfo {pages} {074001} (\bibinfo {year}
  {2022})}\BibitemShut {NoStop}%
\bibitem [{\citenamefont {Sattath}(2023)}]{Sattath-2023a}%
  \BibitemOpen
  \bibfield  {author} {\bibinfo {author} {\bibfnamefont {O.}~\bibnamefont
  {Sattath}},\ }\href@noop {} {\bibfield  {journal} {\bibinfo  {journal}
  {Communications of the ACM}\ }\textbf {\bibinfo {volume} {66}},\ \bibinfo
  {pages} {78} (\bibinfo {year} {2023})}\BibitemShut {NoStop}%
\bibitem [{\citenamefont {Tomamichel}\ \emph {et~al.}(2013)\citenamefont
  {Tomamichel}, \citenamefont {Fehr}, \citenamefont {Kaniewski},\ and\
  \citenamefont {Wehner}}]{Tomamichel-2013a}%
  \BibitemOpen
  \bibfield  {author} {\bibinfo {author} {\bibfnamefont {M.}~\bibnamefont
  {Tomamichel}}, \bibinfo {author} {\bibfnamefont {S.}~\bibnamefont {Fehr}},
  \bibinfo {author} {\bibfnamefont {J.}~\bibnamefont {Kaniewski}},\ and\
  \bibinfo {author} {\bibfnamefont {S.}~\bibnamefont {Wehner}},\ }\href@noop {}
  {\bibfield  {journal} {\bibinfo  {journal} {New Journal of Physics}\ }\textbf
  {\bibinfo {volume} {15}},\ \bibinfo {pages} {103002} (\bibinfo {year}
  {2013})}\BibitemShut {NoStop}%
\bibitem [{Bro(2023)}]{Broadbent-2023a}%
  \BibitemOpen
  \ (\bibinfo  {publisher} {Schloss Dagstuhl – Leibniz-Zentrum f\"{u}r
  Informatik},\ \bibinfo {year} {2023})\BibitemShut {NoStop}%
\bibitem [{\citenamefont {Ananth}\ \emph {et~al.}(2023)\citenamefont {Ananth},
  \citenamefont {Kaleoglu},\ and\ \citenamefont {Liu}}]{Ananth-2023a}%
  \BibitemOpen
  \bibfield  {author} {\bibinfo {author} {\bibfnamefont {P.}~\bibnamefont
  {Ananth}}, \bibinfo {author} {\bibfnamefont {F.}~\bibnamefont {Kaleoglu}},\
  and\ \bibinfo {author} {\bibfnamefont {Q.}~\bibnamefont {Liu}},\ }in\
  \href@noop {} {\emph {\bibinfo {booktitle} {Annual International Cryptology
  Conference}}}\ (\bibinfo {organization} {Springer},\ \bibinfo {year} {2023})\
  pp.\ \bibinfo {pages} {66--98}\BibitemShut {NoStop}%
\bibitem [{\citenamefont {Goyal}\ \emph {et~al.}(2024)\citenamefont {Goyal},
  \citenamefont {Malavolta},\ and\ \citenamefont {Raizes}}]{Goyal-2024a}%
  \BibitemOpen
  \bibfield  {author} {\bibinfo {author} {\bibfnamefont {V.}~\bibnamefont
  {Goyal}}, \bibinfo {author} {\bibfnamefont {G.}~\bibnamefont {Malavolta}},\
  and\ \bibinfo {author} {\bibfnamefont {J.}~\bibnamefont {Raizes}},\ }in\
  \href@noop {} {\emph {\bibinfo {booktitle} {Theory of Cryptography
  Conference}}}\ (\bibinfo {organization} {Springer},\ \bibinfo {year} {2024})\
  pp.\ \bibinfo {pages} {193--224}\BibitemShut {NoStop}%
\bibitem [{\citenamefont {Poremba}\ \emph {et~al.}(2024)\citenamefont
  {Poremba}, \citenamefont {Ragavan},\ and\ \citenamefont
  {Vaikuntanathan}}]{Poremba-2024a}%
  \BibitemOpen
  \bibfield  {author} {\bibinfo {author} {\bibfnamefont {A.}~\bibnamefont
  {Poremba}}, \bibinfo {author} {\bibfnamefont {S.}~\bibnamefont {Ragavan}},\
  and\ \bibinfo {author} {\bibfnamefont {V.}~\bibnamefont {Vaikuntanathan}},\
  }\href {https://arxiv.org/abs/2411.04730} {\bibinfo {title} {Cloning games,
  black holes and cryptography}} (\bibinfo {year} {2024}),\ \Eprint
  {https://arxiv.org/abs/2411.04730} {arXiv:2411.04730 [quant-ph]} \BibitemShut
  {NoStop}%
\bibitem [{\citenamefont {Escolà-Farràs}\ and\ \citenamefont
  {Speelman}(2024)}]{Escola-2024lossy}%
  \BibitemOpen
  \bibfield  {author} {\bibinfo {author} {\bibfnamefont {L.}~\bibnamefont
  {Escolà-Farràs}}\ and\ \bibinfo {author} {\bibfnamefont {F.}~\bibnamefont
  {Speelman}},\ }\href {https://arxiv.org/abs/2405.13717} {\bibinfo {title}
  {Lossy-and-constrained extended non-local games with applications to
  cryptography: Bc, qkd and qpv}} (\bibinfo {year} {2024}),\ \Eprint
  {https://arxiv.org/abs/2405.13717} {arXiv:2405.13717 [quant-ph]} \BibitemShut
  {NoStop}%
\bibitem [{\citenamefont {Allerstorfer}\ \emph {et~al.}(2022)\citenamefont
  {Allerstorfer}, \citenamefont {Buhrman}, \citenamefont {Speelman},\ and\
  \citenamefont {{Verduyn Lunel}}}]{Allerstorfer-2022a}%
  \BibitemOpen
  \bibfield  {author} {\bibinfo {author} {\bibfnamefont {R.}~\bibnamefont
  {Allerstorfer}}, \bibinfo {author} {\bibfnamefont {H.}~\bibnamefont
  {Buhrman}}, \bibinfo {author} {\bibfnamefont {F.}~\bibnamefont {Speelman}},\
  and\ \bibinfo {author} {\bibfnamefont {P.}~\bibnamefont {{Verduyn Lunel}}},\
  }\href {https://doi.org/10.48550/ARXIV.2208.04341} {\bibinfo {title} {On the
  role of quantum communication and loss in attacks on quantum position
  verification}} (\bibinfo {year} {2022})\BibitemShut {NoStop}%
\bibitem [{\citenamefont {Deutsch}(1983)}]{Deutsch-1983a}%
  \BibitemOpen
  \bibfield  {author} {\bibinfo {author} {\bibfnamefont {D.}~\bibnamefont
  {Deutsch}},\ }\href {https://doi.org/10.1103/PhysRevLett.50.631} {\bibfield
  {journal} {\bibinfo  {journal} {Phys. Rev. Lett.}\ }\textbf {\bibinfo
  {volume} {50}},\ \bibinfo {pages} {631} (\bibinfo {year} {1983})}\BibitemShut
  {NoStop}%
\bibitem [{\citenamefont {Renes}\ and\ \citenamefont
  {Boileau}(2009)}]{Renes-2009a}%
  \BibitemOpen
  \bibfield  {author} {\bibinfo {author} {\bibfnamefont {J.~M.}\ \bibnamefont
  {Renes}}\ and\ \bibinfo {author} {\bibfnamefont {J.-C.}\ \bibnamefont
  {Boileau}},\ }\href {https://doi.org/10.1103/PhysRevLett.103.020402}
  {\bibfield  {journal} {\bibinfo  {journal} {Phys. Rev. Lett.}\ }\textbf
  {\bibinfo {volume} {103}},\ \bibinfo {pages} {020402} (\bibinfo {year}
  {2009})}\BibitemShut {NoStop}%
\bibitem [{\citenamefont {Berta}\ \emph {et~al.}(2010)\citenamefont {Berta},
  \citenamefont {Christandl}, \citenamefont {Colbeck}, \citenamefont {Renes},\
  and\ \citenamefont {Renner}}]{Berta-2010a}%
  \BibitemOpen
  \bibfield  {author} {\bibinfo {author} {\bibfnamefont {M.}~\bibnamefont
  {Berta}}, \bibinfo {author} {\bibfnamefont {M.}~\bibnamefont {Christandl}},
  \bibinfo {author} {\bibfnamefont {R.}~\bibnamefont {Colbeck}}, \bibinfo
  {author} {\bibfnamefont {J.~M.}\ \bibnamefont {Renes}},\ and\ \bibinfo
  {author} {\bibfnamefont {R.}~\bibnamefont {Renner}},\ }\href
  {https://doi.org/10.1038/nphys1734} {\bibfield  {journal} {\bibinfo
  {journal} {Nature Physics}\ }\textbf {\bibinfo {volume} {6}},\ \bibinfo
  {pages} {659–662} (\bibinfo {year} {2010})}\BibitemShut {NoStop}%
\bibitem [{\citenamefont {Coles}\ \emph {et~al.}(2017)\citenamefont {Coles},
  \citenamefont {Berta}, \citenamefont {Tomamichel},\ and\ \citenamefont
  {Wehner}}]{Coles-2017a}%
  \BibitemOpen
  \bibfield  {author} {\bibinfo {author} {\bibfnamefont {P.~J.}\ \bibnamefont
  {Coles}}, \bibinfo {author} {\bibfnamefont {M.}~\bibnamefont {Berta}},
  \bibinfo {author} {\bibfnamefont {M.}~\bibnamefont {Tomamichel}},\ and\
  \bibinfo {author} {\bibfnamefont {S.}~\bibnamefont {Wehner}},\ }\href
  {https://doi.org/10.1103/RevModPhys.89.015002} {\bibfield  {journal}
  {\bibinfo  {journal} {Rev. Mod. Phys.}\ }\textbf {\bibinfo {volume} {89}},\
  \bibinfo {pages} {015002} (\bibinfo {year} {2017})}\BibitemShut {NoStop}%
\bibitem [{\citenamefont {Cubitt}\ \emph {et~al.}(2003)\citenamefont {Cubitt},
  \citenamefont {Verstraete}, \citenamefont {D\"ur},\ and\ \citenamefont
  {Cirac}}]{Cubitt-2003a}%
  \BibitemOpen
  \bibfield  {author} {\bibinfo {author} {\bibfnamefont {T.~S.}\ \bibnamefont
  {Cubitt}}, \bibinfo {author} {\bibfnamefont {F.}~\bibnamefont {Verstraete}},
  \bibinfo {author} {\bibfnamefont {W.}~\bibnamefont {D\"ur}},\ and\ \bibinfo
  {author} {\bibfnamefont {J.~I.}\ \bibnamefont {Cirac}},\ }\href
  {https://doi.org/10.1103/PhysRevLett.91.037902} {\bibfield  {journal}
  {\bibinfo  {journal} {Phys. Rev. Lett.}\ }\textbf {\bibinfo {volume} {91}},\
  \bibinfo {pages} {037902} (\bibinfo {year} {2003})}\BibitemShut {NoStop}%
\bibitem [{\citenamefont {Chuan}\ \emph {et~al.}(2012)\citenamefont {Chuan},
  \citenamefont {Maillard}, \citenamefont {Modi}, \citenamefont {Paterek},
  \citenamefont {Paternostro},\ and\ \citenamefont {Piani}}]{Chuan-2012a}%
  \BibitemOpen
  \bibfield  {author} {\bibinfo {author} {\bibfnamefont {T.~K.}\ \bibnamefont
  {Chuan}}, \bibinfo {author} {\bibfnamefont {J.}~\bibnamefont {Maillard}},
  \bibinfo {author} {\bibfnamefont {K.}~\bibnamefont {Modi}}, \bibinfo {author}
  {\bibfnamefont {T.}~\bibnamefont {Paterek}}, \bibinfo {author} {\bibfnamefont
  {M.}~\bibnamefont {Paternostro}},\ and\ \bibinfo {author} {\bibfnamefont
  {M.}~\bibnamefont {Piani}},\ }\href
  {https://doi.org/10.1103/PhysRevLett.109.070501} {\bibfield  {journal}
  {\bibinfo  {journal} {Phys. Rev. Lett.}\ }\textbf {\bibinfo {volume} {109}},\
  \bibinfo {pages} {070501} (\bibinfo {year} {2012})}\BibitemShut {NoStop}%
\bibitem [{\citenamefont {Streltsov}\ \emph {et~al.}(2012)\citenamefont
  {Streltsov}, \citenamefont {Kampermann},\ and\ \citenamefont
  {Bru\ss{}}}]{Streltsov-2012a}%
  \BibitemOpen
  \bibfield  {author} {\bibinfo {author} {\bibfnamefont {A.}~\bibnamefont
  {Streltsov}}, \bibinfo {author} {\bibfnamefont {H.}~\bibnamefont
  {Kampermann}},\ and\ \bibinfo {author} {\bibfnamefont {D.}~\bibnamefont
  {Bru\ss{}}},\ }\href {https://doi.org/10.1103/PhysRevLett.108.250501}
  {\bibfield  {journal} {\bibinfo  {journal} {Phys. Rev. Lett.}\ }\textbf
  {\bibinfo {volume} {108}},\ \bibinfo {pages} {250501} (\bibinfo {year}
  {2012})}\BibitemShut {NoStop}%
\bibitem [{\citenamefont {Bu\ifmmode~\check{z}\else \v{z}\fi{}ek}\ and\
  \citenamefont {Hillery}(1996)}]{Buzek-1996a}%
  \BibitemOpen
  \bibfield  {author} {\bibinfo {author} {\bibfnamefont {V.}~\bibnamefont
  {Bu\ifmmode~\check{z}\else \v{z}\fi{}ek}}\ and\ \bibinfo {author}
  {\bibfnamefont {M.}~\bibnamefont {Hillery}},\ }\href
  {https://doi.org/10.1103/PhysRevA.54.1844} {\bibfield  {journal} {\bibinfo
  {journal} {Phys. Rev. A}\ }\textbf {\bibinfo {volume} {54}},\ \bibinfo
  {pages} {1844} (\bibinfo {year} {1996})}\BibitemShut {NoStop}%
\bibitem [{\citenamefont {Scarani}\ \emph {et~al.}(2005)\citenamefont
  {Scarani}, \citenamefont {Iblisdir}, \citenamefont {Gisin},\ and\
  \citenamefont {Ac\'{\i}n}}]{Scarani-2005a}%
  \BibitemOpen
  \bibfield  {author} {\bibinfo {author} {\bibfnamefont {V.}~\bibnamefont
  {Scarani}}, \bibinfo {author} {\bibfnamefont {S.}~\bibnamefont {Iblisdir}},
  \bibinfo {author} {\bibfnamefont {N.}~\bibnamefont {Gisin}},\ and\ \bibinfo
  {author} {\bibfnamefont {A.}~\bibnamefont {Ac\'{\i}n}},\ }\href
  {https://doi.org/10.1103/RevModPhys.77.1225} {\bibfield  {journal} {\bibinfo
  {journal} {Rev. Mod. Phys.}\ }\textbf {\bibinfo {volume} {77}},\ \bibinfo
  {pages} {1225} (\bibinfo {year} {2005})}\BibitemShut {NoStop}%
\bibitem [{\citenamefont {Watrous}(2018)}]{WatrousBook}%
  \BibitemOpen
  \bibfield  {author} {\bibinfo {author} {\bibfnamefont {J.}~\bibnamefont
  {Watrous}},\ }\href@noop {} {\emph {\bibinfo {title} {The Theory of Quantum
  Information}}}\ (\bibinfo  {publisher} {Cambridge University Press},\
  \bibinfo {year} {2018})\ \Eprint
  {https://arxiv.org/abs/https://cs.uwaterloo.ca/~watrous/TQI/}
  {https://cs.uwaterloo.ca/~watrous/TQI/} \BibitemShut {NoStop}%
\bibitem [{\citenamefont {Buhrman}\ \emph {et~al.}(2014)\citenamefont
  {Buhrman}, \citenamefont {Chandran}, \citenamefont {Fehr}, \citenamefont
  {Gelles}, \citenamefont {Goyal}, \citenamefont {Ostrovsky},\ and\
  \citenamefont {Schaffner}}]{Buhrman-2014a}%
  \BibitemOpen
  \bibfield  {author} {\bibinfo {author} {\bibfnamefont {H.}~\bibnamefont
  {Buhrman}}, \bibinfo {author} {\bibfnamefont {N.}~\bibnamefont {Chandran}},
  \bibinfo {author} {\bibfnamefont {S.}~\bibnamefont {Fehr}}, \bibinfo {author}
  {\bibfnamefont {R.}~\bibnamefont {Gelles}}, \bibinfo {author} {\bibfnamefont
  {V.}~\bibnamefont {Goyal}}, \bibinfo {author} {\bibfnamefont
  {R.}~\bibnamefont {Ostrovsky}},\ and\ \bibinfo {author} {\bibfnamefont
  {C.}~\bibnamefont {Schaffner}},\ }\href
  {https://epubs.siam.org/doi/10.1137/130913687} {\bibfield  {journal}
  {\bibinfo  {journal} {SIAM Journal on Computing}\ }\textbf {\bibinfo {volume}
  {43}},\ \bibinfo {pages} {150} (\bibinfo {year} {2014})}\BibitemShut
  {NoStop}%
\bibitem [{\citenamefont {Bennett}\ \emph
  {et~al.}(1999{\natexlab{a}})\citenamefont {Bennett}, \citenamefont
  {DiVincenzo}, \citenamefont {Mor}, \citenamefont {Shor}, \citenamefont
  {Smolin},\ and\ \citenamefont {Terhal}}]{Bennett-1999a}%
  \BibitemOpen
  \bibfield  {author} {\bibinfo {author} {\bibfnamefont {C.~H.}\ \bibnamefont
  {Bennett}}, \bibinfo {author} {\bibfnamefont {D.~P.}\ \bibnamefont
  {DiVincenzo}}, \bibinfo {author} {\bibfnamefont {T.}~\bibnamefont {Mor}},
  \bibinfo {author} {\bibfnamefont {P.~W.}\ \bibnamefont {Shor}}, \bibinfo
  {author} {\bibfnamefont {J.~A.}\ \bibnamefont {Smolin}},\ and\ \bibinfo
  {author} {\bibfnamefont {B.~M.}\ \bibnamefont {Terhal}},\ }\href
  {https://journals.aps.org/prl/abstract/10.1103/PhysRevLett.82.5385}
  {\bibfield  {journal} {\bibinfo  {journal} {Physical Review Letters}\
  }\textbf {\bibinfo {volume} {82}},\ \bibinfo {pages} {5385} (\bibinfo {year}
  {1999}{\natexlab{a}})}\BibitemShut {NoStop}%
\bibitem [{\citenamefont {Bennett}\ \emph
  {et~al.}(1999{\natexlab{b}})\citenamefont {Bennett}, \citenamefont
  {DiVincenzo}, \citenamefont {Fuchs}, \citenamefont {Mor}, \citenamefont
  {Rains}, \citenamefont {Shor}, \citenamefont {Smolin},\ and\ \citenamefont
  {Wootters}}]{Bennett-1999b}%
  \BibitemOpen
  \bibfield  {author} {\bibinfo {author} {\bibfnamefont {C.~H.}\ \bibnamefont
  {Bennett}}, \bibinfo {author} {\bibfnamefont {D.~P.}\ \bibnamefont
  {DiVincenzo}}, \bibinfo {author} {\bibfnamefont {C.~A.}\ \bibnamefont
  {Fuchs}}, \bibinfo {author} {\bibfnamefont {T.}~\bibnamefont {Mor}}, \bibinfo
  {author} {\bibfnamefont {E.}~\bibnamefont {Rains}}, \bibinfo {author}
  {\bibfnamefont {P.~W.}\ \bibnamefont {Shor}}, \bibinfo {author}
  {\bibfnamefont {J.~A.}\ \bibnamefont {Smolin}},\ and\ \bibinfo {author}
  {\bibfnamefont {W.~K.}\ \bibnamefont {Wootters}},\ }\href
  {https://doi.org/10.1103/PhysRevA.59.1070} {\bibfield  {journal} {\bibinfo
  {journal} {Phys. Rev. A}\ }\textbf {\bibinfo {volume} {59}},\ \bibinfo
  {pages} {1070} (\bibinfo {year} {1999}{\natexlab{b}})}\BibitemShut {NoStop}%
\bibitem [{\citenamefont {Childs}\ \emph {et~al.}(2013)\citenamefont {Childs},
  \citenamefont {Leung}, \citenamefont {Mančinska},\ and\ \citenamefont
  {Ozols}}]{Childs-2013a}%
  \BibitemOpen
  \bibfield  {author} {\bibinfo {author} {\bibfnamefont {A.~M.}\ \bibnamefont
  {Childs}}, \bibinfo {author} {\bibfnamefont {D.}~\bibnamefont {Leung}},
  \bibinfo {author} {\bibfnamefont {L.}~\bibnamefont {Mančinska}},\ and\
  \bibinfo {author} {\bibfnamefont {M.}~\bibnamefont {Ozols}},\ }\href
  {https://doi.org/10.1007/s00220-013-1784-0} {\bibfield  {journal} {\bibinfo
  {journal} {Communications in Mathematical Physics}\ }\textbf {\bibinfo
  {volume} {323}},\ \bibinfo {pages} {1121–1153} (\bibinfo {year}
  {2013})}\BibitemShut {NoStop}%
\bibitem [{\citenamefont {Ito}(2010)}]{Ito-2010a}%
  \BibitemOpen
  \bibfield  {author} {\bibinfo {author} {\bibfnamefont {T.}~\bibnamefont
  {Ito}},\ }in\ \href {https://arxiv.org/abs/0908.2363} {\emph {\bibinfo
  {booktitle} {Automata, Languages and Programming: 37th International
  Colloquium, ICALP 2010, Bordeaux, France, July 6-10, 2010, Proceedings, Part
  I 37}}}\ (\bibinfo {organization} {Springer},\ \bibinfo {year} {2010})\ pp.\
  \bibinfo {pages} {140--151}\BibitemShut {NoStop}%
\bibitem [{\citenamefont {Brandao}\ and\ \citenamefont
  {Harrow}(2017)}]{Brandao-2013a}%
  \BibitemOpen
  \bibfield  {author} {\bibinfo {author} {\bibfnamefont {F.~G.}\ \bibnamefont
  {Brandao}}\ and\ \bibinfo {author} {\bibfnamefont {A.~W.}\ \bibnamefont
  {Harrow}},\ }\href
  {https://link.springer.com/article/10.1007/s00220-017-2880-3} {\bibfield
  {journal} {\bibinfo  {journal} {Communications in Mathematical Physics}\ ,\
  \bibinfo {pages} {469}} (\bibinfo {year} {2017})}\BibitemShut {NoStop}%
\bibitem [{\citenamefont {Fawzi}\ and\ \citenamefont
  {Ferm{\'e}}(2023)}]{Fawzi-2023a}%
  \BibitemOpen
  \bibfield  {author} {\bibinfo {author} {\bibfnamefont {O.}~\bibnamefont
  {Fawzi}}\ and\ \bibinfo {author} {\bibfnamefont {P.}~\bibnamefont
  {Ferm{\'e}}},\ }\href {https://ieeexplore.ieee.org/document/10213279}
  {\bibfield  {journal} {\bibinfo  {journal} {IEEE Transactions on Information
  Theory}\ } (\bibinfo {year} {2023})}\BibitemShut {NoStop}%
\bibitem [{\citenamefont {Palazuelos}\ and\ \citenamefont
  {Vidick}(2016)}]{Palazuelos-2016a}%
  \BibitemOpen
  \bibfield  {author} {\bibinfo {author} {\bibfnamefont {C.}~\bibnamefont
  {Palazuelos}}\ and\ \bibinfo {author} {\bibfnamefont {T.}~\bibnamefont
  {Vidick}},\ }\href
  {https://pubs.aip.org/aip/jmp/article-abstract/57/1/015220/910449/Survey-on-nonlocal-games-and-operator-space-theory}
  {\bibfield  {journal} {\bibinfo  {journal} {Journal of Mathematical Physics}\
  }\textbf {\bibinfo {volume} {57}} (\bibinfo {year} {2016})}\BibitemShut
  {NoStop}%
\bibitem [{\citenamefont {Wilde}(2013)}]{Wilde-2011a}%
  \BibitemOpen
  \bibfield  {author} {\bibinfo {author} {\bibfnamefont {M.~M.}\ \bibnamefont
  {Wilde}},\ }\href {https://doi.org/10.1017/CBO9781139525343} {\emph {\bibinfo
  {title} {Quantum information theory}}}\ (\bibinfo  {publisher} {Cambridge
  University Press},\ \bibinfo {year} {2013})\BibitemShut {NoStop}%
\bibitem [{\citenamefont {Schumacher}\ and\ \citenamefont
  {Westmoreland}(2010)}]{Schumacher-2010a}%
  \BibitemOpen
  \bibfield  {author} {\bibinfo {author} {\bibfnamefont {B.}~\bibnamefont
  {Schumacher}}\ and\ \bibinfo {author} {\bibfnamefont {M.}~\bibnamefont
  {Westmoreland}},\ }\href@noop {} {\emph {\bibinfo {title} {Quantum processes
  systems, and information}}}\ (\bibinfo  {publisher} {Cambridge University
  Press},\ \bibinfo {year} {2010})\BibitemShut {NoStop}%
\bibitem [{\citenamefont {Grant}\ \emph {et~al.}(2006)\citenamefont {Grant},
  \citenamefont {Boyd},\ and\ \citenamefont {Ye}}]{Grant-2006a}%
  \BibitemOpen
  \bibfield  {author} {\bibinfo {author} {\bibfnamefont {M.}~\bibnamefont
  {Grant}}, \bibinfo {author} {\bibfnamefont {S.}~\bibnamefont {Boyd}},\ and\
  \bibinfo {author} {\bibfnamefont {Y.}~\bibnamefont {Ye}},\ }in\ \href
  {https://link.springer.com/book/10.1007/0-387-30528-9} {\emph {\bibinfo
  {booktitle} {Global Optimization: From Theory to Implementation}}}\ (\bibinfo
   {publisher} {Springer},\ \bibinfo {year} {2006})\ pp.\ \bibinfo {pages}
  {155--210}\BibitemShut {NoStop}%
\bibitem [{\citenamefont {Diamond}\ and\ \citenamefont
  {Boyd}(2016)}]{Diamond-2016a}%
  \BibitemOpen
  \bibfield  {author} {\bibinfo {author} {\bibfnamefont {S.}~\bibnamefont
  {Diamond}}\ and\ \bibinfo {author} {\bibfnamefont {S.}~\bibnamefont {Boyd}},\
  }\href@noop {} {\bibfield  {journal} {\bibinfo  {journal} {Journal of Machine
  Learning Research}\ }\textbf {\bibinfo {volume} {17}},\ \bibinfo {pages} {1}
  (\bibinfo {year} {2016})}\BibitemShut {NoStop}%
\bibitem [{\citenamefont {Agrawal}\ \emph {et~al.}(2018)\citenamefont
  {Agrawal}, \citenamefont {Verschueren}, \citenamefont {Diamond},\ and\
  \citenamefont {Boyd}}]{Agrawal-2018a}%
  \BibitemOpen
  \bibfield  {author} {\bibinfo {author} {\bibfnamefont {A.}~\bibnamefont
  {Agrawal}}, \bibinfo {author} {\bibfnamefont {R.}~\bibnamefont
  {Verschueren}}, \bibinfo {author} {\bibfnamefont {S.}~\bibnamefont
  {Diamond}},\ and\ \bibinfo {author} {\bibfnamefont {S.}~\bibnamefont
  {Boyd}},\ }\href@noop {} {\bibfield  {journal} {\bibinfo  {journal} {Journal
  of Control and Decision}\ }\textbf {\bibinfo {volume} {5}},\ \bibinfo {pages}
  {42} (\bibinfo {year} {2018})}\BibitemShut {NoStop}%
\end{thebibliography}%

\appendix

\section{Monogamy-of-Entanglement Games, Quantum Position Verification, and Orthogonality Broadcasting}\label{app:ob-cloning-games-and-MoE}
In this section, we clarify the relation between monogamy-of-entanglement (MoE) games, quantum position verification (QPV), and orthogonality broadcasting. In particular, we show why MoE games are both a less direct reduction for proving security of QPV in the no-PE model and at times cannot establish security while the orthogonality broadcasting methodology can.

\subsection*{A Review on Monogamy-of-Entanglement (MoE) Games} 
We begin with the description of a MoE game (see Fig.~\ref{fig:MoE-game} for depiction), the physical idea it operationally captures, and the primary tool used to analyze MoE games \cite[Lemma 2]{Tomamichel-2013a}. For clarity, we describe the procedure with references to the subsequent formal definitions.

As depicted in Fig.~\ref{fig:MoE-game}, in a MoE game $G$ (Definition \ref{def:MoE-game}), Alice's measurements $\{\cM^{\theta} = \{F^{\theta}_{x}\}\}_{\theta}$ are publicly known and fix the dimension of the quantum system she is to receive. Bob and Charlie prepare a quantum state $\rho_{ABC}$ according to their strategy $\cS$ (Definition \ref{def:MoE-strategy}). They then forward the $A$ system to Alice. Alice draws $\theta$ from $\Theta$ uniformly at random and applies the measurement $\cM^{\theta} = \{F^{\theta}_{x}\}$ to the $A$ system, obtaining outcome $x_{A} \in \cX$. Bob and Charlie therefore share the conditional state 
\begin{align}
	\rho^{BC}_{x_{A} \vert \theta} \coloneq \frac{\Tr_{A}[F^{\theta}_{x_{A}} \otimes \mbbm{1}^{B} \otimes \mbbm{1}^{C} \rho^{ABC}]}{\Tr[F^{\theta}_{x_{A}}\rho^{A}]} \ .
\end{align} 
Alice then announces $\theta$ to Bob and Charlie who then apply their corresponding measurements $\{P_{x}^{\theta}\}_{x \in \cX}$, $\{Q_{x}^{\theta}\}_{x \in \cX}$ as specified by their strategy $\cS$ (Definition \ref{def:MoE-strategy}). Bob and Charlie's measurement outcomes are $x_{B},x_{C}$ respectively. Bob and Charlie `win' whenever $x_{A} = x_{B} = x_{C}$ as by Definition \ref{def:MoE-winning-probability}.
\begin{figure}[t]
\centering
\includegraphics[width=0.9\textwidth]{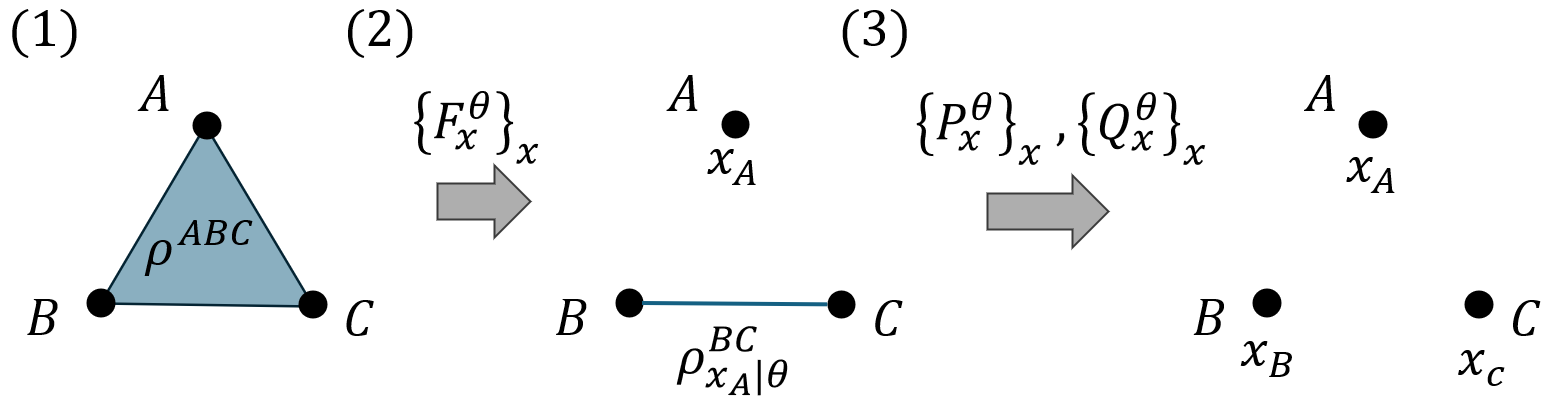}
\caption{\footnotesize Depiction of a Monogamy-of-Entanglement game. In the first step, Alice, Bob, and Charlie share a joint quantum state $\rho^{ABC}$ that Bob and Charlie chose. In Step 2, Alice chooses $\theta$ uniformly randomly and measures the $A$ system with POVM $\{F^{\theta}_{x}\}_{x}$ to obtain outcome $x_{A}$. Bob and Charlie then share a conditional state $\rho^{BC}_{x_{A}\vert \theta}$. In step 3, Alice has told Bob and Charlie the value of $\theta$, so they perform measurements $\{P^{\theta}_{x}\}_{x}$, $\{Q^{\theta}_{x}\}_{x}$ to their respective portions of the conditional state obtaining outcomes $x_{B},x_{C}$. Bob and Charlie win the game when $x_{A} = x_{B} = x_{C}$.
}
\label{fig:MoE-game}
\end{figure}

\begin{definition}\label{def:MoE-game} A monogamy-of-entanglement (MoE) game $G$ is defined by a finite-dimensional Hilbert space $A = \mbb{C}^{d}$ and a set of $\vert \Theta \vert$ POVMs on $A$, $\cM^{\theta} = \{F^{\theta}_{x}\}_{x \in \cX}$. $\theta \in \Theta$ indexes the choice of POVM and $x \in \cX$ indexes the measurement outcome. $\Theta,\cX$ are finite alphabets. 
\end{definition}

\begin{definition}\label{def:MoE-strategy}
\sloppy A strategy $\cS$ for a MoE game $G$ is specified by a tuple
$$\cS = \left(\rho^{ABC}, \{P_{x}^{\theta}\}_{x \in \cX,\theta \in \Theta}, \{Q_{x}^{\theta}\}_{x \in \cX,\theta \in \Theta}\right) \ , $$ where $\rho^{ABC}$ is quantum state, $B,C$ are arbitrary finite-dimensional Hilbert spaces, and $\{P_{x}^{\theta}\}_{x \in \cX}$ (resp.~$\{Q_{x}^{\theta}\}_{x \in \cX}$) is a POVM on $B$ (resp.~$C$) for every $\theta \in \Theta$.  
\end{definition}

\begin{definition}\label{def:MoE-winning-probability}
The winning probability of a MoE game $G$ with a strategy $\cS$ is defined as
\begin{align}\label{eq:MoE-strat-win-prob}
	p_{\text{win}}(G,\cS) \coloneq \sum_{\theta \in \Theta} \Tr[\Pi^{\theta} \rho^{ABC}], \quad \text{where} \quad \Pi^{\theta} \coloneq \sum_{x} F_{x}^{\theta} \otimes P^{\theta}_{x} \otimes Q^{\theta}_{x} \ .
\end{align}
The optimal winning probability is $p_{\text{win}}(G) \coloneq \sup_{\cS} p_{\text{win}}(G,\cS)$.
\end{definition}

The intuition behind an MoE game is as follows. Imagine Alice and Bob share a maximally entangled state $\dyad{\Psi^{+}}^{AB}$ and Alice performs projective measurements, i.e. $\cM^{\theta} = \{U_{\theta}\dyad{x}U_{\theta}^{\dagger}\}_{x \in \cX}$ where $U_{\theta}$ is some unitary for every $\theta \in \Theta$. A standard calculation using the transpose trick \cite{Wilde-2011a} will verify that Bob's conditional state $\rho^{B}_{x \vert \theta} = U^{\Trans}\dyad{x}\ol{U}$ where $\ol{X}$ denotes the entry-wise conjugate of $X$ and $\Trans$ is the matrix transpose. As $\{\rho^{B}_{x \vert \theta}\}_{x \in \cX}$ is a set of mutually orthogonal states, this means in this scenario if Alice measures with $\cM^{\theta}$ and tells Bob $\theta$, Bob can always determine $x$. However, if Bob shares a maximally entangled state with Alice, then Charlie is completely independent, i.e. $\rho_{ABC} = \dyad{\Psi^{+}}^{AB} \otimes \rho^{C}$. This implies Charlie's conditional state will be completely independent of Alice's measurement and so it will be hard for \textit{both} Bob and Charlie to guess the value of $x$ correctly. `Monogamy-of-entanglement' generalizes the above to the idea that the more entangled Bob is with Alice, the less entangled Charlie is with Alice. An MoE game tests this idea operationally. Namely, it tests how entangled $\rho_{AB}$ is versus $\rho_{AC}$ by examining the winning probability, \eqref{eq:MoE-strat-win-prob}. If monogamy of entanglement applies, then there does not exist $\rho^{ABC}$ such that $\sum_{\theta} \Tr[\Pi^{\theta} \rho^{ABC}] = 1$ where $\Pi^{\theta}$ is defined in \eqref{eq:MoE-strat-win-prob}. Note that this does not directly capture anything about broadcasting or orthogonality broadcasting.

It is difficult to prove there does not exist $\rho^{ABC}$ such that $\sum_{\theta} \Tr[\Pi^{\theta} \rho^{ABC}] = 1$ as one needs to optimize over all possible $\{P^{\theta}_{x}\}$ and $\{Q^{\theta}_{x}\}$ in Definition \ref{def:MoE-strategy}. As such, MoE games are instead analyzed using the bound 
$\sum_{\theta} \Tr[\Pi^{\theta}\rho^{ABC}] \leq \Vert \sum_{\theta} \Pi^{\theta} \Vert$, where $\Vert \cdot \Vert$ is the operator norm and this follows from $\Tr[M\rho^{ABC}] \leq \Vert M \Vert$. This is still difficult due to optimizing over both measurements, so further tools are needed. In particular, \cite{Tomamichel-2013a} establishes and uses the following lemma.
\begin{lemma}\cite{Tomamichel-2013a}\label{lem:MoE-op-norm}
	Let $R_{0},R_{1},...,R_{n}$ be $n$ positive semidefinite operators. Let $\{\pi^{k}\}_{k \in [n]}$ be a set of $n$ permutations such that $\pi^{k_{1}}(i) \neq \pi^{k_{2}}(i)$ for all $i \in [n]$, $k_{1},k_{2} \in [n]$ such that $k_{1} \neq k_{2}$. Then the following inequality holds
	\begin{align}
		\Big \Vert \sum_{i \in [n]} R_{i} \Big \Vert \leq \sum_{k \in [n]} \max_{i \in [n]} \Big \Vert \sqrt{R_{i}} \sqrt{R_{\pi^{k}(i)}} \Big \Vert \ .
	\end{align}
\end{lemma}

\subsection*{Reduction of Specific QPV Protocols to Monogamy-of-Entanglement Games}
\begin{figure}[b]
\centering
\includegraphics[width=\textwidth]{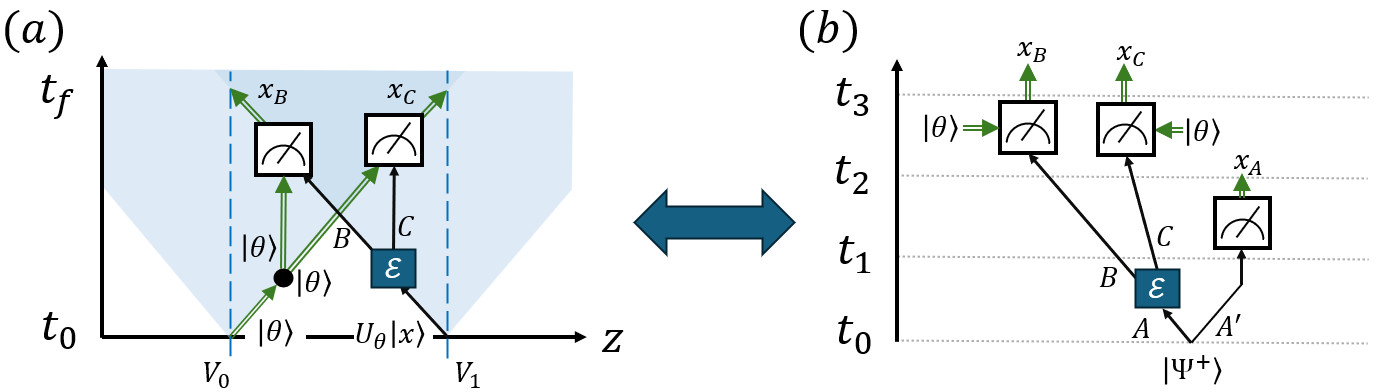}
\caption{\footnotesize Depiction of the reduction of a certain class of QPV protocols to monogamy-of-entanglement games. (a) Depicts the spacetime diagram of the relevant QPV protocols where one dishonest prover receives $\ket{\theta}$ to copy and forward and the other party receives $U_{\theta}\ket{x}$ which they attempt to broadcast with $\cE^{A \to BC}$. (b) Depicts a mathematically equivalent setup to the QPV protocol. At the initial time, a maximally entangled state is shared. At $t_{1}$, $\cE^{A \to BC}$ has been applied. At $t_{2}$, a random $\theta \in \Theta$ has been chosen and the measurement $\{F^{\theta}_{x} \coloneq U_{x}^{\Trans}\dyad{x}\ol{U}_{x}\}_{x}$ has been applied to the $A'$ system. At $t_{3}$, measurements $\{P^{\theta}_{x}\}$, $\{Q^{\theta}_{x}\}$ have been applied according to $\theta$. Note that $t_{1}$, $t_{2}$, and $t_{3}$ respectively correspond to (1),(2), and (3) in Fig.~\ref{fig:MoE-game}. In other words, it depicts a specific strategy of a MoE game.
}
\label{fig:QPV-to-MoE}
\end{figure}

We now explain how the reduction of a specific instance of QPV to MoE works as it will help to clarify how MoE and orthogonality broadcasting compare. Certain systems are labeled differently than in the main text to make the conversion to an MoE game clearer. Consider a QPV protocol where $V_{0}$ sends $\theta \in \Theta$ and $V_{1}$ sends $\rho^{A}_{x \vert \theta} \coloneq U_{\theta}\dyad{x}U_{\theta}^{\dagger}$ where $\{U_{\theta}\}_{\theta \in \Theta}$ is a set of unitaries and $x,\theta$ are drawn uniformly from their sets (See Fig.~\ref{fig:QPV-to-MoE} for depiction). The dishonest prover that receives $\rho_{x \vert \theta}$, Bob, must attempt to broadcast the orthogonality using some broadcast channel $\cE^{A \to BC}$. The dishonest prover that receives $\theta$, Charlie, simply copies and forwards the value to Charlie. Then, upon receiving $\theta$ and their respective quantum system, Bob and Charlie perform local measurements conditioned on $\theta \in \Theta$ to obtain outcomes $x_{B},x_{C}$ and send these to the verifiers closer to them.\footnote{Recalling Fig.~\ref{Fig:orthogonality_broadcasting}, this shows this specific type of QPV protocol is exactly orthogonality broadcasting.}

However, using the `transpose trick' \cite{Wilde-2011a}, one may verify $\Tr_{A'}[U^{\Trans}_{\theta}\dyad{x}\ol{U}_{\theta} \otimes \mbbm{1}_{A} \dyad{\Psi^{+}}] = U_{\theta}\dyad{x}U_{\theta}^{\dagger}$ where $\ket{\Psi^{+}}_{A'A}$ is the maximally entangled state and the transpose is defined in the computational basis. It follows that the QPV protocol is \textit{mathematically} equivalent to the following procedure (depicted in Fig.~\ref{fig:QPV-to-MoE}):
\begin{enumerate}
 \item $V_{0}$ sends nothing and $V_{1}$ sends the $A$ system of $\dyad{\Psi^{+}}^{A'A}$ to Charlie.
 \item Charlie broadcasts the $A$ system with $\cE^{A \to BC}$.
 \item At some point before Bob and Charlie decide to measure their local portions of $\rho^{ABC} \coloneq (\id^{A'} \otimes \cE^{A \to BC})(\dyad{\Psi^{+}})$, $\theta$, $V_{1}$ randomly draws $\theta$, measures the $A'$ system with $\{U^{\Trans}_{\theta}\dyad{x}\ol{U}_{\theta}\}_{x \in \cX}$, and provides the value of $\theta$ to both Bob and Charlie.
 \item Bob and Charlie perform their measurements as in the true protocol.
\end{enumerate}
Note this new procedure is an MoE game with some specification of the strategy, namely $\rho_{ABC}$. It follows that the optimal strategy for the MoE game defined by $\cM^{\theta} \coloneq \{U^{\Trans}_{\theta}\dyad{x}\ol{U}_{\theta}\}$ provides an upper bound on all possible QPV attacks in the no-PE model. Moreover, it just so happens that this is tight when $\{U_{\theta}\}_{\theta} = \{\sigma_{X},H\}$, which are the Pauli X and Hadamard unitaries, i.e. the bound is tight for BB84 QPV \cite{Tomamichel-2013a}.

\begin{remark}
\cite[Appendix A]{Poremba-2024a} introduces a ``$1\to 2$ cloning game" where $\cE_{A \to BC}$ is specified. As they show, this is like restricting the strategy of the MoE game to avoid the last reduction to MoE given above.
\end{remark}

\subsection*{Comparing Orthogonality Broadcasting and Monogamy-of-Entanglement Games}
With the sufficient background, we may now separate orthogonality broadcasting from MoE games. First, in the previous subsection, a footnote specified where the equivalence of this specific type of QPV protocol to orthogonality broadcasting occurs. This was prior to converting to an alternative protocol using a maximally entangled state. This alone shows it is at least more direct for the protocols above. Second, as orthogonality broadcasting considers a different aspect of quantum mechanics than MoE games, it can be applied to QPV protocols that do not satisfy the structure necessary for the reduction to an MoE game. Indeed, noting that both systems in \eqref{eq:qq-states} used in Theorem \ref{thm:qq-qc-separation} are quantum, it follows the MoE game reduction given above cannot be done at all for this ensemble. This already distinguishes the two methods. However, even if one of the inputs is classical, the reduction to MoE games is insufficient with current methods as the following example shows.

\begin{example}\label{ex:MoE-not-sufficient}
Note that \eqref{eq:overlapping-BB84-states} has one system always classical, but the quantum states partition (with respect to the other system's classical value) as $\{\ket{0},\ket{1},\ket{2}\}$,$\{\ket{0 \pm 1}\}$,$\{\ket{1 \pm 2}\}$. Therefore, this is an ensemble that cannot be generated in the form $\{\ket{\theta} \, (U_{\theta}\ket{x}\}_{\theta,x}$ as the reduction in the previous paragraph required. However, the situation is significantly worse--- if we tried to complete these bases, the measurements are
\begin{equation}
\begin{aligned}\label{eq:OBB84-MoE-game}
	\cM^{0} &= \{\dyad{0},\dyad{1},\dyad{2}\}\\
	 \cM^{1} &= \{\dyad{0},\dyad{1+2},\dyad{1-2}\} \\
	 \cM^{2} &= \{\dyad{0+1},\dyad{0-1},\dyad{2}\} \ .
\end{aligned}
\end{equation}
If we define $G_{O}$ to be the MoE game defined by the measurements in \eqref{eq:OBB84-MoE-game}, then the methods of \cite{Tomamichel-2013a} obtain trivial bounds. This is because the measurements pairwise contain a projective POVM element that is identical. For completeness, we show this. First, following the notation of \cite[Theorem 4]{Tomamichel-2013a},
\begin{align}
	c(G_{O}) = \max_{\substack{\theta,\theta' \in \Theta \\ \theta \neq \theta'}} \max_{x,x' \in \cX} \Big \Vert \sqrt{F^{\theta}_{x}} \sqrt{F^{\theta'}_{x'}} \Vert = \Vert \dyad{0} \dyad{0} \Vert = 1 \ ,
\end{align}
where we used $F^{0}_{0} = \dyad{0} = F^{1}_{0}$ by \eqref{eq:OBB84-MoE-game}. Then \cite[Theorem 4]{Tomamichel-2013a} provides the trivial upper bound of $1$. Moreover, the more in-depth analysis of \cite[Theorem 3]{Tomamichel-2013a} is also trivial as we now show. First,
\begin{align}
	p_{\text{win}}(G_{O}) = \frac{1}{3}\sum_{\theta} \Tr[\Pi^{\theta} \rho_{ABC}] \leq \frac{1}{3} \Vert \sum_{\theta} \Pi^{\theta} \Vert \leq& \frac{1}{3} \sum_{k} \max_{\theta} \Vert \Pi^{\theta}\Pi^{\pi^{k}(\theta)} \Vert \ ,
\end{align}
where $\Pi^{\theta}$ is defined as in \eqref{eq:MoE-strat-win-prob} and the second inequality is Lemma \ref{lem:MoE-op-norm}. In principle we would need to optimize over strategies, but we will just pick one that makes the bound trivial. Let both Bob and Charlie use the same measurements as Alice as given in \eqref{eq:OBB84-MoE-game}. Then $\Pi^{\theta} = \sum_{i \in [3]} {F^{\theta}_{i}}^{\otimes 3}$ for $\theta \in [3]$. Because each $\theta,\theta'$ share a rank-one projector, $\Vert \Pi^{0}\Pi^{1} \Vert = \Vert \Pi^{0}\Pi^{2} \Vert = \Vert \Pi^{1}\Pi^{2} \Vert = 1$. It follows with this choice $p_{\text{win}}(G_{O}) \leq \frac{1}{3} \sum_{k} \max_{\theta} \Vert \Pi^{\theta}\Pi^{\pi^{k}(\theta)} \Vert  = \frac{1}{3} \cdot 3 = 1$. Thus the bound remains trivial. This shows MoE games, at least using the standard tool, (Lemma \ref{lem:MoE-op-norm}) cannot establish security for this ensemble.
\end{example}

\section{Necessity of Non-Separable Communication}
In this section we establish Theorem \ref{Thm:no-sep-comm}. 

\begin{proof}[Proof of Theorem \ref{Thm:no-sep-comm}]
Consider the eight orthogonal product states stated in the theorem, which we label
\begin{equation}
\begin{aligned}
    \ket{\varphi_1}=\ket{1}^A\ket{1}^B &\quad
    \ket{\varphi_2}=\ket{1}^A\ket{2}^B\notag\\
    \ket{\varphi_{3}}=\ket{2}^A\ket{+_{23}}^B &\quad
    \ket{\varphi_{4}}=\ket{2}^A\ket{-_{23}}^B\notag\\
    \ket{\varphi_{5}}=\ket{3}^A\ket{+_{24}}^B &\quad
    \ket{\varphi_{6}}=\ket{3}^A\ket{-_{24}}^B\notag\\
    \ket{\varphi_{7}}=\ket{4}^A\ket{\wt{+}_{34}}^B &\quad
    \ket{\varphi_{8}}=\ket{4}^A\ket{\wt{-}_{34}}^B.
\end{aligned}
\end{equation}
where we remind the reader $\ket{\pm_{mn}}=\frac{1}{\sqrt{2}}(\ket{m}\pm\ket{n}$ and $\ket{\wt{\pm}_{mn}}=\frac{1}{\sqrt{2}}(\ket{m}\pm i\ket{n}$. It is clear that Alice's action in any LOSQC discrimination protocol should involve her measuring in the computational basis and distributing the outcome information to Bob.  It is therefore Bob's job to broadcast the orthogonality of his portion of the state.  More precisely, he must perform a channel $\ket{b_x}^B\mapsto \rho_x^{A_2B_1}$ such that $\{\rho_x^{A_2},\rho_{x+1}^{A_2}\}$ and $\{\rho_x^{B_1},\rho_{x+1}^{B_1}\}$ form sets of orthogonal operators for $x=1,3,5,7$. Moreover, as per the statement of the theorem, we are interested in establishing that any such broadcasting channel $\cE_{B \to A'B'}$ must map one of the states on the $B$ space to the joint space $A'B'$ such that it is entangled. We will prove this is the case by contradiction, i.e. we will look at the structure necessary for $\cE_{B \to A'B'}$ to map each state to a separable state.

Recall that a separable operator $R$ is a positive operator from $A \otimes B$ to $A \otimes B$, $X \in \Pos(A \otimes B)$, such that there exists a finite alphabet $\cI$ and sets of positive operators $\{P_{i}\}_{i \in \cI} \subset \Pos(A)$, $\{Q_{i}\}_{i \in \cI} \subset \Pos(B)$ such that $R = \sum_{i} P_{i} \otimes Q_{i}$. Note that by the spectral decomposition theorem, we may in fact decompose such an operator into a linear combination of tensor products of unit vectors. Thus, we would say $\rho_{k}^{A'B'} = \cE_{B \to A'B'}(\rho_{k})$ is separable if we may write it as $\sum_{i} \dyad{x_{i}} \otimes \dyad{y_{i}}$ where $\{\ket{x_{i}}\}_{i \in \cI} \subset A'$, $\{\ket{y_{i}}\}_{i \in \cI} \subset B'$ are (not necessarily normalized) vectors. Note that, when $\rho_{k}^{A'B'}$ is separable, this implies a purification, of the form 
\begin{align}\label{eq:purified-separable-state}
    \sum_{i} \ket{x_{i}}^{A'}\ket{y_{i}}^{B'}\ket{i}^{E} \ ,
\end{align} 
where $\{\ket{i}\}_{i}$ is an orthonormal basis. That this is a purification may be verified by taking the trace over the $E$ system. Moreover, by the isometric equivalence of purifications on the purifying space \cite{WatrousBook}, \textit{all} purifications of $\rho_{k}^{A'B'}$ are of this form up to a choice of basis on the $E$ system. It follows that if $\cE_{A \to A'B'}$ maps a pure state $\ket{\varphi}$ to a separable state $\rho^{A'B'}$, then if we consider the isometric extension of $\cE$, $U_{A \to A'B'E}$, $U\ket{\varphi}$ is a pure state and thus a purification of $\rho_{k}^{A'B'}$. It follows that it must have the form of \eqref{eq:purified-separable-state}. This is the structure we will make use of in the rest of the proof.

Suppose that Bob \textit{does not} distribute entanglement to Alice given the state $\ket{\varphi_2}$. Then by the above explanation, this means he applies an isometry $U:A\to A'B'E$ such that
\begin{equation}
U\ket{2}=\sum_{i}\ket{x_i}^{A'}\ket{y_i}^{B'}\ket{i}^E.
\end{equation}
Turning to states $\ket{\varphi_3}$ and $\ket{\varphi_4}$, by linearity, the action of Bob's isometry has the form
\begin{align*}
\ket{\alpha_3}^{A'B'E}=&U\ket{+_{23}} 
=\frac{1}{\sqrt{2}}\sum_i\left(\ket{x_i}^{A'}\ket{y_i}^{B'}+\ket{\psi_i}^{A'B'}\right)\ket{i}^E \\
\ket{\alpha_4}^{A'B'E}=&U\ket{-_{23}} 
=\frac{1}{\sqrt{2}}\sum_i\left(\ket{x_i}^{A'}\ket{y_i}^{B'}-\ket{\psi_i}^{A'B'}\right)\ket{i}^E
\end{align*}
where $U\ket{3}=\sum_i\ket{\psi_i}^{A'B'}\ket{i}^E$, which is without loss of generality as we may always decompose a bipartite entangled state into conditional states \cite{Schumacher-2010a} where note that $\ket{\psi_{i}}^{A'B'}$ are not necessarily normalized by our form.

We now wish to consider the requirement that $\alpha_3^{A'}\perp\alpha_4^{A'}$ and $\alpha_3^{B'}\perp\alpha_4^{B'}$. Note that, by the orthonormality of $\{\ket{i}^{E}\}$, the marginals are
\begin{equation}\label{eq:3-4-possibly-ent}
\begin{aligned}
    \alpha^{A'B'}_{3} =& \frac{1}{2}\sum_{i}\left[\left(\ket{x_i}^{A'}\ket{y_i}^{B'}+\ket{\psi_i}^{A'B'}\right) \text{h.c.}\right] \\ 
    \alpha^{A'B'}_{4} =&  \frac{1}{2}\sum_{i}\left[\left(\ket{x_i}^{A'}\ket{y_i}^{B'}-\ket{\psi_i}^{A'B'}\right) \text{h.c.}\right] \ ,
\end{aligned}
\end{equation}
where $\text{h.c.}$ denotes the adjoint state. Now, to satisfy $\alpha_3^{A'}\perp\alpha_4^{A'}$ and $\alpha_3^{B'}\perp\alpha_4^{B'}$, by linearity, one may verify from \eqref{eq:3-4-possibly-ent} that it must be the case that the marginals on the $A'$ and $B'$ of
\begin{align*}
\ket{x_i}^{A'}\ket{y_i}^{B'}+\ket{\psi_i}^{A' B'} =&(\ket{x_i}^{A'}+\ket{x_i'}^{A'})\ket{y_i} +\sum_j\ket{x_{j|i}'}^{A'}\ket{y_{j|i}^\perp}^{B'} \\
\ket{x_i}^{A'}\ket{y_i}^{B'}-\ket{\psi_i}^{A' B'} =&(\ket{x_i}^{A_2}-\ket{x_i'}^{A'})\ket{y_i} -\sum_j\ket{x_{j|i}'}^{A'}\ket{y_{j|i}^\perp}^{B'}
\end{align*}
must be orthogonal for each $i$.  Here, all the $\{\ket{y_{j|i}^\perp}\}_j$ form an orthonormal set orthogonal to $\ket{y_i}$ which follows from decomposing the state into conditional states based on a basis that includes $\ket{y_{i}}$.  Because of the orthonormality of the set, the requirement of orthogonality of Alice's marginals implies that $\ket{x'_{j|i}}=0$ for all $j$. It follows from this simplification that the orthogonality of Bob's marginals then requires that $\ket{x_i}=\pm\ket{x_i'}$ for every $i$.  In summary,  we have
\begin{align*}
U\ket{2}&=\sum_i\ket{x_i}^{A'}\ket{y_i}^{B'}\ket{i}^E\notag\\
U\ket{3}&=\sum_i(-1)^{r_i}\ket{x_i}^{A'}\ket{y_i}^{B'}\ket{i}^E,
\end{align*}
for some $r_i\in\{0,1\}$.  But by the same argument, we must have
\begin{equation*}
U\ket{4}=\sum_i(-1)^{s_i}\ket{x_i}^{A'}\ket{y_i}^{B'}\ket{i}^E,
\end{equation*}
for some $s_i\in\{0,1\}$.  Then turning to states $\ket{\varphi_7}$ and $\ket{\varphi_8}$ gives
\begin{align*}
U\ket{\wt{+}_{34}}&= \frac{1}{\sqrt{2}} \sum_i\left((-1)^{s_i}+i(-1)^{r_i}\right)\ket{x_i}^{A'}\ket{y_i}^{B'}\ket{i}^E,\notag\\
U\ket{\wt{-}_{34}}&= \frac{1}{\sqrt{2}} \sum_i\left((-1)^{s_i}-i(-1)^{r_i}\right)\ket{x_i}^{A'}\ket{y_i}^{B'}\ket{i}^E \ .
\end{align*}
By direct calculation, the reduced density matrices of these states are not orthogonal (in fact they are the same), and hence they cannot be locally distinguished. This proves that the ensemble of states cannot be perfectly distinguished by LOSQC if Bob does not distribute entanglement to Alice.

Let us now show that an entangling LOSQC strategy can perfectly discriminate the states.  In fact, it only requires two bits of communication from Alice to Bob, and just one qubit communication from Bob to Alice!  The protocol involves Bob performing an isometry with action 
\begin{equation}
\begin{aligned}
U\ket{1} &= \ket{11}^{A'B'}  \\
U\ket{2}&=\frac{1}{\sqrt{2}}(\ket{22}+\ket{33})^{A'B'} \\
U\ket{3}&=\frac{1}{\sqrt{2}}(\ket{22}-\ket{33})^{A'B'} \\
U\ket{4}&=\frac{1}{\sqrt{2}}(\ket{23}+\ket{32})^{A'B'}.
\end{aligned}
\end{equation}
Then
\begin{equation}
\begin{aligned}
U\ket{+_{23}}&=\ket{22},&U\ket{-_{23}}&=\ket{33}\notag\\
U\ket{+_{24}}&=\ket{+_{23}}\ket{+_{23}},&U\ket{-_{24}}&=\ket{-_{23}}\ket{-_{23}}\notag\\
U\ket{\wt{+}_{34}}&=\ket{\wt{+}_{23}}\ket{\wt{+}_{23}},&U\ket{\wt{-}_{34}}&=\ket{\wt{-}_{23}}\ket{\wt{-}_{23}}.
\end{aligned}
\end{equation}
Hence, the correct pairwise orthogonality is achieved. Moreover, the local states for $U\ket{1}, U\ket{2}$ are $\dyad{1}$ and $\frac{1}{2}[\dyad{2}+\dyad{3}]$, so these are also pairwise orthogonal. Thus, the states can be perfectly discriminated.
\end{proof}

\section{Uncertainty Relation}
We begin with the derivation of the uncertainty relation. We prove the more general form and then provide sufficient details to see how to obtain Lemma \ref{lem:main-text-UR} as a special case of the method.

\begin{theorem}\label{thm:generalized-uncertainty-relation}
    Let $\{\ket{\gamma}_{i}^{AB}\}_{i \in \cI}$ be a set of (possibly unnormalized) vectors. Let 
    \begin{align*}
        \ket{\gamma_{\pmb{\alpha}}}^{AB} := \sum_{i} \alpha_{i} \ket{\gamma_{i}}^{AB} \quad \ket{\gamma_{\pmb{\beta}}}^{AB} := \sum_{i} \beta_{i} \ket{\gamma_{i}}^{AB} \ .
    \end{align*}
    Then,
    \allowdisplaybreaks
    \begin{align*}
        D_{\text{tr}}(\gamma_{\pmb{\alpha}}^{B}, \gamma_{\pmb{\beta}}^{B})
        \leq& \min_{\sigma \in \cS(|\cI|)} \sum_{i} D_{\text{tr}}(\mbf{p}(i)\gamma_{i}^{B},\mbf{q}(\sigma(i))\gamma_{\sigma(i)}^{B}) + \sum_{i}\sum_{j\neq i} |z_{ij}|F(\gamma_{i}^{A},\gamma_{j}^{A}) \\
        \leq& D_{\text{tr}}(\mbf{p},\mbf{q}) + \sum_{i}\sum_{j\neq i} |z_{ij}|F(\gamma_{i}^{A},\gamma_{j}^{A})
    \end{align*}
    where $\sigma$ is a permutation on $|\cI|$ elements, $\mbf{p}(i) := |\alpha_{i}|^{2}$, $\mbf{q}(i) := |\beta_{i}|^{2}$, and for every $(i,j\neq i)$, $z_{ij} := \alpha_{i}\alpha_{j}^{\ast} - \beta_{i}\beta_{j}^{\ast}$.
\end{theorem}
\begin{proof}
    By direct calculation,
    \begin{align*}
        \gamma_{\pmb{\alpha}} = \sum_{i} |\alpha_{i}|^{2} \dyad{\gamma_{i}} + \sum_{i,j\neq i} r^{\alpha}_{ij} \Big[ e^{-i \phi^{\alpha}_{ij}}\ket{\gamma_{i}}\bra{\gamma_{j}} + e^{i \phi^{\alpha}_{ij}}\ket{\gamma_{j}}\bra{\gamma_{i}}  \Big] \ ,
    \end{align*}
    where for every $(i,j)$ $r^{\alpha}_{ij} := |\alpha_{i}\alpha_{j}^{\ast}|$ and $\varphi_{ij}^{\alpha} := \text{atan2}(\text{Im}(\alpha_{i}\alpha_{j}^{\ast}),\text{Re}(\alpha_{i}\alpha_{j}^{\ast}))$. The same calculation holds for $\gamma_{\pmb{\beta}}$.

    For every $i$, let $z_{i} := |\alpha_{i}|^{2} - |\beta_{i}|^{2}$. For every $(i,j \neq i)$, let $z_{ij} := r^{\alpha}_{ij}e^{-i \phi^{\alpha}_{ij}} - r^{\beta}_{ij}e^{-i \phi^{\beta}_{ij}}$. Therefore, 
    \begin{align*}
        & \left \Vert \gamma_{\pmb{\alpha}}^{B} -  \gamma_{\pmb{\beta}}^{B} \right \Vert_{1} \\
        =&  \left \Vert \Tr_{A}\left[\gamma_{\pmb{\alpha}}^{AB} -  \gamma_{\pmb{\beta}}^{AB} \right] \right \Vert_{1} \\
        =& \Big \Vert\sum_{i} z_{i}\Tr_{A}[\dyad{\gamma_{i}}] + \sum_{i}\sum_{j \neq i} \Tr_{A}\left[z_{ij}\ket{\gamma_{i}}\bra{\gamma_{j}} - z_{ij}^{\ast}\ket{\gamma_{j}}\bra{\gamma_{i}} \right] \Big \Vert_{1} \\
        \leq& \left \Vert \sum_{i} z_{i} \Tr_{A}[\dyad{\gamma_{i}}] \right \Vert_{1} + \sum_{i}\sum_{j \neq i} |z_{ij}|\left \Vert \Tr_{A}\left[\ket{\wt{\gamma}_{i}^{j}}\bra{\gamma_{j}} +\ket{\gamma_{j}}\bra{\wt{\gamma}_{i}^{j}}\right] \right \Vert_{1} \ ,
    \end{align*}
    where the inequality is the triangle inequality and we have defined $\ket{\wt{\gamma}_{i}^{j}} := e^{-i \delta \phi_{ij}}\ket{\gamma_{i}}$ and $\delta \phi_{ij} := \text{atan2}(\text{Im}(z_{ij}),\text{Re}(z_{ij}))$.

    To bound the first term, note that we can pick any permutation $\sigma \in \cS(|\cI|)$ to obtain the following
    \begin{align*}
    & \left \Vert \sum_{i} z_{i} \Tr_{A}[\dyad{\gamma_{i}}] \right \Vert_{1} \\
    =& \left \Vert \sum_{i} \vert \alpha_{i} \vert^{2} \Tr_{A}[\dyad{\gamma_{i}}] - \sum_{i}  \vert \beta_{i} \vert^{2} \Tr_{A}[\dyad{\gamma_{i}}] \right \Vert_{1} \\
    =& \left \Vert \sum_{i} \vert \alpha_{i} \vert^{2} \Tr_{A}[\dyad{\gamma_{i}}] - \sum_{i}  \vert \beta_{\sigma(i)} \vert^{2} \Tr_{A}[\dyad{\gamma_{\sigma(i)}}] \right \Vert_{1} \\
    =& \left \Vert \sum_{i} \mbf{p}(i) \Tr_{A}[\dyad{\gamma_{i}}] - \sum_{i} \mbf{q}(\sigma(i)) \Tr_{A}[\dyad{\gamma_{\sigma(i)}}] \right \Vert_{1} \\
    \leq& \sum_{i} \Vert \mbf{p}(i) \gamma_{i}^{B} - \mbf{q}(\sigma(i)) \gamma_{\sigma(i)}^{B} \Vert_{1} \ ,
    \end{align*}
    where the first equality is definition of $z_{i}$, the second is that the second sum is invariant under the permutation on the indices, the third equality is definition of $\mbf{p},\mbf{q}$, and the inequality is triangle inequality. As this held for arbitrary permutation, we can minimize over the permutation.
    
    We now need to handle the cross terms. Define $P_{\pm}^{ij}$ be the projector onto the $\pm$ eigenspace of $\Tr_{A}\Big[\ket{\wt{\gamma}_{i}^{j}}\bra{\gamma_{j}} + \ket{\gamma_{j}}\bra{\wt{\gamma}_{i}^{j}}\Big]$. Then by definition of trace norm,
    \begin{align*}
        \Big \|\Tr_{A}\left[\ket{\wt{\gamma}_{j}^{i}}\bra{\gamma_{i}} + \ket{\gamma_{j}}\bra{\wt{\gamma}_{i}^{j}}\right] \Big\|_{1} 
        = \Tr[(\mbbm{1}_{A} \otimes P_{+}^{ij} - P_{-}^{ij})\left(\ket{\wt{\gamma}_{i}^{j}}\bra{\gamma_{j}} + \ket{\gamma_{j}}\bra{\wt{\gamma}_{i}^{j}}\right)]
    \end{align*}
    
    Now,
    \begin{align*}
    &F(\gamma_{i}^{A},\gamma_{j}^{A}) \\
    =& F(\wt{\gamma}_{i}^{j,A},\gamma_{j}^{A}) \\
    =&
    \max_{U} \left|\Tr[(\mbbm{1} \otimes U)\ket{\wt{\gamma}_{i}^{j}}\bra{\gamma_{j}}^{AB}]\right| \\
    = & \frac{1}{2} \Big( \max_{U} \left|\Tr[(\mbbm{1} \otimes U)(\ket{\wt{\gamma}_{i}^{j}}\bra{\gamma_{j}} )]\right| +\max_{U} \left|\Tr[(\mbbm{1} \otimes U)(\ket{\gamma_{j}}\bra{\wt{\gamma}_{i}^{j}} )]\right| \Big) \\
    \geq & \frac{1}{2}\left( \max_{U} \left| \Tr[(\mbbm{1} \otimes U)\left(\ket{\wt{\gamma}_{i}^{j}}\bra{\gamma_{j}} + \ket{\gamma_{j}}\bra{\wt{\gamma}_{i}^{j}}\right )]\right| \right) \ ,
    \end{align*}
    where the first equality is noting $\gamma_{i} = \dyad{\gamma_{i}} = \dyad{\wt{\gamma}_{i}^{j}}$ and the second is Uhlmann's theorem. 
    
   Choosing $U = P^{ij}_{+} - P^{ij}_{-} + (\mbbm{1}^{A} - P^{ij}_{+} - P^{ij}_{-} )$, we have
    \begin{align*}
    F(\gamma_{i}^{A},\gamma_{j}^{A})
    \geq & \frac{1}{2}\left( \max_{U} \left| \Tr[(\mbbm{1} \otimes U)\left(\ket{\wt{\gamma}_{i}^{j}}\bra{\gamma_{j}} + \ket{\gamma_{j}}\bra{\wt{\gamma}_{i}^{j}}\right )]\right| \right) \\
    \geq & \frac{1}{2}\Tr[\mbbm{1} \otimes (P^{ij}_{+}-P^{ij}_{-})\left(\ket{\wt{\gamma}_{i}^{j}}\bra{\gamma_{j}} + \ket{\gamma_{j}}\bra{\wt{\gamma}_{i}^{j}}\right )] \\
    =& \frac{1}{2} \left \| \Tr_{A}\left[\ket{\wt{\gamma}_{i}^{j}}\bra{\gamma_{j}} + \ket{\gamma_{j}}\bra{\wt{\gamma}_{i}^{j}}\right] \right \| \ .
    \end{align*}
    Reordering gets us
    \begin{align*}
    \left \|  \gamma_{\pmb{\alpha}}^{B} -  \gamma_{\pmb{\beta}}^{B} \right \|_{1}
    \leq& \min_{\sigma \in \cS(|\cI|)} \sum_{i} \|\mbf{p}(i) \gamma_{i}^{B} - \mbf{q}(i) \gamma_{\sigma(i)}^{B} \|_{1}  + 2\sum_{i} \sum_{j \neq i} |z_{ij}|F(\gamma_{i}^{A},\gamma_{j}^{A}) \ .
    \end{align*}
    Dividing by two gets us the trace distance and fidelity bound.
\end{proof}
To obtain the special case in the main text (Lemma \ref{lem:main-text-UR}), the proof is largely the same, but we have more structure. Namely, there is only one $z_{ij}$ term, which is $z_1 \coloneq \frac{1}{2}(\sin(\theta)e^{-i\phi}-\sin(\omega)e^{-i\phi'})$ and the bound on `the first term' becomes quite simple in terms of a single scalar, which is $z_{2} \coloneq \frac{1}{2}(\cos(\theta)-\cos(\omega))$. For clarity, we provide the calculations. By direct calculation, one obtains
\begin{align*}
        \gamma_{\theta} =& \cos^{2}(\theta/2) \dyad{\gamma_{0}} + \sin^{2}(\theta/2) \dyad{\gamma_{1}} + \sin(\theta)/2\left[e^{-i\phi}\ket{\gamma_{0}}\bra{\gamma_{1}} + e^{i\phi}\ket{\gamma_{1}}\bra{\gamma_{0}}\right] \\
        \gamma_{\omega} =& \cos^{2}(\omega/2) \dyad{\gamma_{0}} + \sin^{2}(\omega/2) \dyad{\gamma_{1}} + \sin(\omega)/2\left[e^{-i\phi'}\ket{\gamma_{0}}\bra{\gamma_{1}} + e^{i\phi'}\ket{\gamma_{1}}\bra{\gamma_{0}}\right] \ .
\end{align*}
Starting from linearity,
\begin{align*}
    & \Vert \gamma_{\theta}^{B} - \gamma_{\omega}^{B} \Vert_{1} \\
    =& \Big\Vert \Tr_{A}\Big[\left\{\cos^{2}(\theta/2) - \cos^{2}(\omega/2) \right\} \dyad{\gamma_{0}} + \left\{\sin^{2}(\theta/2) - \sin^{2}(\omega/2) \right\} \dyad{\gamma_{1}} + \\
    & \hspace{7.5mm} +\frac{1}{2}\{\sin(\theta)e^{-i\phi}-\sin(\omega)e^{-i\phi} \}\ket{\gamma_{0}}\bra{\gamma_{1}} +\frac{1}{2}\{\sin(\theta)e^{i\phi}-\sin(\omega)e^{i\phi} \}\ket{\gamma_{1}}\bra{\gamma_{0}} \Big] \Big\Vert_{1} \\
    =& \Big\Vert \Tr_{A}\Big[z_2 \dyad{\gamma_{0}} - z_2 \dyad{\gamma_{1}} +z_1 \ket{\gamma_{0}}\bra{\gamma_{1}} + z_{1}^{\ast} \ket{\gamma_{1}}\bra{\gamma_{0}} \Big] \Big\Vert_{1} \\
    \leq& |z_{2}| \Big\Vert \Tr_{A}\Big[\dyad{\gamma_{0}} - \dyad{\gamma_{1}} \Big] \Big\Vert_{1} + \Big \Vert \Tr_{A}\Big[ z_1 \ket{\gamma_{0}}\bra{\gamma_{1}} + z_{1}^{\ast} \ket{\gamma_{1}}\bra{\gamma_{0}} \Big] \Big\Vert_{1} \\
    =& |z_{2}| \Big\Vert \Tr_{A}\Big[\dyad{\gamma_{0}} - \dyad{\gamma_{1}} \Big] \Big\Vert_{1} + \vert z_1 \vert \Big \Vert \Tr_{A}\Big[ \ket{\wt{\gamma}_{0}}\bra{\gamma_{1}} + \ket{\gamma_{1}}\bra{\wt{\gamma}_{0}} \Big] \Big\Vert_{1}  \ ,
\end{align*}
where the second equality is we have used trigonometric identities and our definitions of $z_1,z_2$, the inequality is triangle inequality, and the final equality uses $z_{1} \coloneq |z_{2}|e^{-i\overline{\phi}} $ and $\ket{\wt{\gamma}_{0}} \coloneq e^{-i\overline{\phi}}\ket{\gamma_{0}}$. The rest of the proof is identical to the main theorem proof.

\section{Reduction of LOSQC to LOSCC for qubit-qudit systems}

In this section we prove Theorem \ref{thm:C2-Cd-LOSCC}. This is established in two steps. The first is to reduce the type of ensemble that could separate LOSQC and LOSCC in $\mbb{C}^{2} \otimes \mbb{C}^{d}$ to be of the form $\cS_{\Theta}$ where $\Theta = 2$ and each set is a set of mutually orthogonal pure states. The second step is to show that LOSQC implies an LOSCC strategy for any such ensemble, which is effectively a corollary of Proposition \ref{Prop:classical-broadcasting-equal-2}.

First we prove the reduction.
\begin{lemma}\label{lem:dim-reduc-part-1}
    If there is an orthogonal product set in $\mbb{C}^{2} \otimes \mbb{C}^{d}$ that separates LOSQC and LOSCC in the zero-error case, then there exists a set that separates LOSQC and LOSCC of the form
    \begin{align}\label{eq:separating-form}
        \{\ket{0}^{A}\ket{b^{0}_{i}}\}_{i \in \cK^{0}} \cup \{\ket{1}^{A}\ket{b^{1}_{j}}\}_{j \in \cK^{1}} \ ,
    \end{align}
    where for each $j \in \{0,1\}$, $|\cK^{j}| \leq d$ and $\{\ket{b^{j}_{z}}\}_{z \in \cK^{j}}$ is a set of orthonormal vectors .
\end{lemma}
\begin{proof}
    We start by considering an arbitrary globally orthogonal product set $\cS \coloneq \{\ket{a_{k}}^{A}\ket{b_{k}}\}$. We define the set of Alice's (resp.~Bob's) possible states as $S_{A} \coloneq \{\ket{a_{k}}\}$ (resp.~$S_{B} \coloneq \{\ket{b_{k}}\}$. If Alice's possible states do not contain two mutually orthogonal states, then Bob has to discriminate all the states. This means there are at most $d$ states on Bob's side and they are mutually orthogonal, and thus an LOSCC strategy exists (Bob just discriminates them locally). Thus, if there is a separation between LOSQC and LOSCC in $\mbb{C}^{2} \otimes \mbb{C}^{d}$, then Alice's possible states must contain two mutually orthogonal states.

    We now consider the case where Alice has more than two states. Consider Alice's set of states contains two mutually orthogonal states, which up to local unitary are $\ket{0},\ket{1}$, and some third state $\ket{\varphi} \not \in \{\ket{0},\ket{1}\}$. Then consider $S_{B}^{0} := \{\ket{b_{k}} : \ket{a_{k}} = \ket{0}\}$ and $S_{B}^{1}$, $S_{B}^{\varphi}$ defined similarly. As $\ket{\varphi}$ overlaps with both $\ket{0}$ and $\ket{1}$, the states in $S_{B}^{\varphi}$ must be each orthogonal to all states in $S_{B}^{0}$ and all states in $S_{B}^{1}$ for the set of states to be globally orthogonal. Thus, $\text{span}(S_{B}^{\varphi})$ defines a subspace that is orthogonal to $\text{span}(S_{B}^{0} \cup S_{B}^{1})$. Furthermore, if $\ket{\varphi^{\perp}}$ such that $\bra{\varphi^{\perp}}\ket{\varphi} = 0$, then $\ket{\varphi^{\perp}} \not \in \{\ket{0},\ket{1}\}$ but overlaps with both states, because we have a qubit space on Alice's side. Thus, $\text{span}(S^{\varphi^{\perp}}_{B})$ is also a subspace that is orthogonal to $\text{span}(S_{B}^{0} \cup S_{B}^{1})$. Thus, Bob can apply a projector to separate $S_{B}^{0} \cup S_{B}^{1}$ from $S_{B}^{\varphi} \cup S_{B}^{\varphi^{\perp}}$ prior to communication, where we may take $S_{B}^{\varphi^{\perp}}$ to be trivial if $\ket{\varphi^{\perp}}$ is not relevant. Applying this argument iteratively, we may conclude if Alice has pairs of orthogonal states, $\{\ket{0},\ket{1}\}$, $\{\ket{\varphi},\ket{\varphi^{\perp}}\}$, etc., each pair results in its own subspace on Bob's side that Bob can determine prior to communication. Note that if $\linspan(S^{\varphi}_{B}) \cap \linspan(S^{\varphi^{\perp}}_{B}) = \emptyset$, then Bob can determine the state locally for states $\ket{a_{k}} \in \{\ket{\varphi},\ket{\varphi^{\perp}}\}$. Thus we can remove any such pair. In total, we have reduced the existence of a separation in this setting to the existence of a separating ensemble of the form where Alice has sets of orthogonal pairs, $\{\ket{0},\ket{1}\}$, $\{\ket{\varphi},\ket{\varphi^{\perp}}\}$, etc., where for each orthogonal pair on Alice's side, there are states on Bob's corresponding set that overlap, e.g. there exists $\ket{b} \in S^{0}_{B}, \ket{b'} \in S^{1}_{B}, \ket{b''} \in S^{\varphi}_{b}, \ket{b'''} \in S^{\varphi^{\perp}}_{B}$ such that $\bra{b}\ket{b'} \neq 0$ and $\bra{b''}\ket{b'''} \neq 0$. This implies that Alice must broadcast the orthogonality of each pair $\{\ket{0},\ket{1}\}$, $\{\ket{\varphi},\ket{\varphi^{\perp}}\}$, etc. However, the uncertainty relation (Lemma \ref{lem:main-text-UR}) shows that Alice cannot broadcast the orthogonality of more than one such pair, so there cannot be more than one pair.
    
    Combining everything we have shown, if there is a separating ensemble $\cS$ in $\mbb{C}^{2} \otimes \mbb{C}^{d}$, then up to local unitaries, $\ket{a_{k}} \in \{\ket{0},\ket{1}\}$ for all $k$. We then may define $S_{B}^{0} := \{\ket{b_{k}} : \ket{a_{k}} = \ket{0}\}$ and similarly for $S_{B}^{1}$. This is the structure of the ensemble in \eqref{eq:separating-form}. This completes the proof.
\end{proof}

\begin{lemma}\label{lem:C2-Cd-LOSQC-to-LOSCC-reduction}
    If a globally orthogonal product set of the form 
    \begin{align}\label{eq:LOSQC-to-LOSCC-reduc-ensemble}
        \{\ket{0}^{A}\ket{b^{0}_{i}}\}_{i \in \cK^{0}} \cup \{\ket{1}^{A}\ket{b^{1}_{j}}\}_{j \in \cK^{1}} \ ,
    \end{align}
    is perfectly LOSQC discriminable, then it is also perfectly LOSCC discriminable.
\end{lemma}
\begin{proof}
    Assume that the ensemble is perfectly discriminable under LOSQC. As it is perfectly discriminable under LOSQC, it does not matter what the distribution over the set of states is, so we may assume that it is equiprobable. Moreover, as Alice's input is classical, without loss of generality, she just copies it. Thus, without loss of generality, is that Bob applies a broadcast channel $\cE^{B \to A'B'}$ sending the $A'$ space to Alice and keeping the $B'$ space for himself and then both Alice and Bob receive the value $0$ or $1$ that Alice copied. In other words, we have reduced any LOSQC strategy for \eqref{eq:LOSQC-to-LOSCC-reduc-ensemble} to an orthogonal broadcasting strategy for $\cS_{\Theta} = \{\rho_{i\vert 0} \coloneq \ket{b^{0}_{i}}\}_{i \in \cK^{0}} \cup \{\rho_{j \vert 1} \coloneq \ket{b^{1}_{j}}\}_{j \in \cK^{1}}$ (Recall Fig.~\ref{Fig:orthogonality_broadcasting} and the definition of orthogonality broadcasting). Then, as $|\Theta| = 2$ and we may assume the inputs are equiprobable, by Propositions \ref{Prop:classical-broadcasting-equal-2} and \ref{Prop:classical-orthogonal-broadcasting}, we know there exists a fully classical broadcasting map $\cE^{B \to XY}$ that achieves the same winning probability. As our assumption is the winning probability is unity and $\cE^{B \to XY}$ defines a measurement, there is a perfect LOSCC strategy. 
\end{proof}

\begin{proof}[Proof of Theorem \ref{thm:C2-Cd-LOSCC}]
    By Lemma \ref{lem:dim-reduc-part-1}, if there was an LOSQC but not LOSCC discriminable ensemble, it would be of the form given in \eqref{eq:separating-form}. By Lemma \ref{lem:C2-Cd-LOSQC-to-LOSCC-reduction}, if an ensemble of this form is LOSQC discriminable, then it is LOSCC discriminable. Thus, there cannot exist an LOSQC but not LOSCC discriminable ensemble in $\mbb{C}^{2} \otimes \mbb{C}^{d}$. Finally, to get the form in the theorem, one may add back any states that are locally determinable on Bob's side.
\end{proof}

\section{LOSQC Error Bounds}\label{app:LOSQC-error-bounds}
To establish Theorem \ref{thm:discrimination-cond}, we need the following lemma.
\begin{lemma}
\label{Lem:prod-norm}
If $0\leq W,X,Y,Z\leq\mbbm{1}$, then
\begin{align}
\Vert W\otimes X- Y\otimes Z\Vert_1\leq \Vert W-Y\Vert_1+\Vert X-Z\Vert_1.
\end{align}
\begin{proof}
Let $-\mbbm{1}\leq \tau\leq \mbbm{1}$ be such that
\begin{align*}
& \Vert W\otimes X- Y\otimes Z\Vert_1 \\
=&\Tr[\tau\left( W\otimes X- Y\otimes Z\right)]\\
=&\Tr\Big[\tau\Big(W\otimes\frac{X+Z}{2}-Y\otimes\frac{X+Z}{2} + \frac{W+Y}{2}\otimes X - \frac{W+Y}{2}\otimes Z\Big)\Big] \\
\leq& \Vert W-Y\Vert_1+\Vert X-Z\Vert_1,
\end{align*}
since $-\mbbm{1}\leq \Tr_B\left[\tau^{AB}\left(\mbbm{1}\otimes\frac{X+Z}{2}\right)\right]\leq \mbbm{1}$ and $-\mbbm{1}\leq \Tr_A\left[\tau^{AB}\left(\frac{W+Y}{2}\otimes\mbbm{1}\right)\right]\leq \mbbm{1}$.
\end{proof}
\end{lemma}

\begin{proof}[Proof of Theorem \ref{thm:discrimination-cond}]
 We assume that Alice and Bob apply local isometries $U^{A\to AB'}$ and $V^{B \to A'B}$, respectively, on their given states. We define $\ket{\alpha_{k}} := U\ket{a_{k}}$ and $\ket{\beta_{k}} := V\ket{b_{k}}$. The simultaneous communication occurs and Alice holds systems $AA'$ while Bob holds systems $BB'$.  The four possible states after the isometries have the form
\begin{align*}
    U\otimes V\ket{\psi_0}^{AB}&=\ket{\alpha_0}^{AB'}\ket{\beta_0}^{A'B}\notag\\
     U\otimes V\ket{\psi_1}^{AB}&=\ket{\alpha_1}^{AB'}\ket{\beta_1}^{A'B}\notag\\
     U\otimes V\ket{\psi_\theta}^{AB}&=\ket{\alpha_2}^{AB'}\ket{\beta_\theta}^{A'B}\notag\\
     U\otimes V\ket{\psi_\omega}^{AB}&=\ket{\alpha_3}^{AB'}\ket{\beta_\omega}^{A'B},
\end{align*}
where $\ket{\beta_{\theta}}^{A'B}=\cos(\theta/2)\ket{\beta_0}+\sin(\theta/2) e^{i\phi}\ket{\beta_1}$ and $\ket{\beta_{\omega}}^{A'B}=\cos(\omega/2)\ket{\beta_0} + \cos(\omega/2)e^{i\phi'}\ket{\beta_1}$.

After the communication, POVMs $\{P_k\}^{AA'}$ and $\{Q_k\}^{BB'}$ are performed by Alice and Bob, respectively.  The overall (unnormalized) success probability is given by
\begin{align*}
    P_S:=& \sum_{k}\Tr\Big[\Big(P_k^{AA'}\otimes Q_k^{BB'})(U^{AB'}\otimes V^{BA'}) \psi_k^{AB}(U^{AB'}\otimes V^{BA'}\Big)^\dagger\Big].
\end{align*}
The completion relation demands that $\sum_k P_k=\mbbm{1}^{AA'}$ and $\sum_k Q_k=\mbbm{1}^{BB'}$. Suppose that 
\begin{align*}
    \Tr\left[(P_k^{AA'}\otimes Q_k^{BB'})\left(\alpha_k^{AB'}\otimes\beta_k^{A'B}\right)\right] \geq 1-\epsilon\qquad\forall k.
\end{align*}
From this we obtain constraints on Alice and Bob's POVM elements:
\begin{align}
    \Tr[P_k(\alpha_k^A\otimes\beta_k^{A'})]&\geq 1-\epsilon\label{Eq:outcome-1a}\\
    \Tr[Q_k(\alpha_k^{B'}\otimes\beta_k^{B})]& \geq 1-\epsilon.\label{Eq:outcome-1b}
\end{align}
Since $P_k\leq \mbbm{1}-P_j$ and $Q_k\leq \mbbm{1}-Q_{j}$ for all $j\not=k$, we use the previous equations to obtain
\begin{align}
        -\Tr[P_{k}(\alpha_j^A\otimes\beta_j^{A'})]>-\epsilon\label{Eq:outcome-2a}\\
    -\Tr[Q_{k}(\alpha_j^{B'}\otimes\beta_j^{B})]>-\epsilon.\label{Eq:outcome-2b}
\end{align}
 Adding Eqns. \eqref{Eq:outcome-2a}--\eqref{Eq:outcome-2b} to Eqns. \eqref{Eq:outcome-1a}--\eqref{Eq:outcome-1b} yields
\begin{align}
    1-2\epsilon&<\Tr\left[P_k\left(\alpha_k^A\otimes\beta_k^{A'}-\alpha_j^A\otimes\beta_j^{A'}\right)\right] \notag \\
    &<\frac{1}{2}\left\Vert \alpha_k^A\otimes\beta_k^{A'}-\alpha_j^A\otimes\beta_j^{A'}\right\Vert_1 \notag \\
    & \leq\sqrt{1-F(\alpha_k^{A},\alpha_j^{A})^2F(\beta_k^{A'},\beta_j^{A'})^2}\label{Eq:split-bound-1}
\end{align}
\begin{align}
     1-2\epsilon&<\Tr\left[Q_k\left(\alpha_k^{B'}\otimes\beta_k^{B}-\alpha_j^{B'}\otimes\beta_j^{B}\right)\right] \notag \\
     &<\frac{1}{2}\left\Vert \alpha_k^{B'}\otimes\beta_k^{B}-\alpha_j^{B'}\otimes\beta_j^{B}\right\Vert_1 \notag \\
     & \leq\sqrt{1-F(\alpha_k^{B'},\alpha_j^{B'})^2F(\beta_k^{B},\beta_j^{B})^2}.\label{Eq:split-bound-2}
\end{align}
This says that the isometries $U$ and $V$ must split the states $\ket{a_k}\ket{b_k}$ and $\ket{a_j}\ket{b_j}$ into parts that are (roughly) mutually orthogonal for both parties, for all pairs $j\not=k$.

Applying Eq.  \eqref{Eq:split-bound-1} on the first two states, $\ket{\alpha_{0}}\ket{\beta_{0}}, \ket{\alpha_{1}}\ket{\beta_{1}}$, yields
\begin{align*}
    1-2\epsilon&\leq\sqrt{1-F(\alpha_0^{A},\alpha_1^{A})^2F(\beta_0^{A'},\beta_1^{A'})^2}\\ 
    &\leq \sqrt{1-|\ip{\alpha_0}{\alpha_1}|^2F(\beta_0^{A'},\beta_1^{A'})^2}\\
    &=\sqrt{1-|\ip{a_0}{a_1}|^2F(\beta_0^{A'},\beta_1^{A'})^2} \ , 
\end{align*}
which under re-ordering means,
\begin{align}
   \Rightarrow\qquad F(\beta_0^{A'},\beta_1^{A'})&\leq \frac{2\sqrt{\epsilon(1-\epsilon)}}{|\ip{a_0}{a_1}|^2} =: \delta .
\end{align}
Applying \eqref{eq:TD-F-uncertainty-relation} of Lemma \ref{lem:main-text-UR} multiplied by two,
    \begin{align}
        & \frac{4|x|\sqrt{\ve(1-\ve)}}{|\langle a_{0} | a_{1} \rangle|^{2}} \notag \\
        >& \|\beta_{\theta}^{B} - \beta_{\omega}^{B} \|_{1} - |w|\|\beta_{0}^{B}-\beta_{1}^{B}\|_{1} \notag \\
        \geq& \|\alpha_{2}^{B'}\otimes \beta_{\theta}^{B} - \alpha_{3}^{B'}\otimes \beta_{\omega}^{B} \|_{1} - \|\alpha_{2}^{B'} - \alpha_{3}^{B'} \|_{1} - |w|\|\beta_{0}^{B} - \beta_{1}^{B}\|_{1} \label{eq:unif-error-proof-int-bound} \ ,
    \end{align}
    where the second inequality is by Lemma \ref{Lem:prod-norm}, $w = \frac{1}{2}(\cos(\theta)-\cos(\omega))$, and $x = \frac{1}{2}(\sin(\theta)e^{-i\phi}-\sin(\omega)e^{-i\phi'})$. Note that $\langle a_{2} | a_{3} \rangle = \langle \alpha_{2} | \alpha_{3} \rangle$ and $\|\alpha_{2}^{B'} - \alpha_{3}^{B'}\|_{1} \leq 2\sqrt{1-|\langle \alpha_{2} | \alpha_{3} \rangle |^{2}}$ by Fuchs-van de Graaf inequality, so $\|\alpha_{2}^{B'} - \alpha_{3}^{B'}\|_{1} \leq \sqrt{1-|\langle a_{2} | a_{3} \rangle|^{2}}$ and similarly $\|\beta_{0}^{B} - \beta_{1}^{B}\|_{1} \leq 2\sqrt{1- |\bra{b_{0}}\ket{b_{1}}|^{2}}$. Then the inequality given in \eqref{eq:unif-error-proof-int-bound} can be relaxed to
    \begin{align*}
         & \frac{4|x|\sqrt{\ve(1-\ve)}}{|\langle a_{0} | a_{1} \rangle|^{2}} + 2\left[\sqrt{1-|\langle a_{2} | a_{3} \rangle|^{2}} + |w|\sqrt{1- |\bra{b_{0}}\ket{b_{1}}|^{2}} \right]  \geq  \|\alpha_{2}^{B'}\otimes \beta_{+}^{B} - \alpha_{3}^{B'}\otimes \beta_{-}^{B} \|_{1} .
    \end{align*}
    We require $1-2\ve < \frac{1}{2}\|\alpha_{2}^{B'}\otimes \beta_{+}^{B} - \alpha_{3}^{B'}\otimes \beta_{-}^{B} \|_{1}$. So in total, we need
    \begin{align*}
       \frac{2|x|\sqrt{\ve(1-\ve)}}{|\langle a_{0} | a_{1} \rangle|^{2}} + \sqrt{1-|\langle a_{2} | a_{3} \rangle|^{2}} + |w|\sqrt{1-|\langle b_{0} | b_{1} \rangle|^{2}}
        > 1 - 2\ve .
    \end{align*}
    \end{proof}

\subsection{Semidefinite Program LOSQC Error Bounds}
The rest of the proofs follow roughly the same idea and tools. We therefore first summarize the proof idea and collect the relevant tools. First, as noted in the main text, one can always upper bound LOSQC success probability by what may be viewed as a generalization of the approximate orthogonality broadcasting measure in the main text, \eqref{Eq:broadcasting-defn}. To see this, consider Alice and Bob are supplied states $\cS = \{\ket{\psi_{k}} \coloneq \ket{a_{k}}^{A}\ket{b_{k}}^{B}\}_{k \in \cK}$ according to some distribution $p$ over $\cK$. For any fixed choices of broadcasting isometries\footnote{The choice of isometries is without loss of generality by considering the isometric extension of the broadcasting channels.} $U,V$, we define $\ket{\alpha_{k}}^{AB'} \coloneq U\ket{a_{k}}$, $\ket{\beta_{k}}^{A'B} \coloneq V\ket{b_{k}}$. Then for a fixed choice of $U,V$, the optimal strategy is
\begin{align}
    \Pr_{\text{LOSQC}}[\cS_{\Theta} \vert U,V]
        =& \max_{ \{\pi^{AA'}_{i}\},\{\tau^{BB'}_{i}\}} \sum_{i} p(i) \Tr[\pi^{\theta}_{i} \otimes \tau^{\theta}_{i}U\otimes V \psi_{i} (U \otimes V)^{\dagger}] \nonumber \\
        \leq & \min_{P \in \{A,B\}} \max_{\{\gamma_{i}^{PP'}\}} \sum_{i} p(i) \Tr[\gamma^{PP'}_{i} \alpha_{i} \otimes \beta_{i} ] \ , \label{eq:local-distinguisihing}
\end{align}
where each maximization is over choices of POVM and the inequality uses that the probability both parties is correct is limited by the worse of the two parties. 

Next, when states are sufficiently equiprobable for it to apply, one may upper bound $\sum_{i} \frac{1}{2} \Tr[\gamma^{PP'}_{i} \alpha_{i} \otimes \beta_{i} ]$ in terms of distinguishing states pairwise as is necessary to apply the uncertainty relation. For example,
\begin{align}
    & \sum_{i \in \{1,2,3,4\}} \frac{1}{2}\Tr[\gamma^{PP'}_{i} \alpha_{i} \otimes \beta_{i}] \nonumber \\
    \leq& \frac{1}{2}\Big\{\Tr[\gamma^{PP'}_{1}\alpha_{1} \otimes \beta_{1}] + \Tr[( \mbbm{1}-\gamma^{PP'}_{1})\alpha_{2} \otimes \beta_{2}] \Big\} \nonumber
    \\
    & + \frac{1}{2}\Big\{\Tr[\gamma^{PP'}_{3}\alpha_{3} \otimes \beta_{3}] + \Tr[( \mbbm{1}-\gamma^{PP'}_{4})\alpha_{4} \otimes \beta_{4}] \Big\} \nonumber \\
    \leq& \max_{0 \leq \pi \leq 1} \frac{1}{2}\Big\{\Tr[\pi \alpha_{1} \otimes \beta_{1}] + \Tr[( \mbbm{1}-\pi)\alpha_{2} \otimes \beta_{2}] \Big\} \nonumber \\
    & + \max_{0 \leq \tau \leq 1} \frac{1}{2}\Big\{\Tr[\tau \alpha_{3} \otimes \beta_{3}] + \Tr[( \mbbm{1}-\tau)\alpha_{4} \otimes \beta_{4}] \Big\} \nonumber \\
    =& 1/2(1+D_{\tr}(\alpha_1 \otimes \beta_1, \alpha_2 \otimes \alpha_2)) + 1/2(1+D_{\tr}(\alpha_3 \otimes \beta_3, \alpha_4 \otimes \alpha_4)) \nonumber \\
    =& \frac{1}{2} + \frac{1}{2}\Big[D_{\tr}(\alpha_1 \otimes \beta_1, \alpha_2 \otimes \alpha_2) + D_{\tr}(\alpha_3 \otimes \beta_3, \alpha_4 \otimes \alpha_4)\Big] \ , \label{eq:pairwise-decomposition}
\end{align}
where the equality is by the Holevo-Helstrom theorem and we have omitted the $P,P'$ superscripts for simplicity.

One may then split over the tensor product in the trace distance and convert to the guessing probability, i.e.
\begin{equation}
\begin{aligned} 
D_{\tr}(\alpha_{\theta} \otimes \beta_{\theta},\alpha_{\omega} \otimes \beta_{\omega}) 
\leq& D_{\Tr}(\alpha_{\theta},\alpha_{\omega}) + D_{\Tr}(\beta_{\theta},\beta_{\omega}) \\
=& 2[p_{g}(\alpha_{\theta},\alpha_{\omega}) + p_{g}(\beta_{\theta},\beta_{\omega}) - 1] \label{eq:prod-splitting}
\end{aligned}
\end{equation}
where the first follows from dividing Lemma \ref{Lem:prod-norm} by two and the second is by the Holevo-Helstrom theorem.

At this point the remaining idea is to treat the guessing probabilities as optimization variables. One may place bounds on these guessing probabilities in terms of the uncertainty relation for guessing probability, which are constraints that are independent of the choice of $U,V$. Thus, one reduces the bounds to optimizing over probabilities. This will then be a semidefinite program over these probabilities in disciplined convex programming (DCP) form \cite{Grant-2006a} and then may be solved using CVXPY \cite{Diamond-2016a,Agrawal-2018a}.

\begin{proof}[Proof of Proposition \ref{prop:SDP-version-of-Thm-4}]
Following the proof of Theorem \ref{thm:discrimination-cond}, without loss of generality,
    \begin{align*}
    U\otimes V\ket{\psi_0}^{AB}&=\ket{\alpha_0}^{AB'}\ket{\beta_0}^{A'B}\notag\\
     U\otimes V\ket{\psi_1}^{AB}&=\ket{\alpha_1}^{AB'}\ket{\beta_1}^{A'B}\notag\\
     U\otimes V\ket{\psi_\theta}^{AB}&=\ket{\alpha_2}^{AB'}\ket{\beta_\theta}^{A'B}\notag\\
     U\otimes V\ket{\psi_\omega}^{AB}&=\ket{\alpha_3}^{AB'}\ket{\beta_\omega}^{A'B},
\end{align*}
where $\ket{\beta_{\theta}}^{A'B}=\cos(\theta/2)\ket{\beta_0}+\sin(\theta/2) e^{i\phi}\ket{\beta_1}$ and $\ket{\beta_{\omega}}^{A'B}=\cos(\omega/2)\ket{\beta_0} + \cos(\omega/2)e^{i\phi'}\ket{\beta_1}$.

Then, for some a given choice of $U,V$, starting from \eqref{eq:local-distinguisihing},
\begin{align}
    & \min_{P \in \{A,B\}} \max_{\{\gamma_{i}^{PP'}\}} \sum_{i} p(i) \Tr[\gamma^{PP'}_{i} \alpha_{i} \otimes \beta_{i} ] \nonumber \\
    \leq& \frac{1}{2} + \min_{P \in \{A,B\}} \frac{p}{2} D_{\tr}(\alpha_{0} \otimes \beta_{0},\alpha_{1} \otimes \beta_{1}) + \frac{1-p}{2} D_{\tr}(\alpha_{\theta} \otimes \beta_{\theta},\alpha_{\omega} \otimes \beta_{\omega}) \ , \label{eq:LOSQC-SDP-bound-1}
\end{align}
where the inequality uses that $p(0) = p(1) = p/2$, $p(\theta) = p(\omega) = (1-p)/2$ and a similar argument to that in \eqref{eq:pairwise-decomposition}. Applying \eqref{eq:prod-splitting} and simplifying we obtain
\begin{align*}
    & \Pr_{\text{LOSQC}}[\cS_{\Theta} \vert U,V] \\
    \leq& \min_{P \in \{A,B\}} p \cdot [p_{g}(\alpha_{0},\alpha_{1}) + p_{g}(\beta_{0},\beta_{1})] + (1-p) [p_{g}(\alpha_{\theta},\alpha_{\omega}) + p_{g}(\beta_{\theta},\beta_{\omega})] - 1 \ . 
\end{align*}
Thus, we have
\begin{align*}
    \Pr_{\text{LOSQC}}[\cS_{\Theta}]
    \leq& \max_{U,V} \min_{P \in \{A,B\}} p \cdot [p_{g}(\alpha_{0},\alpha_{1}) + p_{g}(\beta_{0},\beta_{1})] \\
    & \hspace{2.3cm} + (1-p) [p_{g}(\alpha_{\theta},\alpha_{\omega}) + p_{g}(\beta_{\theta},\beta_{\omega})] - 1 \ . 
\end{align*}
By data-processing, $p_{g}(\alpha_{0},\alpha_{1}) \leq p_{g}(\psi_0, \psi_{1})$, $p_{g}(\alpha_{\theta},\alpha_{\omega}) \leq p_{g}(\psi_\theta , \psi_{\omega})$, which are determinable constants using Holevo-Helstrom. Moreover, we may apply Corollary \ref{cor:guessing-probability-UR} to $\beta_\theta, \beta_\omega$ in terms of $\beta_0 , \beta_1$ for any choice of $U,V$. Therefore, defining $r_0 \coloneq p_{g}(\beta_{0}^{A'},\beta_{1}^{A'})$, $r_1 \coloneq p_{g}(\beta_{0}^{B},\beta_{1}^{B})$ and similarly for $s_0, s_1$ in terms of Bob completes the derivation of the optimization program. 

To see the optimization problem is convex, note that the constraints are concave in the optimization variables and the objective function includes the pointwise minimization of linear functions, i.e. it is a concave objective function. As the problem is a maximization, we may conclude it is a convex optimization program. 

To see the unity conditions, note that $f(p,\cS,\mbf{x}) = 1$ if and only if $p_{g}(a_{0},a_{1}) = p_{g}(a_3 , a_4) = x_0 = x_1 = 1$. This proves the orthogonality conditions and proves $s_0, s_1, r_0, r_1$ must all be unity. Looking at the constraints, under everything being unity, one obtains $1 \leq \frac{1}{2}[\vert z_2 \vert + 1] \Leftrightarrow 1 = \vert z_2 \vert = \vert \cos(\theta) - \cos(\omega) \vert$, which is equivalent to the stated condition.
\end{proof}

\begin{proof}[Proof of Theorem \ref{thm:Overlapping-BB84}]
     As usual, we let $U \otimes V\ket{\psi_{i}}^{AB} = \ket{\alpha_{i}}^{AB'}\ket{\beta_{i}}^{A'B}$ for each $i$. We do not make use of the fact that one register is classical for the majority of the proof as it will simplify the presentation of the subsequent proof.
    
    Starting from \eqref{eq:local-distinguisihing},
    \begin{align*}
        & \min_{P \in \{A,B\}} \max_{\{\gamma_{i}^{PP'}\}} \sum_{i} p(i) \Tr[\gamma^{PP'}_{i} \alpha_{i} \otimes \beta_{i} ] \nonumber \\
        =& \min_{P \in \{A,B\}} \max_{\{\gamma_{i}^{PP'}\}} \Big[\frac{1}{8}\sum_{i \in \{1,2,3,4\}} \Tr[\gamma^{PP'}_{i} \alpha_{i} \otimes \beta_{i}] + \frac{1}{8}\sum_{i \in \{1,5,6,7\}} \Tr[\gamma^{PP'}_{i} \alpha_{i} \otimes \beta_{i}] \Big] \\
        =&  \frac{1}{4} \min_{P \in \{A,B\}} \max_{\{\gamma_{i}^{PP'}\}} \Big[\sum_{i \in \{1,2,3,4\}} \frac{1}{2}\Tr[\gamma^{PP'}_{i} \alpha_{i} \otimes \beta_{i}] + \sum_{i \in \{1,5,6,7\}} \frac{1}{2}\Tr[\gamma^{PP'}_{i} \alpha_{i} \otimes \beta_{i}] \Big] \ ,
    \end{align*}
    where we have used our assumption on the probabilities so that we may split distinguishing $\psi_{1}$ into each sum.
    
    Using the bound of the form \eqref{eq:pairwise-decomposition} in both cases, we obtain
    \begin{align}
        & \Pr_{\text{LOSQC}}[\cS_{\Theta} \vert U,V] \nonumber \\
        \leq& \frac{1}{4} + \frac{1}{8} \min_{P \in \{A,B\}} \Big[\sum_{i \in \{1,3,6\}} D_{\tr}(\alpha_{i} \otimes \beta_{i}, \alpha_{i+1} \otimes \beta_{i+1}) + D_{\tr}(\alpha_{1} \otimes \beta_{1}, \alpha_{5} \otimes \beta_{5}) \Big] \ ,
    \end{align}
    where we have distributed the $1/4$.

    Next, introducing $p_{g}^{P}$ to denote the guessing probability of party $P$ with access to their portions of the state, by \eqref{eq:prod-splitting},
    \begin{align*} 
        D_{\tr}(\alpha_{\theta} \otimes \beta_{\theta},\alpha_{\omega} \otimes \beta_{\omega})
        \leq& 2[p_{g}^{P}(\alpha_{\theta},\alpha_{\omega}) + p_{g}^{P}(\beta_{\theta},\beta_{\omega}) - 1] \ ,
    \end{align*}
    so we obtain
    \begin{align}
        \Pr_{\text{LOSQC}}[\cS_{\Theta} \vert U,V] 
        \leq& \frac{1}{4} + \frac{1}{4}\min_{P \in \{A,B\}} \Big[\sum_{i \in \{1,3,6\}} \{p_{g}^{P}(\alpha_{i},\alpha_{i+1}) + p_{g}^{P}(\beta_{i},\beta_{i+1})\} \nonumber \\
        &\hspace{2.5cm} + p_{g}^{P}(\alpha_{1},\alpha_{5}) + p_{g}^{P}(\beta_{1}, \beta_{5}) - 4\Big] \ . \label{eq:Bennett-1-c-LOSQC-bound-1}
    \end{align}

    At this point we note that, by construction, the $\alpha_{i}$ term is always the same as the one it is being compared to, so $p_{g}^{P}(\alpha_{i},\alpha_{i+1}) = 1/2 = p_{g}^{P}(\alpha_{1},\alpha_{5})$ for $i \in \{1,3,6\}$ a term is always $1/2$. This simplifies the objective function to
    \begin{align}
        \Pr_{\text{LOSQC}}[\cS_{\Theta} \vert U,V]
        \leq& \frac{1}{4} + \frac{1}{4}\min_{P \in \{A,B\}} \Big[\sum_{i \in \{1,3,6\}} \{p_{g}^{P}(\beta_{i},\beta_{i+1})\} + p_{g}^{P}(\beta_{1},\beta_{5}) - 2\Big] \ . \label{eq:Bennett-1-LOSQC-c-bound-2}
    \end{align}
    Then to bound the optimal strategy we would optimize over the choice of $U,V$. However the bound is now independent of $U$ and we may replace maximizing over $V$ via constraints that hold from Corollary \ref{cor:guessing-probability-UR}, which apply for any choice of $V$. We define $p_{0} \coloneq p_{g}^{P}(V\ket{1},V\ket{2})$, $p_{1} \coloneq p_{g}^{P}(V\ket{1+2},V\ket{1-2})$, $p_{2} \coloneq p_{g}^{P}(V\ket{1},V\ket{0})$, $p_{3} \coloneq p_{g}^{P}(V\ket{0+1},V\ket{0-1})$ for $P \in \{A,B\}$. Applying Lemma \ref{lem:main-text-UR} to \eqref{eq:Bennett-1-LOSQC-c-bound-2}, we obtain $4(\Pr_{LOSQC}[\cS] - 1/4)$ is upper bounded by
    \begin{equation}
        \begin{aligned}
            \max \; \; & \min \Big\{\sum_{i \in [3]} a_i , \sum_{i \in [3]} b_i \Big\} -2 \\
            \text{s.t.} \; \; & p_0 \leq \frac{1}{2}\sqrt{1-(2\ol{p}_{1}-1)^{2}} + \frac{1}{2} \quad p \in \{a,b\} \\
            & p_1 \leq \frac{1}{2}\sqrt{1-(2\ol{p}_{0}-1)^{2}} + \frac{1}{2} \quad p \in \{a,b\} \\
            & p_2 \leq \frac{1}{2}\sqrt{1-(2\ol{p}_{3}-1)^{2}} + \frac{1}{2} \quad p \in \{a,b\} \\
            & p_3 \leq \frac{1}{2}\sqrt{1-(2\ol{p}_{2}-1)^{2}} + \frac{1}{2} \quad p \in \{a,b\} \\
            & 0 \leq \mbf{a}, \mbf{b} \leq 1 \ ,
        \end{aligned}
    \end{equation}
    where $[3] = \{0,1,2,3\}$, $\ol{p}$ represents the opposite party, and in applying Lemma \ref{lem:main-text-UR} we have used $\ket{i \pm j} = \frac{1}{\sqrt{2}}\ket{i} \pm \frac{1}{\sqrt{2}}\ket{j}$ and $\ket{i} = \frac{1}{\sqrt{2}}[\ket{i+j} + \ket{i-j}]$ so that $z_{1} = 1$ and $z_2 = 0$ always. Finally, we directly solve this SDP using CVXPY.
\end{proof}
 
As presented in the main text, our method is general enough to apply to the case where both Alice and Bob receive non-commuting states. Here we present a lemma that will both allow us to establish the separation in Theorem \ref{thm:qq-qc-separation} as well as explain why we cannot decrease the success probability of Theorem \ref{thm:Overlapping-BB84} by adding coherence to Alice's side using our methodology.
\begin{lemma}\label{lem:qq-bound}
    Consider the set of globally orthogonal, fully quantum, product qutrit states
    \begin{align}\label{eq:alt-CS-qq}
		\wt{\cS}_{qq} = \begin{Bmatrix}
			\ket{\psi_{1}} = \ket{1}\ket{1} \\ 
            \ket{\psi_{2}} = \ket{1-2}\ket{2}   &
			\ket{\psi_{3}} = \ket{0}\ket{1+2}   \\
            \ket{\psi_{4}} = \ket{0}\ket{1-2}  &
			\ket{\psi_{5}} = \ket{0+1}\ket{0}   \\
            \ket{\psi_{6}} = \ket{2}\ket{0+1}   & \ket{\psi_{7}} = \ket{2}\ket{0-1}  
		\end{Bmatrix} \ ,
\end{align}
where $\ket{\psi_{1}}$ occurs with probability $1/4$ and the rest occur with probability $1/8$. Then 
\begin{align}\label{eq:LOSQC-separation-strong}
    \Pr_{\text{LOSQC}}[\wt{\cS}_{qq}] \leq 0.78033 \ .
\end{align}
\end{lemma}
\begin{proof}
    First, note that for all $i$, $\ket{b_{i}}$ in $\cS_{qq}$ is the same as $\ket{b_{i}}$ in $\cS_{O\BB}$ as given in \eqref{eq:overlapping-BB84-states}. This means the majority of the proof is identical to the previous. In particular, we may start at \eqref{eq:Bennett-1-c-LOSQC-bound-1}. However, we now make choices on how to bound the probabilities on Alice's side that are no longer identical states. In particular, we may use data-processing to bound the guessing probability by the original state, e.g.
    \begin{align*}
        p_{g}^{P}(W\ket{1},W\ket{1-2}) \leq& p_{g}(\ket{1},\ket{1-2}) \\
        =& \frac{1}{2}(1+ \frac{1}{2}\Vert \dyad{1} - \dyad{1-2} \Vert_{1})
        = \frac{1}{2}(1+ \frac{1}{\sqrt{2}}) \eqqcolon t \ .
    \end{align*}
    Then we obtain
    \begin{align}
        & \sum_{i \in \{1,3,6\}} \{p_{g}^{P}(\alpha_{i},\alpha_{i+1}) + p_{g}^{P}(\beta_{i},\beta_{i+1})\} + p_{g}^{P}(\alpha_{1},\alpha_{5}) + p_{g}^{P}(\beta_{1}, \beta_{5}) - 4\Big] \nonumber \\
        =& \; t + p_{g}^{P}(V\ket{1},V\ket{2}) + \frac{1}{2} + p_{g}^{P}(V\ket{1+2},V\ket{1-2}) 
         \nonumber \\
        &\hspace{5mm} + \frac{1}{2} + p_{g}^{P}(V\ket{0+1},V\ket{0-1}) + t + p_{g}^{P}(V\ket{0},V\ket{1}) - 4  \nonumber \\
        =& \; p_{g}^{P}(V\ket{1},V\ket{2}) + p_{g}^{P}(V\ket{1+2},V\ket{1-2})
         \nonumber \\
        &\hspace{5mm} +  p_{g}^{P}(V\ket{0},V\ket{1}) + p_{g}^{P}(V\ket{0+1},V\ket{0-1}) + \frac{1}{\sqrt{2}} - 2 \ . \label{eq:Bennett-1-LOSQC-bound-2}
    \end{align}

    By the same argument as in the previous proof, this ultimately results in an SDP the upper bounds $4(\Pr_{\LOSQC}[\cS_{qq}]-1/4)$:
    \begin{equation}
        \begin{aligned}
            \max \; \; & \min \Big\{\sum_{i \in [3]} a_i , \sum_{i \in [3]} b_i \Big\} + \frac{1}{\sqrt{2}}-2 \\
            \text{s.t.} \; \; & p_0 \leq \frac{1}{2}\sqrt{1-(2\ol{p}_{1}-1)^{2}} + \frac{1}{2} \quad p \in \{a,b\} \\
            & p_1 \leq \frac{1}{2}\sqrt{1-(2\ol{p}_{0}-1)^{2}} + \frac{1}{2} \quad p \in \{a,b\} \\
            & p_2 \leq \frac{1}{2}\sqrt{1-(2\ol{p}_{3}-1)^{2}} + \frac{1}{2} \quad p \in \{a,b\} \\
            & p_3 \leq \frac{1}{2}\sqrt{1-(2\ol{p}_{2}-1)^{2}} + \frac{1}{2} \quad p \in \{a,b\} \\
            & 0 \leq \mbf{a}, \mbf{b} \leq 1 \ ,
        \end{aligned}
    \end{equation}
    where $\ol{p}$ represents the opposite party, and in applying Lemma \ref{lem:main-text-UR} we have used $\ket{i \pm j} = \frac{1}{\sqrt{2}}\ket{i} \pm \frac{1}{\sqrt{2}}\ket{j}$ and $\ket{i} = \frac{1}{\sqrt{2}}[\ket{i+j} + \ket{i-j}]$ so that $z_{1} = 1$ and $z_2 = 0$ always. Finally, this is a semidefinite program which we solved using CVXPY.
\end{proof}

We now use Lemma \ref{lem:qq-bound} to prove Theorem \ref{thm:qq-qc-separation}.
\begin{proof}[Proof of Theorem \ref{thm:qq-qc-separation}]
    Note that $\cS_{qq}$ in \eqref{eq:qq-states} is equivalent to $\wt{\cS}_{qq}$ in \eqref{eq:alt-CS-qq} up to applying the local unitary on Alice's side defined by
    $$ U\ket{0} = \ket{2} \quad U\ket{1}=\ket{1} \quad U\ket{2} = \ket{0} \ . $$
    Thus, we know $\Pr_{\LOSQC}[\cS_{qq}] \leq 0.7805$ by Lemma \ref{lem:qq-bound}. It therefore suffices to prove $\Pr_{\LOSQC}[\cS_{cq}] \geq \frac{1}{2}\left(1 + \frac{1}{\sqrt{2}}\right)$ for $\cS_{cq}$ given in \eqref{eq:cq-states}. We do this by construction of an LOSCC strategy achieving this lower bound.
    
    Consider the POVM $\{\op{\varphi_y}{\varphi_y}\}_{y=1}^4$ defined by vectors
\begin{align}
    \ket{\varphi_1}&=\frac{1}{\sqrt{2}}[\cos(\pi/8)\ket{1}+\sin(\pi/8)(\ket{0}+\ket{2})],\notag\\
    \ket{\varphi_2}&=\frac{1}{\sqrt{2}}[\cos(\pi/8)\ket{1}+\sin(\pi/8)(\ket{0}-\ket{2})],\notag\\
    \ket{\varphi_3}&=\frac{1}{\sqrt{2}}[\sin(\pi/8)\ket{1}-\cos(\pi/8)(\ket{0}+\ket{2})],\notag\\
    \ket{\varphi_4}&=\frac{1}{\sqrt{2}}[\sin(\pi/8)\ket{1}-\cos(\pi/8)(\ket{0}-\ket{2})].
\end{align}
Our LOSCC protocol involves Alice measuring in the computational basis (obtaining outcome $x$) and Bob performing the above POVM (obtaining otucome $y$).  Their guessing strategy is given by the following table:
\begin{center}
\begin{tabular}{ |c|c|c| }  \hline
 Outcomes ($x,y$) & Guess & Prob. the guess is correct \\ 
 \hline
  $\{(1,1), (1,2)\}$ & $\ket{\psi_1}$ & $\cos^2(\pi/8)$\\
 \hline
 $\{(1,3)\}$ & $\ket{\psi_{2}}$ & $\cos^{2}(\pi/8)$ \\ \hline
 $\{(1,4)\}$ & $\ket{\psi_{3}}$ & $\cos^{2}(\pi/8)$ \\ \hline
 $\{(0,1),(0,2)\}$ & $\ket{\psi_{4}}$ & $\frac{1}{2}(\cos(\pi/8)+\sin(\pi/8))^{2}$ \\ \hline
 $\{(0,3),(0,4)\}$ & $\ket{\psi_{5}}$ & $\frac{1}{2}(\cos(\pi/8)+\sin(\pi/8))^{2}$ \\ \hline
 $\{(2,1),(2,2)\}$ & $\ket{\psi_{6}}$ & $\frac{1}{2}(\cos(\pi/8)+\sin(\pi/8))^{2}$ \\ \hline
 $\{(2,3),(2,4)\}$ & $\ket{\psi_{7}}$ & $\frac{1}{2}(\cos(\pi/8)+\sin(\pi/8))^{2}$ \\ \hline
\end{tabular}
\end{center}
Noting that $\tfrac{1}{2}(\cos(\pi/8)+\sin(\pi/8))^2=\cos^2(\pi/8)$, using this strategy, regardless of the priors, the success probability is $\cos^{2}(\pi/8) = \frac{1}{2}(1 + \frac{1}{\sqrt{2}})$.  This proves that $P_{\text{LOSQC}}(\mc{S}_{cq})\geq P_{\text{LOSCC}}(\mc{S}_{cq})\geq \frac{1}{2}(1+\tfrac{1}{\sqrt{2}}))$. This completes the proof.
\end{proof}

Finally, we explain why Lemma \ref{lem:qq-bound} shows that our methodology is too loose to amplify Theorem \ref{thm:Overlapping-BB84} by making both parties have to do approximate broadcasting. First, note that one may obtain $\wt{\cS}_{qq}$ from $\cS_{OBB}$ in \eqref{eq:overlapping-BB84-states} from first applying the unitary 
$$ U\ket{0} = \ket{1} \quad U\ket{1}=\ket{0} \quad U\ket{2} = \ket{2} $$
to Alice's side and then altering $U\ket{a_{2}} \to \ket{1-2}$, $U\ket{a_{5}} \to \ket{0+1}$. This is to say, up to a local unitary, we may view $\wt{S}_{qq}$ as a method for making $\cS_{OBB}$ coherent on both sides. Then, we see the reason this does not tighten the result in this scenario is that our bound on $p^{P}_{g}(W\ket{1},W\ket{1-2}) > 1/2$ where the latter value is what we could use in Theorem \ref{thm:Overlapping-BB84}. It is not clear how to improve this using the uncertainty relation. In particular, if one applies Theorem \ref{thm:generalized-uncertainty-relation} to these states in terms of $W\ket{0},W\ket{1},W\ket{2}$, a direct calculation shows
\begin{align*}
    D_{\tr}((W\ket{0})^{B},(W\ket{1})^{B})
    \leq& \frac{1}{2} + 2F((W\ket{0})^{A},(W\ket{1})^{A}) \\
    \leq& \frac{1}{2} + 2\sqrt{1-(2p_{g}^{A}(W\ket{0},W\ket{1})-1)^{2}}
    \ ,
\end{align*}
where we used Fuchs-van de Graaf inequality and Holevo-Helstrom theorem. A direct calculation will verify this constraint is non-trivial (i.e. strictly less than unity) only when $p_{g}^{A}(W\ket{0},W\ket{1}) \geq 0.984123$, but the other constraints already guarantee that it is smaller than this. This seems to be a common issue in trying to apply Theorem \ref{thm:generalized-uncertainty-relation} more generally.
\end{document}